\documentclass[reqno,12pt]{amsart}
\usepackage{amsmath, latexsym, amsfonts, amssymb, amsthm, amscd}
\usepackage{graphics,epsf,psfrag,dsfont}
\advance\hoffset -.75cm

\numberwithin{equation}{section}

 \newtheorem{Theorem}{Theorem}[section]

\newtheorem{Proposition}[Theorem]{Proposition}

\newtheorem{remark}[Theorem]{Remark}

\newtheorem{theo}{Theorem}[section] 
\newtheorem{lem}[theo]{Lemma} 
 
\newtheorem{rem}[theo]{Remark} 
\newtheorem{pro}[theo]{Proposition} 
\newtheorem{definition}[theo]{Definition}

\makeatletter
\def\captionfont@{\footnotesize}
\def\captionheadfont@{\scshape}

\long\def\@makecaption#1#2{%
  \vspace{2mm}
  \setbox\@tempboxa\vbox{\color@setgroup
    \advance\hsize-6pc\noindent
    \captionfont@\captionheadfont@#1\@xp\@ifnotempty\@xp
        {\@cdr#2\@nil}{.\captionfont@\upshape\enspace#2}%
    \unskip\kern-6pc\par
    \global\setbox\@ne\lastbox\color@endgroup}%
  \ifhbox\@ne % the normal case
    \setbox\@ne\hbox{\unhbox\@ne\unskip\unskip\unpenalty\unkern}%
  \fi
  \ifdim\wd\@tempboxa=\z@ % this means caption will fit on one line
    \setbox\@ne\hbox to\columnwidth{\hss\kern-6pc\box\@ne\hss}%
  \else % tempboxa contained more than one line
    \setbox\@ne\vbox{\unvbox\@tempboxa\parskip\z@skip
        \noindent\unhbox\@ne\advance\hsize-6pc\par}%
\fi
  \ifnum\@tempcnta<64 % if the float IS a figure...
    \addvspace\abovecaptionskip
    \moveright 3pc\box\@ne
  \else % if the float IS NOT a figure...
    \moveright 3pc\box\@ne
    \nobreak
    \vskip\belowcaptionskip 
  \fi
\relax
}
\makeatother
\def\writefig#1 #2 #3 {\rlap{\kern #1 truecm
\raise #2 truecm \hbox{#3}}}

\DeclareMathSymbol{\leqslant}{\mathalpha}{AMSa}{"36} % nicer `smaller or equal'
\DeclareMathSymbol{\geqslant}{\mathalpha}{AMSa}{"3E} % nicer `larger or equal'
\DeclareMathSymbol{\eset}{\mathalpha}{AMSb}{"3F}     % nicer `emptyset'
\renewcommand{\leq}{\;\leqslant\;}                   % redef. of < or =
\renewcommand{\geq}{\;\geqslant\;}                   % redef. of > or =
\newcommand{\bra}{\langle}
\newcommand{\ket}{\rangle}

\newcommand{\cA}{\ensuremath{\mathcal A}}
\newcommand{\cB}{\ensuremath{\mathcal B}}
\newcommand{\cC}{\ensuremath{\mathcal C}}
\newcommand{\cD}{\ensuremath{\mathcal D}}
\newcommand{\cE}{\ensuremath{\mathcal E}}

\newcommand{\cG}{\ensuremath{\mathcal G}}
\newcommand{\cH}{\ensuremath{\mathcal H}}

\newcommand{\cL}{\ensuremath{\mathcal L}}
\newcommand{\cM}{\ensuremath{\mathcal M}}

\newcommand{\cO}{\ensuremath{\mathcal O}}

\newcommand{\cR}{\ensuremath{\mathcal R}}

\newcommand{\cV}{\ensuremath{\mathcal V}}
\newcommand{\cW}{\ensuremath{\mathcal W}}

\newcommand{\bbA}{{\ensuremath{\mathbb A}} }

\newcommand{\bbP}{{\ensuremath{\mathbb P}} }

\newcommand{\bbR}{{\ensuremath{\mathbb R}} }

\newcommand{\bbZ}{{\ensuremath{\mathbb Z}} }

\newcommand{\ga}{\alpha}
\newcommand{\gb}{\beta}
\newcommand{\gga}{\gamma}            % \gg already exists...
\newcommand{\gd}{\delta}
\newcommand{\gep}{\varepsilon}       % \ge already exists...
\newcommand{\gp}{\varphi}
\newcommand{\gr}{\rho}

\newcommand{\go}{\omega}
\newcommand{\gO}{\Omega}
\newcommand{\gl}{\lambda}
\newcommand{\gL}{\Lambda}
\newcommand{\gs}{\sigma}
\newcommand{\gS}{\Sigma}
\newcommand{\gt}{\vartheta}

\newcommand{\mm}{m}
\newcommand{\bk}[1]{\left\bra #1 \right\ket}
\newcommand{\bkt}[1]{\mu_{\ga,T} \left( #1 \right )}

\title[Activity phase transition]{Activity phase transition\\ for constrained dynamics}

\author{T. Bodineau, C. Toninelli}

\begin{document}

\maketitle

\begin{abstract}

We consider two cases of kinetically constrained models, namely East and FA-1f models. The object of interest of our work is the activity $\cA(t)$ defined as the total number of configuration changes in the interval $[0,t]$ for the dynamics on a finite domain.
It has been shown in \cite{GJLPDW1,GJLPDW2} that the large deviations of the activity exhibit a non-equilibirum phase transition in the thermodynamic limit and that reducing the activity is more likely than increasing it due to a blocking mechanism induced by the constraints. 
In this paper, we study the finite size effects around this first order phase transition and analyze the phase coexistence between the active and inactive dynamical phases in dimension 1. In higher dimensions, we show that the finite size effects  are also determined 
by the dimension and the choice of boundary conditions.

\end{abstract}



\section{Introduction}
\label{intro}

\medskip

Kinetically constrained spin models (KCSM) are
interacting particle systems which have been introduced and very much studied  in the physics
literature  to model liquid/glass transition and more
generally glassy dynamics (see \cite{RS,GST} and references therein).
A configuration is given by
assigning to each vertex $x$ of a (finite or infinite) connected graph
$\cG$ its occupation variable
$\eta_x \in\{0,1\}$ which corresponds to an empty or filled site,
respectively.  The evolution is given by  Markovian stochastic dynamics
of Glauber type.   Each site with rate one refreshes its occupation variable to
a filled or to an empty state with probability $\rho$ or $1-\rho$
respectively provided that the current configuration satisfies an
a priori specified local constraint.  
Here we focus on two of the most studied KCSM, the East  \cite{JE} and FA-1f models \cite{FA1,FA2} on hypercubic lattices
% we recall the one-dimensional ($\cG\subset \mathbb Z$) East model and FA-1f models  on hypercubic lattices
($\cG\subset \mathbb Z^d$): the constraint  at $x$ requires for East model  its right nearest neighbour to be empty, for FA-1f model at least one of its nearest neighbours to be empty.
Note that in both cases (and this is a general feature of KCSM) the  constraint which should be satisfied to allow creation/annihilation of a particle at $x$ does not
involve $\eta_x$, thus detailed balance w.r.t. the Bernoulli product
measure at density $\rho$ is an invariant reversible measure for the process. Both models are ergodic on $\cG=\mathbb Z^d$ for any $\rho\in(0,1)$ with a positive spectral gap  which shrinks to zero as $\rho\to 1$ corresponding to the occurrence of diverging mixing times \cite{CMRT}.

Several numerical works and approximated analytical treatments have  shown that relaxation for both models occurs in a more and more spatially heterogenous way as density is increased (see Section 1.5 of \cite{GST} and references therein). For example when measuring the persistence field $p_x(t)$ which equals to one if  site $x$ has never changed its state up to time $t$  and zero otherwise  a clear spatial segregation is observed among sites with 0/1 values of $p$ at time scales corresponding to the typical relaxation time of the persistence function which corresponds to the spatial average of the persistence field.
More quantitatively, if one measures the spatial correlation function of this persistence field, a dynamical correlation length corresponding to the extent of these heterogeneities can be extracted. This length increases as the density is increased.  The occurrence of these dynamical heterogeneites has  lead to  the idea that the dynamics of KCM
takes place on a first-order coexistence line between active and inactive dynamical phases \cite{MGC,JGC}. In
order to exploit this idea in \cite{GJLPDW1,GJLPDW2} the fluctuation of the dynamical activity $\cA(t)$ defined as the number of microscopic configuration changes on a volume of linear size $N$ in the time interval $[0,t]$ has been investigated. 
The mean activity scales as
$$\lim_{N\to\infty}\lim_{t\to\infty} \frac{\bra \cA(t) \ket}{N t} = \mathbb A \, , $$ 
where $\mathbb A$ depends on the density and on the choice of the constraints, as we will detail in Section \ref{models}. 
Thus  one could expect that the probability of observing a deviation from the mean value scales as
\begin{equation}
\label{eq: LD intro}
\lim_{N\to\infty}\lim_{t\to\infty} \;  \frac{1}{Nt} \log \left \bra \frac{\cA(t)}{N t} \simeq a \right \ket =-f(a) \, ,
\end{equation} 
with $0<f (a) <\infty$ for $a\neq \mathbb A$ as it occurs for the models without constraints. However, as it has been observed in \cite{GJLPDW1,GJLPDW2}, due to the presence of the constraint  it is possible to realize at a low cost a trajectory with zero 
activity by starting from a completely filled configuration and imposing that a single site does not change its state (see Section \ref{heuristics} for a detailed explanation of the mechanism behind this phenomenon).
 Analogously  one can obtain a smaller activity than the mean one by blocking for a fraction of time a  single site. As a consequence of this sub-extensive cost for lowering the activity
$f(a)=0$ for $a<\mathbb A$.
%and the above guess is not the correct scaling for these fluctuations . 
For the same reason, the moment generating function 
\begin{eqnarray}
\label{eq: 1st order intro}
\psi(\gl) = \lim_{N \to \infty} \lim_{t \to \infty} \; 
\frac{1}{N t} \log \left \bra \exp \big( \gl \cA(t) \big) \right \ket 
\end{eqnarray}
is non analytic at $\lambda=0$ with a discontinuous first order derivative \cite{GJLPDW1,GJLPDW2}.

%In the latter works   it was left as an open issue to understand which may be the consequences of this transition on the dynamics and it was suggested that in general the large deviation approach could provide a natural way to define a dynamical free energy characterizing glassiness.
%which explains why the study of the large deviation of the activity for the KCSM is currently of great interest for the physics community.
% As explained below, via our results we indeed quantify a consequence of this transition by establishing that the scaling of the  fluctuations of small activities is different from the scaling for the models without constraints.  Before mentioning our results we also wish to recall that in the recent years also in the mathematical  literature there has been a great interest 
%in the study of large deviation functionals which have been proposed as an alternative to replace the notion of free energy in non-equilibrium statistical mechanics.  In this spirit, several studies have been devoted
%to transport properties of non-equilibrium diffusive dynamics  for which the density and the current flowing through the system are the relevant parameters (see for reviews  \cite{BDGJL,BD2007,D, Gallavotti}).
%{\bf quel refernce pour gallavotti?}

\medskip

In this paper, we investigate the finite size scaling of the first order transition \eqref{eq: 1st order intro}.
Our main results are estimates of the cost of phase coexistence between the active and inactive dynamical phases.
From these estimates, the relevant scaling asymptotic in  \eqref{eq: 1st order intro} can be determined.
For East and FA1f in  one  dimension, we prove (Theorem \ref{teo:phasetrans})  that
$$\gp(\alpha)                                                                                                                                                                                                                                                                                        :=\lim_{N\to\infty}\lim_{t\to\infty} \frac{1}{t} \log \left \bra \exp \left( \frac{\alpha \cA (t)}{N} \right ) \right \ket$$  
satisfies $\gp(\alpha)=\alpha \mathbb A$ if $\alpha>\alpha_0$ and $\gp(\alpha)=-\Sigma$ if $\alpha<\alpha_1$.  
This shows that a transition in \eqref{eq: 1st order intro} occurs for a value $\gl = \frac{\ga_c}{N}$ with 
$\ga_1 < \ga_c < \ga_0$.
As a consequence, the scaling of the large deviations differs  for increasing or decreasing the activity (see Theorem \ref{fluctu1}). 
%Furthermore we establish  the scaling of the fluctuations of the activity below the mean value proving that 
%$$\lim_{t\to\infty} \;  \frac{1}{t} \log  \left \bra \frac{\cA(t)}{Nt} \in I  \right \ket =-g(I)$$ 
%with $0<g(I)<\infty$  if $\mathbb A\not\in I$ and $I\cap [0,\mathbb A]\neq \emptyset$ (Theorem \ref{fluctu1}).  

We also analyze the measure on the space-time configurations 
$$\mu_{\alpha,T}^{N}: =  \frac{\bk{\cdot \; \exp( \frac{\ga}{N} \cA(T) )}}{\bk{\exp( \frac{\ga}{N} \cA(T) )}} $$
which corresponds  to the conditional measure with a fixed activity $\frac{\cA(t)}{t} = \bar A$ where $\ga$ is the parameter conjugated to $\bar A$ in Legendre transform
and prove (Theorem \ref{teo:condmes})  that depending on the value of $\alpha$ this measure 
 has very different typical configurations which can be interpreted as active and inactive dynamical phases : for $\alpha>\alpha_0$,  $\mu_{\alpha,T}^{N}$ concentrates  on trajectories with the mean activity and for $\alpha<\alpha_1$,  
 it concentrates on trajectories with zero activity.

Finally, we investigate the higher dimensional cases and show that the finite size scaling depends not only on the dimension but also on the boundary conditions (Theorem \ref{linearityd>1}). This leads to a variety of scalings for the large deviations when  the activity is reduced
 (Theorem \ref{fluctu2}).

%Finally, we investigate the higher dimensional cases. For the East model in dimension $d\geq 2$ it is immediate to verify by using the above one dimensional results that   a dynamical transition occurs for the generating function in the scaling $\lambda=\alpha/N^{c}$ with $c=d-1$, namely  the above behavior carries over for 
%$\gp_{d}(\alpha):=\lim_{N\to\infty}\lim_{t\to\infty} 1/N^{c}\exp(\alpha \cA/N^{d-c})$ (Remark \ref{remeast}). For FA-1f we can only establish the existence of $\alpha_0<0$ s.t. $\gp_d(\alpha)=\alpha \mathbb A$ if $\alpha>\alpha_0$ where $c\leq d-1$ is  an exponent that is fixed  by  the choice of the boundary conditions (Theorem \ref{linearityd>1}). Even if this result  is not enough to prove a phase transition for the generating function it allows, together with the blocking mechanism induced by the constraints,   to establish the scaling of the fluctuations below the mean value:  
%$$\lim_{N\to\infty} \lim_{t\to\infty} \;  \frac{1}{N^c \; t} \log  \left \bra \frac{\cA(t)}{N^d \; t} \in I  \right \ket =-g_{d,c}(I)$$ 
%with $0<g_{d,c} (I)<\infty$ if $I\cap [0,\bbA]\neq 0$ and $\bbA\not\in I$ (Theorem \ref{fluctu2}).

\section{Models and results}
\label{models}

\subsection{East and FA-1f in $d=1$: the phase transition}

The East and FA-1f models in one dimension are Glauber type Markov processes on the configuration space
$\gO = \{0,1\}^{\Lambda}$ where $\Lambda\subset\bbZ$. 
Both models depend on a parameter $\rho$, with $\rho\in(0,1)$, which we will call the {\sl density}.
Here we will consider the models in finite volumes $\Lambda=\Lambda_N:=[1,N]$ and we will be interested in the thermodynamic limit $N\to\infty$.
We call $\gO_N$ the configuration space correspondent to $\Lambda_N$ and denote by greek letters $\eta,\omega$ the elements of $\Omega_N$. Then for any site  $i\in\Lambda_N$ we let  $\eta_i\in(0,1)$ be the value of configuration $\eta$ at site $i$ and we  say that $i$ is {\sl empty} ({\sl filled}) if $\eta_i=0$ ($\eta_i=1$, respectively).

The Markov process corresponding to both models can be informally described as follows. 
Each site $i\in\Lambda_N$ waits an independent mean one exponential time and then, provided  the current configuration satisfies a proper local constraint, we refresh the value of the configuration at $i$ by setting it to $1$ with probability $\rho$ and to zero with probability $1-\rho$. Instead if the constraint is not satisfied nothing occurs. Then the procedure starts again. The specific choice of the constraint identifies the model: for East 
one requires that
  the right nearest neighbour of $i$  is empty; for FA-1f model  one requires that at least one among the right and left nearest neighbours of $i$ is empty. In formulas the constraint at $i$ is satisfied for East and FA-1f in the configuration $\eta$ iff $\eta_{i+1}=0$ and $\eta_{i+1} \, \eta_{i-1}=0$, respectively. Note that in both models the constraint is local (it depends on the configuration on a finite neighborhood of the to-be-updated site) and does not depend on the value of the configuration on the to-be-updated site. Both models belongs to the larger class of Kinetically Constrained Models (KCM in short), which have been introduced and widely studied in physics literature (see for reviews 
  \cite{RS,GST}). 

In order for the above description to be complete, we need to specify what happens at sites near the boundary.
% More precisely, for the East model at the occurrences of the exponential variable on site $N$ it is not meaningfull to impose a constraint that requires the next site to the right of site $N$ to be empty (since $\sigma$ is defined  on $[1,N]$). Analogous problems arise for FA-1f model at sites $1$ and $N$.
A standard choice in statistical mechanics is to defined the dynamics in finite volume by imposing a fixed boundary condition.
%, namely by letting the dynamics inside the volume evolve as if the configuration was defined and fixed outside this volume.
%, namely saying that evolution occurs as if on the exterior of the volume there was a  configuration which does not evolve under the dynamics. 
Here the choice of this boundary condition is very delicate, indeed  due to the presence of the constraints both models are very sensitive to the specific choice of these conditions even on large volumes. For example for the East model it is easy to verify that if we fix a boundary condition equal to one at $N+1$, we start the evolution from $\eta \in \Omega_N$ and we let $x$ be the position of the rightmost zero of $\eta$, then at any subsequent time site  $x$ stays empty and sites $[x+1,N]$ stay filled. In this case we say that the configuration is {\sl frozen} on $[x,N]$ meaning that  under the  evolution, the configuration on these sites remains unchanged. 
On the other hand, for a boundary condition equal to zero at $N+1$, it can be easily verified that there is no site on which the configuration is frozen no matter which is the choice of the initial configuration. From the above observation it follows that in the case of filled boundary condition  the configuration space (even on finite volume) is not irreducible. Indeed there exists configurations $\sigma,\eta\in \Omega_N$ such that it is not possible to devise a path of elementary moves with strictly positive rates 
%(i.e. moves corresponding to occupation variable changes on sites with the constraint satisfied) 
which connects $\sigma$ to $\eta$. Instead if we take an empty boundary condition the configuration space is irreducible. This can be easily verified by constructing  a path which completely empties any configuration starting from the right boundary.   Analogously for FA-1f model if one imposes filled boundary conditions both at $0$ and $N+1$ the   configuration space is reducible. On the other hand
any of the choices which has at least one empty site in the couple  $(0,N+1)$ is sufficient to guarantee irreducibility.
%(again starting from the empty boundary one can easily construct a path which completely empties any configuration). 
Here for both models we will only be interested on choices of the boundary conditions which guarantee irreducibility (and therefore ergodicity as we consider finite systems).

Note also that for both models, no matter which choice we perform for the boundary condition, the dynamics satisfy detailed balance with respect to Bernoulli product measure $\nu$ at  density $\rho$, namely $\nu(\eta):=\prod_{i\in\Lambda_N}\nu_i(\eta_i)$ with $\nu_i(1)=\rho$
(this is a direct consequence of the above observed fact that the constraint at $i$ does not depend on the value of $\eta_i$).
 Therefore $\nu$ is  an invariant measure for the process and
 in the irreducible case this is the unique invariant measure.

\smallskip

Let us now give a formal definition of these processes
via the action of their generator $\cL_N$ on local functions $f:\Omega_N\to\bbR$. 
We introduce
\begin{eqnarray}
\label{gene}
\cL_N f (\eta) = \sum_{i\in \gL_N}
c_i(\eta) 
  \big( f(\eta^i) - f (\eta) \big)
\end{eqnarray}
where $\eta^i$ stands for the configuration $\eta$ changed at $i$, namely
\begin{equation}
  \eta^i_j=
  \begin{cases}
 \eta_j & \text{ if }\, j\neq i\\
1-\eta_j & \text{ if }\, j=i
  \end{cases}
\label{flip}
\end{equation}
and we let
\begin{equation}\label{ci}
c_i(\eta):=r_i (\eta)[\eta_i(1-\rho)+(1-\eta_i)\rho]
\end{equation}
with $r_i$ the function that encodes the constraint at site $i$, namely $r_i(\eta)=1$ ($r_i(\eta)=0)$ iff the constraint at $i$ is (is not) satisfied. Thus $r_i$ is model dependent and for the East model with frozen empty boundary condition at the right boundary we set
\begin{equation}\label{ceast}
r_i(\eta):=(1-\eta_{i+1})\,\,\,\,\,\,{\mbox{ if }} i\in[1,N-1];\,\,\,\,\,\,\,\,\,r_N=1
\end{equation}
for FA-1f with empty boundary condition at the right and left boundary we set
\begin{equation}\label{cFA-1ftwo}
r_i(\eta):=(1-\eta_{i+1}\eta_{i-1})\,\,\,\,\,\,{\mbox{ if }} i\in[2,N-1];\,\,\,\,\,\,\,\,\,r_1=1,\,\,\,r_N=1
\end{equation}
for FA-1f with empty boundary condition at the right boundary and occupied boundary condition at the left boundary we set
\begin{equation}\label{cFA-1fone}
r_i(\eta):=(1-\eta_{i+1}\eta_{i-1})\,\,\,\,\,\,{\mbox{ if }} i\in[2,N-1];\,\,\,\,\,\,\,\,\,r_1=(1-\eta_2),\,\,\,r_N=1
\end{equation}
and finally for
FA-1f with empty boundary condition at the left boundary and occupied boundary condition at the right boundary we set
\begin{equation}r_i(\eta):=(1-\eta_{i+1}\eta_{i-1})\,\,\,\,\,\,{\mbox{ if }} i\in[2,N-1];\,\,\,\,\,\,\,\,\,r_1=1,\,\,\,r_N=(1-\eta_{N-1}).
\label{cfas}
\end{equation}
As already mentioned our analysis will focus on the above choices which are the only ones which guarantee ergodicity.  Therefore in the following when we refer to the East model, to FA-1f with two empty boundaries and to FA-1f with one empty boundary  we mean respectively the choice
\eqref{ceast}, \eqref{cFA-1ftwo} and  \eqref{cFA-1fone} (by 
symmetry reasons the choices \eqref{cFA-1fone} and \eqref{cfas} for FA-1f are equivalent therefore we never consider the case
\eqref{cfas}). Also, when we state results referred to FA-1f model without further
specifying the boundary conditions it means that these results hold for both the choices  \eqref{cFA-1ftwo}
and \eqref{cFA-1fone}.

%where $\tau$ is any configuration on $\bbZ^d\setminus\Lambda_N^d$ which equals $\beta$ on the boundary set $\mathcal{B}$.
%The above defined process can be informally described as follows. On each vertex $i\in\Lambda_N^d$ we have a Poisson clock of mean time one. Poisson clocks on different sites are independent. Then at each ring of the clock at site $i$ we check whether  the current configuration satisfies the constraint at $i$. If the latter is not satisfied nothing occurs. Otherwise  we refresh the value at $i$ by setting it to $1$ with probability $\rho$ and to zero with probability $1-\rho$ (which corresponds to say that, if the constraint is verified, then: if the site is occupied we change its configuration with probability  $1-\rho$, if it is empty we change with probability $\rho$).
% and 
%we let $\mathbb P_{\eta}(\mathcal{A})$ be the corresponding probability w.r.t. Poisson clocks and to the corresponding Bernoulli $\rho$ variables. Given any probability distribution $\mu$ on $\Omega$ we also let
%$P_{\mu}(\mathcal{A})=\int d\mu(\eta)\mathbb P_{\eta}(\mathcal{A})$. In words $P_{\mu}$ is the distribution started at time zero from distribution $\mu$.
%
%We stress that, even if we do not explicitely write this dependence to keep the notation as simple as possible, the generator in \eqref{gene} (and thus the corresponding KCSM) depend on several parameters:
%the density $\rho$, the dimension $d$, the linear size $N$ of the volume,
%the boundary condition $\beta$ and (of course!) the choice of the influence classes which encode
%the constraints.  
Note that the generator \eqref{gene} can be equivalently rewritten as 
\begin{eqnarray}
\cL_N f (\eta) = \sum_{i\in \gL_N}
r_i(\eta)  \big( \nu_{i}(f) - f (\eta) \big)
\end{eqnarray}
with $\nu_i(f)=\int d\nu_i(\eta_i)f(\eta)$ the local mean at site $i$ .

\medskip

The object of interest of our work is
 the total activity \begin{equation}
 \label{defactivity}\cA(t): = \sum_{i\in \gL_N} \cA_i (t)\end{equation}
  where $\cA_i(t)$ is the random variable which corresponds to the number of configuration changes on site $i$ during the time interval $[0,t]$. It is easy to verify that $\cA_i(t)-\int_0^t c_i(\eta(s))ds$ is a martingale and therefore $\cA(t)$ satisfies a law of large numbers with
 $$\lim_{N\to\infty}\lim_{t\to\infty}\frac{\cA(t)}{Nt} =\mathbb A$$ where $\mathbb A$,
which will be referred to in the following as  the mean instantaneous activity, is defined as
 \begin{equation}
 \label{bbA}
 \bbA:=\nu(c_j(\eta)), \qquad j\in \Lambda_{N}\setminus (1,N)
 \end{equation}
note that the definition is well posed since by  translation invariance $\nu(c_j(\eta))=\nu(c_{j'}(\eta))$ for $j,j'\in[2,N-1]$. 
For East, we get  $\mathbb A=2\rho(1-\rho)^2$, for FA-1f instead $\mathbb A=2\rho(1-\rho)(1-\rho^2)$.

Here we will be interested in the study  of the  generating function which controls the
 fluctuations of $\cA(t)$ for a given $N$
\begin{equation}
\label{eq:free}
\gp^{(N)}(\gl) = \lim_{t\to\infty} \frac{1}{t} \log \bra \exp( \gl \cA(t) ) \ket
\end{equation}
where here and in the following $\bra \ket$ denotes the mean  over the evolution of the process and over   the initial configuration which is distributed with the equilibrium Bernoulli measure $\nu$ at density $\rho$ (where the density $\rho$ is fixed by the rates \eqref{ci}).
With a slight abuse of notation for any event $\cE$ we will also denote by $\bra \cE\ket$ the probability of $\cE$, namely we set $\bra\cE\ket:=\bra\mathds{1}_{\cE}\ket$. 
%and namely over the distribution $P_{\nu}$ defined in Section \ref{models}. 
%We explicit in our notation only the dependence on $N$ to stress that this is a finite volume object at variance with the rescaled functional defined below in \eqref{rescaled}, but of course there is a dependence on the whole set of parameters $N$, $d$, $\rho$, $\beta$ and $\{\mathcal{C}\}_i$. 
The main result of this paper is   that in the scaling $\gl = \alpha/N$  a phase transition occurs for this generating function. More precisely
if we define
\begin{equation}
\label{rescaled}
\gp(\ga): = \limsup_{N\to\infty} \gp^{(N)}\left(\frac{\ga}{N}\right)
\end{equation} 
then the following holds:
\begin{Theorem}
\label{teo:phasetrans}
Consider East or FA-1f model in $d=1$ at any $\rho\in(0,1)$. There exists $\alpha_1<\alpha_0<0$ and a constant $\gS> 0$ such that
\begin{itemize}
\item[(i)] for $\alpha>\alpha_0$ it holds $\gp\left(\ga\right) =  \bbA \ga$;
\item[(ii)] for $\alpha<\alpha_1$ it holds $\gp(\ga)= - \gS$.\\
\end{itemize}
\end{Theorem}

As a consequence of this theorem, estimates on the large deviations for a reduced activity can be obtained.

\begin{Theorem}
\label{fluctu1}
Consider East or FA-1f model in $d=1$ at any $\rho\in(0,1)$.  
%For any $\delta>0$ and $c\in[0,1-\frac{\delta}{\bbA}]$ it holds
For any  $u \in[0,1)$ it holds
\begin{eqnarray}
\label{eqfluc}
%-\gS (1-c)\leq  \limsup_{N\to\infty} \lim_{t\to\infty}\frac{1}{t}\log <\frac{\cA(t)}{Nt}\in [c\mathbb A-\delta, c\mathbb A+\delta]>\leq \alpha_0 (\bbA-c\bbA-\delta)
-\gS (1-u)
&\leq&  \liminf_{N\to\infty} \lim_{t\to\infty}\frac{1}{t}
\log \left \bra \frac{\cA(t)}{Nt} \simeq u \bbA \right \ket  \\
&\leq&  \limsup_{N\to\infty} \lim_{t\to\infty}\frac{1}{t}
\log \left \bra \frac{\cA(t)}{Nt} \simeq u \bbA \right \ket 
\leq \alpha_0 \bbA (1-u) \, . \nonumber
\end{eqnarray}

%\begin{remark}
%\label{remarkfluctu1}
%We conjecture (i) and (ii) to hold respectively above and below $\alpha_c:=- \frac{\Sigma}{\mathbb A}$. 
%Should it be the case, \eqref{eqfluc} would become 
%$$
%\forall u \in [0,1], \qquad 
%\lim_{N\to\infty} \lim_{t\to\infty}\frac{1}{t}\log \left \bra \frac{\cA(t)}{Nt}\simeq u \bbA \right \ket  =-\Sigma(1-u) \, .$$
%\end{remark}
\end{Theorem}
\begin{remark} We conjecture that in \eqref{rescaled} the $\limsup_N$ can be replaced by $\lim_N$.
%$\lim_{N\to\infty}\gp^{(N)}\left(\frac{\ga}{N}\right)$ exists. 
In the regime $\alpha>\alpha_0$ and $\alpha<\alpha_1$, this follows from the  proof of  Theorem \ref{teo:phasetrans}.
\end{remark} 
Theorem \ref{teo:phasetrans} (i) and (ii) will be proved in Section \ref{linearity} and \ref{largealpha} respectively.
Theorem \ref{fluctu1} will be proven in Section \ref{new}.

\smallskip

We also analyze the measure on the space-time configurations defined as

\begin{eqnarray}
\label{eq: measure}
\mu_{\alpha,T}^{N} =  \frac{\bk{\cdot \; \exp( \frac{\ga}{N} \cA(T) )}}{\bk{\exp( \frac{\ga}{N} \cA(T) )}} 
\end{eqnarray}
and Theorem \ref{teo:condmes} states that depending on the value of $\alpha$ this measure 
%(which is related to the conditional measure with a fixed activity $\frac{\cA(t)}{t} = \bar A$ where $\ga$ is the parameter conjugated to $\bar A$)
 has very different typical configurations.
 For any configuration $\eta\in\Omega_{N}$, we call $\{\eta(s)\}_{s\geq 0}$ the trajectory of the Markov process generated by $\cL_N$  starting at time zero from $\eta$. Then
 the following holds 

\begin{Theorem}
\label{teo:condmes}
Consider East model and FA-1f model in $d=1$ with $\rho\in(0,1)$. Then there exists $\alpha_1<\alpha_0<0$  and a sequence $\gamma_N$ with  $\lim_{N\to\infty} \gamma_N = 0$ such that 
\begin{itemize}
\item[(i)] for $\alpha>\alpha_0$ 
\begin{eqnarray}
\label{eq: densite haute tris}
\lim_{N\to\infty}\lim_{T \to \infty}
\mu_{\alpha,T}^{N}\left({\int_0^T dt \,  |\sum_{i \in \gL_N } \,\eta_i(t)-N\rho| \leq \gamma_N N T}\right) = 1 \, ;
\end{eqnarray}
\begin{eqnarray}
\label{eq: densite haute tqua}
\lim_{N\to\infty}\lim_{T \to \infty}
\mu_{\alpha,T}^{N}\left({\int_0^T dt \,  |\sum_{i \in \gL_N } \,c_i(\eta(t))-N\bbA| \leq \gamma_N N T}\right) = 1 \, .
\end{eqnarray}
\item[(ii)]
if $\alpha<\alpha_1$ 
\begin{eqnarray}
\label{eq: densite haute tris}
\lim_{N\to\infty}\lim_{T \to \infty}
\mu_{\alpha,T}^{N}\left({\int_0^T dt \,  \sum_{i \in \gL_N } \,\eta_i(t) \geq (1-\gamma_N) N T} \right)= 1 \, ;
\end{eqnarray}
\begin{eqnarray}
\label{eq: densite haute tris}
\lim_{N\to\infty}\lim_{T \to \infty}
\mu_{\alpha,T}^{N}\left({\int_0^T dt \,  \sum_{i \in \gL_N } \,c_i(\eta(t)) \leq \gamma_N N T}\right) = 1 \, .
\end{eqnarray}
\end{itemize}
\end{Theorem}

Theorem \ref{teo:condmes}(i) and (ii) will be proven in Section \ref{condmes1} and \ref{condmes2} respectively where
 stronger results (Lemma \ref{usolemma} and \ref{usolemma2}) concerning the concentration of the number of particles and the activity on mesoscopic boxes (and not on the whole volume) will also be established.

\subsection{Heuristics of the phase transition and open problems}
\label{heuristics}

As we already mentioned in the introduction, the occurrence of a phase transition for the activity large deviations was first discovered in \cite{GJLPDW2}, where it was shown that 
\begin{eqnarray}
\label{eq: psi}
\gl \in \bbR, \qquad
\psi(\gl) = \lim_{N \to \infty} \lim_{t \to \infty} \; 
\frac{1}{N t} \log \left \bra \exp \big( \gl \cA(t) \big) \right \ket  
\end{eqnarray}
has a critical value at $\gl =0 $ (see the left part of figure \ref{fig: courbes}). 
We recall below the mechanism of this phase transition in the case of the one-dimensional East model (the case of FA-1f model is analogous). 
When $\gl >0$, the activity is increased and the large deviation functional is expected to be smooth.
Negative values of $\gl$ lead to a decay of the activity and the constraint will play a crucial role. To have no activity, a possible strategy is to start at time $t = 0$ from a configuration totally filled ($\eta_i = 1$ for all $i$) and to remain in this configuration at any time. 
This can be achieved by preventing $\eta_N$  from flipping to 0, indeed the site $N$ is the only site allowed to flip due to the constraints and if it is maintained equal to $1$ then the rest of the configuration is blocked. This leads to the lower bound 
\begin{eqnarray}
\label{eq: crude LB}
\left \bra  \cA(t)= 0 \right \ket \geq \gr^N \exp( -(1-\gr) t) \, ,
\end{eqnarray}
where $\gr^N$ stands for the cost of the initial configuration and the last term is the probability that a Poisson process of intensity $1-\gr$ has no jump up to time $t$.
After rescaling \eqref{eq: psi}, this shows that $\psi(\gl) \geq  0$ for $\gl <0$ and since $\psi(0)=0$ and $\psi$ is increasing in $\lambda$ it follows that $\psi(\lambda)=0$ for $\lambda\leq 0$. On the other hand by convexity $\psi(\lambda)\geq \lambda\mathbb A$ and therefore at $\lambda=0$ the first order derivative of $\psi$ has a jump. 
As noted in \cite{GJLPDW1,GJLPDW2}, it is remarkable that the phase transition occurs at $\gl = 0$ which corresponds to the unperturbed dynamics. Thus one may wonder if the singularity of the  large deviation functional would lead to 
specific properties of the constrained systems.

In this paper, we investigate the finite size scaling of this first order phase transition through the 
function $\gp(\ga)$ introduced in \eqref{rescaled}. This refined thermodynamic scaling corresponds to a blow up of the region $\gl = 0$. The results of Theorem \ref{teo:phasetrans} are depicted 
in the right part of figure \ref{fig: courbes}.  In a range of $\gl$ of order $1/N$,
we have shown that the transition is shifted from 0.
\begin{figure}[h]
\begin{center}
\psfragscanon 
\centerline{
\epsfxsize = 5 cm
\psfrag{F}[c][l]{\small $\psi(\gl)$}
\psfrag{l}[c][l]{\small  $\gl$}
\epsfbox{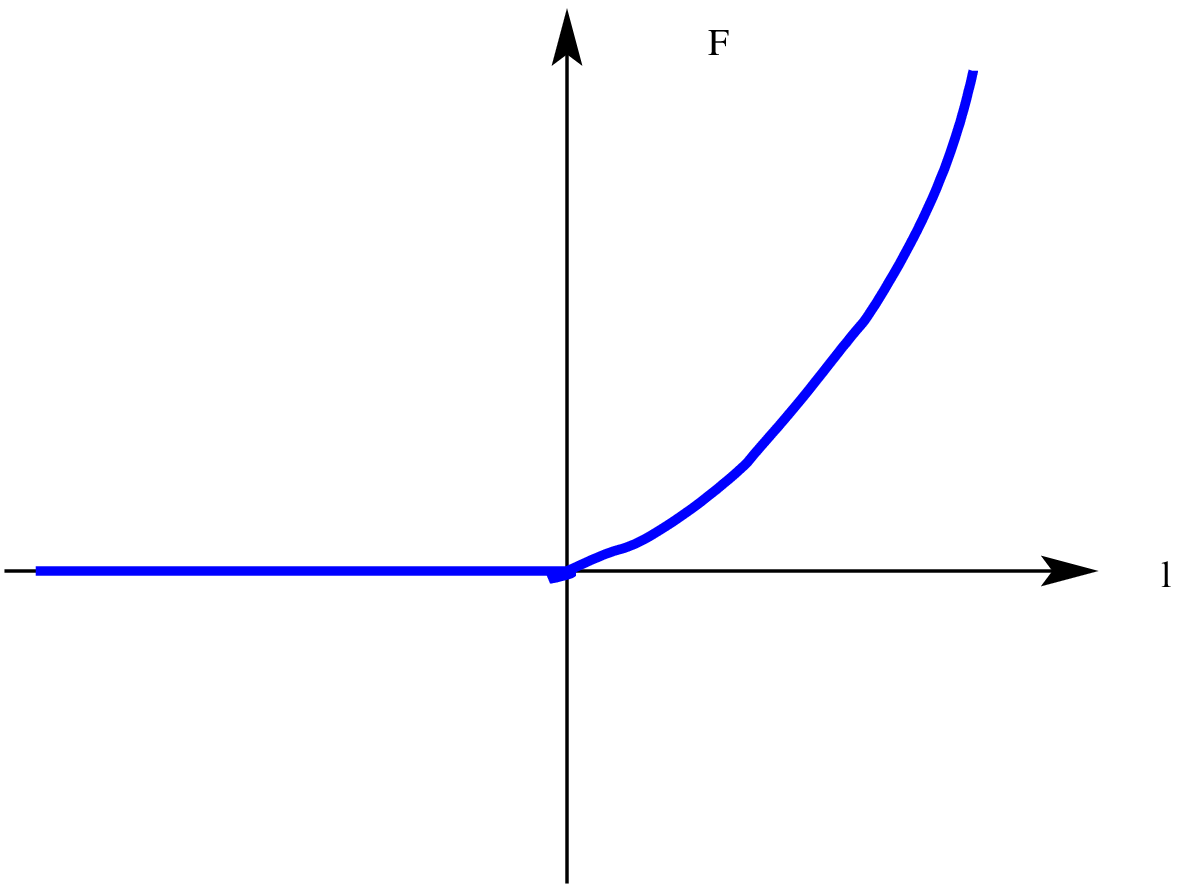} 
\hspace{1.5cm}
\epsfxsize = 6 cm
\psfrag{F}[b][l]{\small $\gp(\ga)$}
\psfrag{l}[c][l]{\small  $\ga$}
\psfrag{a}[c][l]{\small  $\ga_1$}
\psfrag{b}[c][l]{\small  $\ga_0$}
\psfrag{c}[c][l]{}%{$\gl_c$}
\psfrag{S}[c][c]{\small  $\ \ - \gS$}
\epsfbox{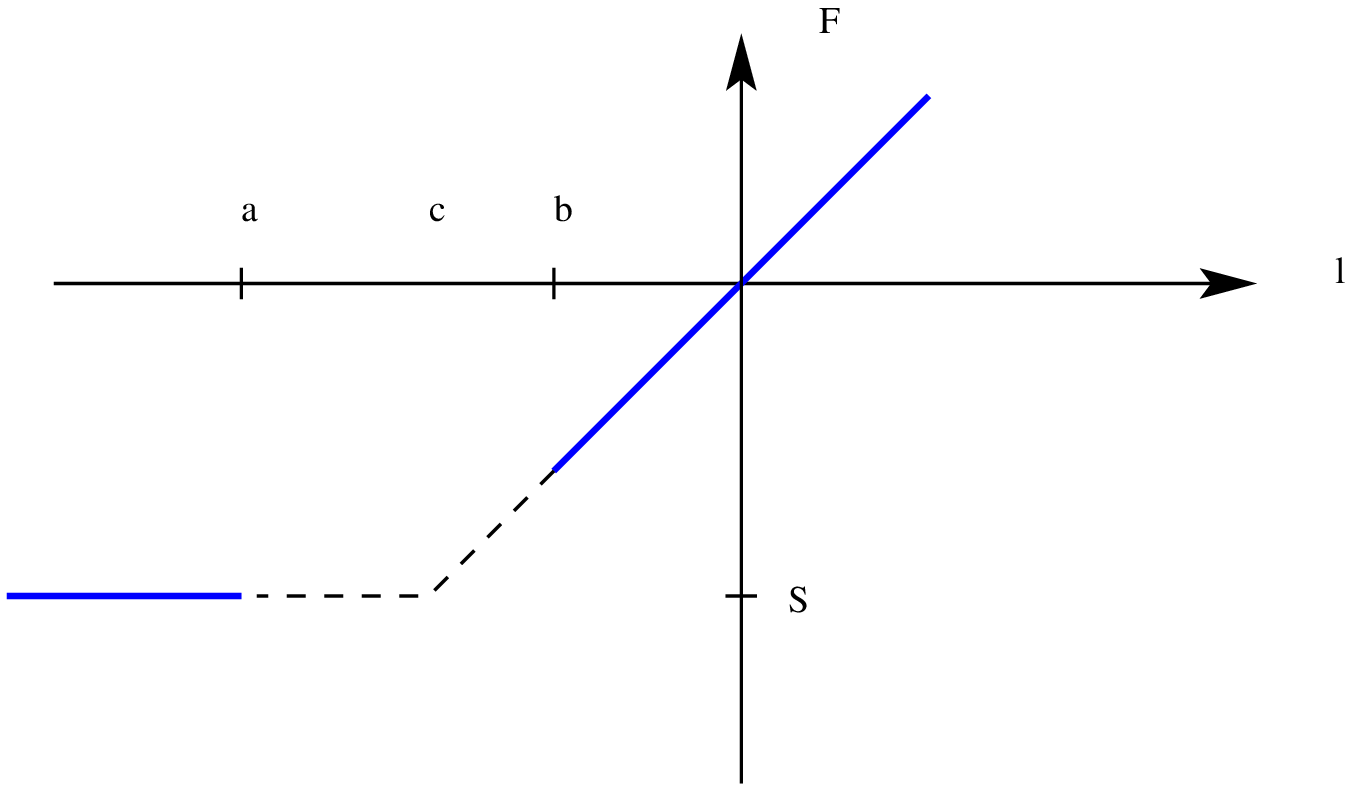} }
\end{center}
\caption{The function $\psi$ \eqref{eq: psi} is depicted on the left. The right figure represents the graph of the function $\gp$ \eqref{rescaled}: the results of Theorem \ref{teo:phasetrans} are in thick line and the conjectured behavior in dashed lines.}
\label{fig: courbes}
\end{figure}
To understand this, we first suppose that no transition takes place at 0 and that $\psi$ could be expanded (analytically). If this was the case then we would expect that for large $N$
\begin{eqnarray*}
\label{eq: expansion}
\gp (\ga) \simeq N \psi ( \frac{\ga}{N}) =   N \left[ \psi (0) + \psi' (0)  \frac{\ga}{N} + {\small O \left( \frac{1}{N^2} \right)} \right] \,
= \bbA \ga+ {\small O \left( \frac{1}{N}\right)} \, ,
\end{eqnarray*}
where we used that $\psi (0) = 0$ and $\psi' (0)$ is equal to the mean activity $\bbA$ 
(in fact $\psi$ is not differentiable at 0, but its right-derivative is equal to $\bbA$).
Part (i) of Theorem \ref{teo:phasetrans} shows that this behavior persists for  negative values of $\ga$ provided $\ga > \ga_0$.
In this regime, we expect that the total activity is shifted from its mean value $\bbA N t$ by an order 
$\ga \,  t \,  \psi''(0)$ which does not scale with $N$. 
For such small shifts of the activity, the system remains very close to its typical state (when $N$ is large). In particular, Theorem \ref{teo:condmes} asserts that the mean density is very concentrated close to its equilibrium value $\gr$.

A phase transition occurs for smaller values of $\ga$ and $\gp$ becomes equal to the constant 
$- \gS > - (1-\gr)$. The estimate \eqref{eq: crude LB} leads to the lower bound $- (1-\gr)$ and thus it is too crude to justify the claimed behavior.
Indeed for $\gl = \frac{\ga}{N}$, an activity of order $o(N)$ will not contribute to the scaling limit
\eqref{rescaled}, thus it is more favorable to leave  a small portion of the system active as depicted in 
figure \ref{fig: interface} instead of forcing the whole configuration to remain totally filled.
\begin{figure}[h]
\begin{center}
\leavevmode
\epsfysize = 4 cm
%\psfragscanon 
%\psfrag{T}[c][l]{{\small  $t$}}
\psfrag{N}[t][l]{{\tiny  $o(N)$}}
\epsfbox{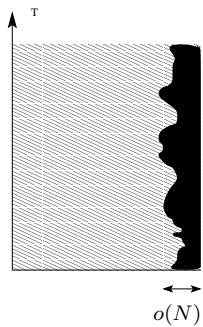}
\end{center}
\caption{A space/time picture of the system is depicted. The left part is  inactive  and separated from the active boundary (at the right) by a fluctuating interface.}
\label{fig: interface}
\end{figure}
By analogy with equilibrium, one can interpret $\gS$ as a surface tension between the inactive and the active region (per unit of time)
\begin{eqnarray*}
\label{eq: refine LB}
\left \bra  \frac{1}{N t} \cA(t) \approx 0 \right \ket  \simeq \exp( - \gS t) \, ,
\end{eqnarray*}
where $\frac{1}{N t} \cA(t) \approx 0$ means that the rescaled activity is close to 0 in the thermodynamic limit. Contrary to the strategy in \eqref{eq: crude LB}, the interface between the inactive and the active region is now allowed to fluctuate and the probabilistic cost is lowered.
The surface tension $\gS$ can be obtained from a variational problem which is specified for FA-1f with two empty boundaries in Section \ref{largealpha} and for East and FA-1f with one empty boundary in Section \ref{largealpha2}. However, our results do not provide a complete description of the system and it remains to prove that the typical configurations look like figure \ref{fig: interface}.
Nevertheless, Theorem \ref{teo:condmes} ensures that in the inactive regime almost all the sites are equal to 1. 
This confirms the conjectured picture.

Our results  (Theorems \ref{teo:phasetrans} and \ref{teo:condmes})  do not provide  the entire  phase diagram for the generating function $\gp$ (only for $\ga \not \in [\ga_1,\ga_0]$).
%the only informations comes from the upper and lower bounds which follow from the results outiside this region combined with the monotonicity and convexity of $\gp$.
However, this is enough to deduce (see Theorem \ref{fluctu1}) the correct order of the scaling for the large deviations of the activity below the mean value 
$$
\forall u \in [0,1], \qquad 
-\Sigma(1-u) \leq 
\lim_{N\to\infty} \lim_{t\to\infty}\frac{1}{t}\log \left \bra \frac{\cA(t)}{Nt}\simeq u \bbA \right \ket  < 
\ga_0 \bbA (1-u) \, .
$$
This scaling is anomalous compared to the  extensive scaling in $N$ of the unconstrained models \eqref{eq: LD intro}.

%we would get  $\lim_{N\to\infty}\lim_{t\to\infty}1/(Nt) \log<\cA(t)/(Nt)\simeq a>=-g(a)$ with $g(a)>0$ for $a<\mathbb A$. 

 We conjecture that there is a unique critical value $\ga_c$ and that the two regimes remain valid up to $\ga_c$ as depicted in figure \ref{fig: courbes}, namely $\gp=-\Sigma$ for $\alpha\leq \alpha_c$ and $\gp=\alpha\bbA$ for $\alpha\geq\alpha_c$.
This would imply 
$$
\ga_c = - \frac{\gS}{\bbA} \, .
$$
This conjecture is supported by numerical simulations and we refer to \cite{BLT} for an account on these numerical results. 
If this  conjecture is verified, then Theorem \ref{fluctu1} can be improved and the large deviations for reducing the activity
would be given by
$$
\forall u \in [0,1], \qquad 
\lim_{N\to\infty} \lim_{t\to\infty}\frac{1}{t}\log \left \bra \frac{\cA(t)}{Nt}\simeq u \bbA \right \ket  =-\Sigma(1-u) \, .
$$

\begin{figure}[h]
\begin{center}
%\leavevmode
\psfragscanon 
\centerline{
\epsfxsize = 5 cm
\psfrag{m}[c][l]{\small $m(h)$}
\psfrag{h}[c][l]{\small  $h$}
\epsfbox{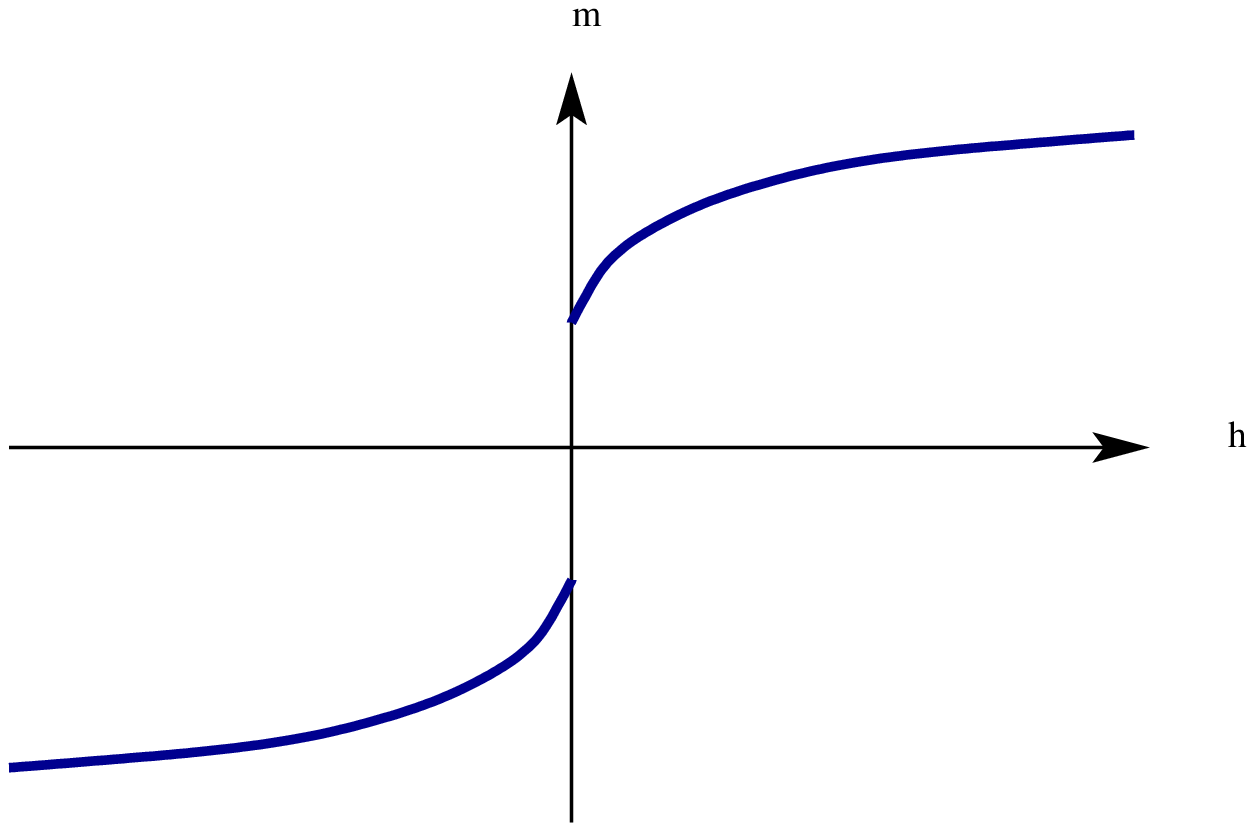} 
\hspace{1.5cm}
\epsfxsize = 6 cm
\psfrag{m}[c][l]{\small $m_N(\ga/N)$}
\psfrag{h}[c][l]{\small  $\ga$}
\psfrag{c}[b][l]{\small  $\ga_c$}
\epsfbox{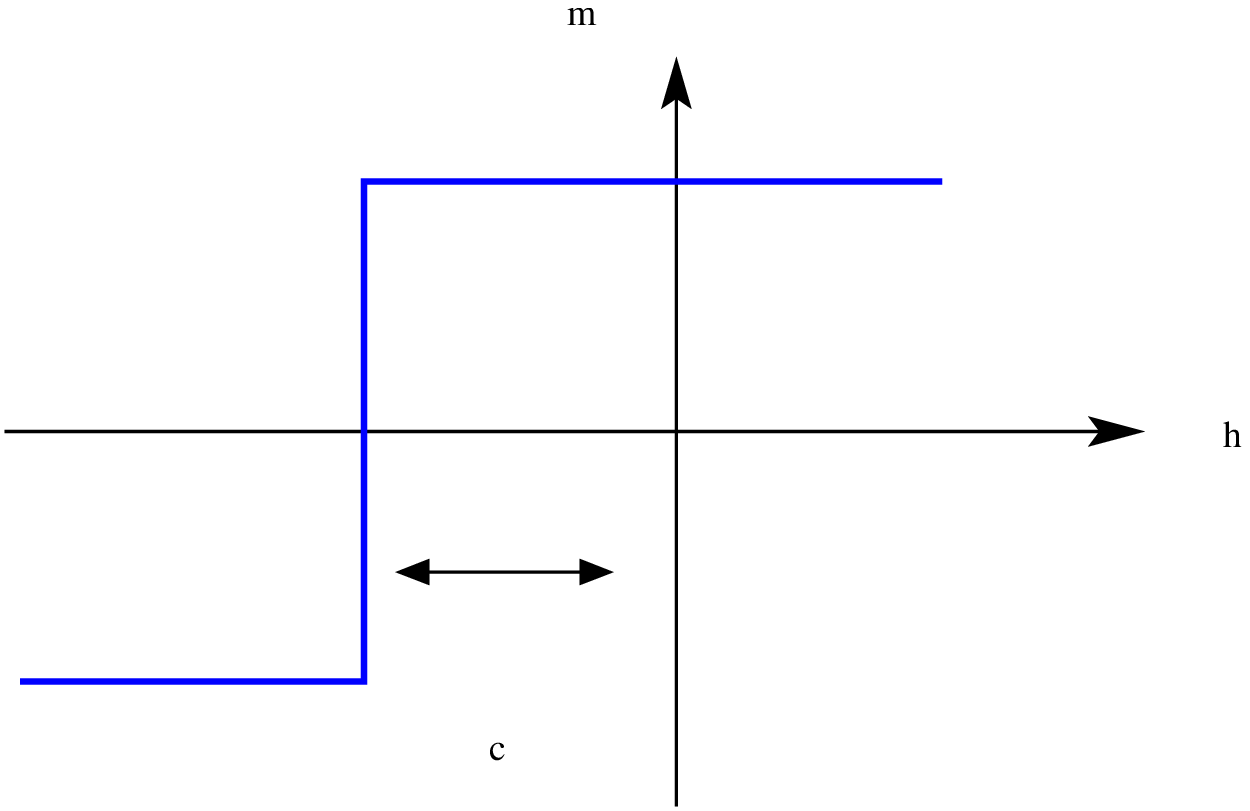} 
}
\end{center}
\caption{On the left, the graph of the magnetization $m(h)$ for the infinite volume Ising model.
On a finite domain of size $N$ (large) and with $+$ boundary conditions, the finite size scaling shows that the discontinuity occurs close to $\frac{\ga_c}{N}$ as depicted on the right.}
\label{fig: courbes ising}
\end{figure}

In  \cite{GJLPDW1,GJLPDW2}, it was suggested that the large deviation approach could provide a natural way to define a dynamical free energy characterizing glassiness.
It is currently an open question to understand if (and how) the dynamical phase transition  in $\psi(\gl)$ \eqref{eq: psi} can lead to quantitative predictions on the model at $\gl = 0$.
Equilibrium statistical mechanics could serve as a guide to clarify this issue. Indeed a similar  phenomenon to the one depicted above for East and FA-1f occurs  in the finite size scaling of the ferromagnetic Ising model.
For an Ising model in the phase transition regime ($T< T_c$), a first order phase transition occurs in the magnetic field $h$ and the magnetization $m(h)$ is discontinuous at $h=0$ (see figure \ref{fig: courbes ising}). On a finite domain, say a square of size $N$, with external boundary conditions $+$, then the magnetization $m_N(h)$ is continuous and approaches the graph 
of $m(h)$. A finite size scaling \cite{SS} shows that up to rescaling $h = \ga /N$, the magnetization $m_N(\frac{\ga}{N})$ converges to a step function with  a jump at a critical value $\ga_c \not = 0$  (see figure \ref{fig: courbes ising}).
The shift of the transition is reminiscent of the shift for the constrained models and it can be understood as follows. For $\ga \in  [\ga_c, 0]$, the magnetization is slightly lowered but remains close to the magnetization $m^* = \lim_{h \to 0^+} m(h)$  imposed by the $+$ boundary conditions, then for $\ga < \ga_c$ the negative magnetic field forces a droplet of the $-$ phase which fills the system.
The creation of this droplet has a cost proportional to a surface order $\tau N$ (where $\tau$ is the surface tension term), but leads to an energy gain $- 2 m^* h N^2$. Thus, the critical value is obtained for
$$
- 2 m^* h N^2 =  \tau N  \quad \Rightarrow 
\quad \ga_c =  - \frac{\tau}{2 m^*} \, .
$$
In this analogy, the magnetization plays a role similar to the activity and $h, \gl$ are the conjugate parameters. Even so a first order phase transition occurs for the Ising model at $h$ equal 0, 
it is known that the $+$ pure phase, i.e. the Gibbs measure obtained from the $+$ boundary conditions after the thermodynamic limit, is well behaved and that the cumulants of the magnetization can be obtained by taking the successive derivative of the pressure for $h \to 0^+$.   
As $\gl$ has no physical meaning (contrary to $h$), it is not clear from the mere knowledge of the first order phase transition how to deduce precise informations on the constrained dynamics at $\gl =0$.

\subsection{East and FA-1f  in $d\geq 2$}

In the previous sections we have considered one dimensional East and FA-1f models. Both models can be    extended  to higher dimensions in a very natural way. Let us set some notation. Let $d>1$, then $\Lambda_{N}^d:=[1,N]^d$ and $\Omega_N^d:=\{0,1\}^{\Lambda_{N}^d}$ and let $\vec e_j$ with $j\in[1,d]$ be the Euclidean basis vectors. The East and FA-1f models in dimension $d$ at density $\rho\in(0,1)$ are Glauber type Markov processes with generator $\cL_N^d$ which acts on $f:\Omega_N^d\to\bbR$ exactly as in \eqref{gene} with the sum over $i$ running now on $i\in
\Lambda_N^d$, with $c_i$ defined again as in \eqref{ci} and $r_i$ defined  for the East model as
\begin{equation}\label{ceastd}
r_i(\eta)=1-\eta_{i+\vec e_1} \qquad {\mbox{~if~}} i\cdot\vec e_1\in [1,N-1];\qquad r_i(\eta)=1 {\mbox{~otherwise~}}
\end{equation} 
and for FA-1f model with completely empty boundary conditions as
\begin{equation}\label{cFA-1fd}
r_i(\eta)=1-\prod_{j=1}^d\eta_{i+\vec e_j} \eta_{i-\vec e_j} \qquad {\mbox{~if~}} i\cdot\vec e_j\in [2,N-1]  \,\,\,\,\,\forall j\in [1,d];\qquad r_i(\eta)=1 {\mbox{~otherwise~}}.
\end{equation}

As in the one dimensional case, the above choice of the boundary condition is the only ergodic choice for the East constraints, while in the FA-1f case
any choice with at least one empty boundary site is ergodic.
%On can consider for example FA-1f with one empty boundary site namely 
%\begin{equation}\label{cFA-1fd}
%r_i(\eta)=1-\prod_{j=1}^d\eta_{i+\vec e_j} \eta_{i-\vec e_j} 
%\end{equation}
%where if $j\not\in\Lambda_N^d$ we let $\eta_{j}=0$ if $j=(N+1)\vec e_1+\vec e_2=0$ and otherwise
% $\eta_j=1$.
We will be interested here in all the choices of the boundary conditions which correspond to requiring a completely empty hyperplane of linear size $N$ and dimension $c$ with $c\in[0,d-1]$ ($c=0$ corresponds to the choice of a single empty boundary site).
%For each choice of the boundary condition we have a set of sites inside $\Lambda_N^d$ which are unconstrained which  we denote by $\cU_N$ (more precisely site $i$ belongs to $\cU_N$ iff for any $\eta\in\Omega_N$ it holds $r_i(\eta)=1$).
In short we will say that we consider a boundary condition of dimension $c$ in this case (note that the only ergodic choice for East corresponds to a particular boundary condition of dimension $d-1$). Note that as in the one dimensional case both dynamics satisfy detailed balance with respect to Bernoulli product measure $\nu$ at density $\rho$, namely $\nu(\eta):=\prod_{i\in\Lambda_N^d}\nu_i(\eta_i)$ with $\nu_i(1)=\rho$.
%Let  $c$ be such that $\lim_{N\to\infty} |\cU_N|/N^c=a>0$, then we say that the dimension of the unconstrained set is $c$. For example for East we get $c=d-1$, for FA-1f with completely empty boundary $c=d-1$, for FA-1f with one empty boundary $c=0$.

\medskip

Let $\cA(t)$ be defined as in \eqref{defactivity} where now the sum runs over all sites inside $\Lambda_N^d$. As for the one dimensional case it is immediate to verify that $\cA(t)$ satisfies the law of large number
$$\lim_{t\to\infty}\lim_{N\to\infty}\frac{\cA(t)}{N^d t}=\bbA$$ with $\bbA$ defined as in \eqref{bbA} for $j\in\Lambda_N^d$ with $i\cdot\vec e_j \in[2,N-1]$ for all $j\in [1,d]$. Note that $\bbA$ coincides with the one for the corresponding one dimensional model at the same density.
Let $\gp^{(N)}$ be defined as in \eqref{eq:free}   and let the rescaled generating function $\gp_d(\alpha)$ and the measure $\mu_{\alpha,T}^{N,d}$ be defined as
\begin{equation}
\label{rescaledd}
\gp_d(\alpha):=\limsup_{N\to\infty}\frac{1}{N^{c}}\gp^{(N)}\left(\frac{\alpha}{N^{d-c}}\right)
\end{equation}
\begin{equation}
\label{mesureh}
\mu^{N,d}_{\alpha,T}:=\frac{\bk{\cdot\exp\left(\frac{\alpha}{N^{d-c}}\cA(T)\right)}}{\bk{\exp\left(\frac{\alpha}{N^{d-c}}\cA(T)\right)}}.
\end{equation}
Note that the above definitions if we set $d=1$ and $c=0$  are compatible with the definitions used in Section \ref{models} for the one dimensional case, namely $\gp_1(\alpha)=\gp(\alpha)$ and $\mu^{N,1}_{\alpha,T}=\mu^{N}_{\alpha,T}$ with $\gp(\alpha)$ and $\mu_{\alpha,T}^N$ defined by equations \eqref{rescaled} and \eqref{eq: measure} respectively.
We stress that the finite size scaling depends on the choice of the boundary conditions. We expect this to be the choice which leads to a phase transition for the generating function, as in the one dimensional case.  This conjecture, as will be further clarified by the proof of Theorem \ref{poor}, is related  to the fact that  in order 
to have no activity for a $d$-dimensional model
a possible strategy is to start at time zero from a completely filled configuration and prevent all the $O(N^c)$ sites which are in contact with the boundary empty set from flipping.

\medskip

%\begin{rem}\label{remeast}
The $d$-dimensional East model corresponds to $N^{d-1}$ independent one dimensional East models so that $\gp_d(\alpha)=\gp_1(\alpha)$ for any $d$.  Thus the phase transition results of Theorem \ref{teo:phasetrans} hold for $\gp_d$ and 
Theorem \ref{teo:condmes} applies as well in this case.
%\end{rem}

For FA-1f it is not immediate to generalize the one dimensional results since the model cannot be decoupled into independent one dimensional FA-1f models.
In this case we prove
\begin{Theorem}
\label{linearityd>1}
Consider FA-1f model in dimension $d\geq 2$ with boundary condition of dimension $c$ with $c\in[0,d-1]$.  There exists $\alpha_0<0$ and a sequence $\gamma_N$ with $\lim_{N\to\infty}\gamma_N=0$
such that for $\alpha>\alpha_0$
\begin{equation*}
\gp_d(\alpha)=\bbA\alpha \, ,
\end{equation*}

\begin{equation*}
\lim_{N\to\infty}\lim_{T\to\infty}\mu_{\alpha,T}^{N,d}\left(\int_{0}^T dt |\sum_{i\in\Lambda_N^d}\eta_i(t)-N^d\rho|\leq \gamma_N N^d T \right)=1\, ,
\end{equation*}

\begin{equation*}
\lim_{N\to\infty}\lim_{T\to\infty}\mu_{\alpha,T}^{N,d}\left(\int_{0}^T dt |\sum_{i\in\Lambda_N^d}c_i(\eta(t))-N^d\bbA|\leq \gamma_N N^d T \right)=1\, .
\end{equation*}
\end{Theorem}

\begin{Theorem}
\label{fluctu2}
Consider East or FA-1f model in $d\geq 2$ at any $\rho\in(0,1)$.  For any $u \in[0,1]$ it holds
\begin{eqnarray}
\label{eqfluc2}
-(1-\rho)(1-u) &\leq& 
 \liminf_{N\to\infty} \lim_{t\to\infty} \; 
 \frac{1}{t N^c} \log  \left \bra \frac{\cA(t)}{N^dt}\simeq u \bbA \right \ket 
 \\
&\leq&
 \limsup_{N\to\infty} \lim_{t\to\infty} \; 
 \frac{1}{t N^c} \log  \left \bra \frac{\cA(t)}{N^dt}\simeq u\bbA \right \ket 
\leq \alpha_0 \bbA (1-u) \, ,
\nonumber
\end{eqnarray}
where $\ga_0<0$ was introduced in Theorem \ref{linearityd>1}.
\end{Theorem}

These results correspond to those of Theorems \ref{teo:phasetrans}(i) and \ref{teo:condmes} (i) and \ref{fluctu1} in the one dimensional case. For the large negative $\alpha$ regime, the result for dimension larger than one is much less precise.
\begin{Theorem}
\label{poor}
Consider  FA-1f model in dimension $d\geq 2$ with boundary condition of dimension $c\in[0,d-1]$.   For any $\delta>0$ there exists $\alpha_1(\delta)<0$
such that 
\begin{equation}\lim_{N\to\infty}\lim_{T\to\infty}\mu_{\alpha,T}^{N,d}\left(\int_{0}^T dt \sum_{i\in\Lambda_N^d}\eta_i(t)\geq (1-\delta) N^d T \right)=1  \, .
\end{equation}
\end{Theorem} 
%Let 
  Theorem \ref{linearityd>1} and  \ref{poor} will be proven in Section \ref{genesmall} and  \ref{secpoor} respectively. 
%if we define
%$$\gp_d(\alpha):=\lim_{N\to\infty}\lim_{T\to\infty}\frac{1}{TN^{d-1}}\log\bk{\exp(\alpha\cA(T)/N)}$$
%we get immediately $\gp_d(\alpha)=\gp(\alpha)$ and the results for this case follow immediately from the one dimensional results.
%
%Inste

\section{Preliminary results}

In Section \ref{var} we recall some basic tools from the Donsker-Varadhan large deviation theory  and  we  establish a variational formula for $\gp^{(N)}(\lambda)$ \eqref{eq: exact}. We underline that this formula is valid in any dimension. Finally, in Section \ref{subsec: sg} we recall a result on the spectral gap of the generator for KCM which will be used in some of our  proofs.

\subsection{Donsker-Varadhan theory}
%{A variational formula for $\gp^{(N)}(\lambda)$ and the positivity of the spectral gap}
\label{var}

%ricordardsi di dire che se non e' specificato il boundary e'...
%\section{Basic tools}

%unificare la notazione \lim T o \lim t.
The Donsker-Varadhan theory for large deviations \cite{DZ} will be a basic tool to derive our results.
Fix the dimension $d$ and let $\mathcal{D}_N$ be the Dirichlet form corresponding to the generator $\cL^d_N$ \eqref{gene}
which is defined on any function $g$ as 
\begin{equation}
\label{dirichlet}
\cD_N({g}) := -  \nu ( g \, \cL_N^d g). 
\end{equation}
For future use we note that
the Dirichlet form  can be rewritten by using the definition \eqref{gene} as
\begin{eqnarray}
\label{eq: dirichlet}
\cD_N({g}) = 
\sum_{i\in \gL_N^d}
\nu \left( c_i  (\eta)  \big( {g(\eta^i)} - {g(\eta)} \big)^2 \right)=\sum_{i\in \gL_N^d}
\nu \left( r_i(\eta){\mbox{Var}}_i(g)\right)
\end{eqnarray}
with $\mbox{Var}_i(g)=\nu_i(f-\nu_i(g))^2$.

For any smooth function $V:\Omega_N^d\to \mathbb{R}$,  we define the time average of $V$ over the process as
$$\pi_t(V):=\frac{1}{t}\int_0^t V(\eta(s))ds$$
where  $\eta(s)$ is a trajectory of the Markov process starting at time zero from $\eta$. 
The dynamics is reversible with respect to the measure $\nu$, thus the Donsker-Varadhan theory asserts that 
for any $\gamma\in\mathbb{R}$ 
\begin{equation}
\label{dv}
\lim_{t\to\infty} \; \frac{1}{t} \log \left \bra \exp \big(\gamma \, t\pi_t(V) \big) \right \ket
=
\sup_f \Big \{ \gamma\nu(fV) -  \cD_N (\sqrt f) \Big\} \, ,
\end{equation}
where we recall that $\bra \cdot \ket$ is the mean over the evolution of the process and over the initial configuration which is distributed with 
the equilibrium Bernoulli measure $\nu$ and the supremum is over the positive functions $f$ which satisfy $\nu(f)=1$. 
Note that the r.h.s of \eqref{dv} corresponds to the  largest eigenvalue of the modified operator
$\mathcal{L} + \gamma V$. 
Furthermore, if one defines the empirical measure in $\Omega_N^d$  by setting for any $A\subset\Omega_N^d$ and any $t \geq 0$
\begin{equation}
\label{eq: emprirical measure}
\pi_t(A) = \frac{1}{t} \int_0^t ds \;  \mathds{1}_{A} (\eta(s)),
\end{equation}
Donsker-Varahdan theory establishes that the large deviation functional of the empirical measure is the Dirichlet form $\cD_N$.
Thus if we let $\psi$ be any function from 
$\Omega_N^d\to \mathbb R$, for any $[a,b]\subset \mathbb R$ it holds
\begin{equation}
\lim_{t\to\infty} \frac{1}{t}\log \left \bra \frac{1}{t} \int_0^t ds \psi(\eta(s))\in[a,b]  \right \ket 
=
\lim_{t\to\infty} \frac{1}{t}\log \bra \pi_t(\psi)\in[a,b] \ket
= -\inf_{g: \nu(g)=1,~~ g\geq 0\atop \nu(g \psi ) \in [a,b] } \cD_N ({\sqrt g}) \, .
\label{DV2}
\end{equation}

%perche' c'era il minore ??
%where $mathds{1}_{\mathcal{M}}$ is the characteristic function of the set ${\mathcal{M}}$ (namely $\bramathds{1}_{\mathcal{M}}\pi_T\ket$ is the probability that the empirical measure belongs to this set).

%$\nu:$ path tali che $1/T\int dt \sum\eta_i(t)>c$. allors $<mathds{1}_{\nu}>$=Proba path tali che $\pi_t(\sum\eta_i)>c$.
%Allora minimizzo $-D(\sqrt g)$ sulle g tali che $\mu(gmathds{1}_{\sum\eta_i>c})=1$

%\section{Preliminary results}
%
%
%\subsection{The modified process and a variational formula for $\gp_N (\gl)$}
%
%
%We are now ready to derive the variational formula for $\gp^{(N)} (\gl)$.

For any $\lambda\in\bbR$, we consider the  modified dynamics obtained by rescaling the time by $\exp(\lambda)$. The generator reads
\begin{eqnarray}
\cL_{N,\gl} f (\eta) = \sum_{i\in \gL_N^d}
\exp( \gl) c_i  (\eta)  \big( f(\eta^i) - f (\eta) \big) = 
\exp( \gl) \cL_N  f (\eta),
\end{eqnarray}
which is again reversible with respect to $\nu$.
Then by evaluating the Radon-Nykodim derivative $d \bbP_t/d \bbP_t^{\lambda}$ where $\bbP_t$ and $\bbP_t^{\lambda}$ denote respectively the 
 probability of the trajectory
up to time t for the process evolving under $\cL_N$  and $\cL_{N,\gl}$, we obtain 
\begin{eqnarray}
\label{eq: Radon Nykodim} 
\frac{d\mathbb P_t}{d\mathbb P_t^{\lambda}}=
\exp\left(- \gl \cA(t) + (\exp( \gl) - 1)  \int_0^t \, ds \;  \cH (\eta(s)) \right),
\end{eqnarray}
with
\begin{eqnarray}
\label{eq: H}
\cH (\eta) = \sum_{i\in \gL_N^d} c_i  (\eta) \, .
\end{eqnarray}
where we recall that $\cA(t)$ is the total activity up to time $t$ and $c_i$ has been defined in \eqref{ci}.
This implies in particular for any function $\psi(t):(\eta_{s})_{s\leq t}\to \mathbb R$
\begin{equation}
\label{cm}
\left \bra  \psi(t)\exp( \gl \cA(t)) \right \ket 
=
\left \bra \psi(t)\exp \left( (\exp( \gl) - 1)  \int_0^t \, ds \;  \cH (\eta(s)) \right) \right \ket _{\lambda}
\end{equation}
where here and in the following $< \cdot >_{\lambda}$ is 
the expectation over the modified dynamics, i.e. the mean over the initial configuration distributed with $\nu$ and over the evolution of the process with generator $\cL_{N, \gl}$. 
Note that  \eqref{cm} together with definition \eqref{eq:free} and formula \eqref{dv} ($\cL_{N, \gl}$ is reversible with respect to $\nu$ thus Donsker-Varadhan theory applies), leads to the following variational formula
\begin{eqnarray}
\gp^{(N)} (\gl) &=& \sup_f \left\{ (\exp( \gl) - 1)  \nu (f \, \cH) - \exp(\gl)  \cD_N(\sqrt{f}) \right\} \label{eq: exact}
\end{eqnarray}
where the supremum is taken over the positive functions $f:\Omega_N \to\bbR$ such that  $\nu(f) = 1$.
\medskip

\subsection{Positivity of the spectral gap}
\label{subsec: sg}

Finally, another tool which we will use is the knowledge of a positive lower bound uniform on $N$ for the spectral gap of $\cL_{N}$ which is defined as
\begin{equation}
\label{defgap}
{\mbox{gap}}(\cL_N):=\inf_{f, f\neq const} \frac{\cD_N(f)}{\mbox{Var}(f)},
\end{equation}
where $\mbox{Var}(f)$ is the variance w.r.t. the invariant Bernoulli measure $\nu$ on $\Lambda_{N}$ and  the minimization is over the functions $f:\Omega_N\to\bbR$ which are not  constant.  The following holds

\begin{Proposition}[\cite{AD},\cite{CMRT1},\cite{CMRT2}]
\label{teogap}
Consider  East or FA-1f model in any dimension and with any choice of the boundary condition which guarantees ergodicity.
For any $\rho\in(0,1)$ there exists $S_{\rho}>0$ such that 
$$\inf_{N}{\mbox{gap}}({\cL}_N)>S_{\rho} \, .$$
From this result and Definition  \eqref{defgap}  it follows immediately that
\begin{equation}\label{spgapineq}
\cD_N(f)\geq S_{\rho}{\mbox{Var}}(f).
\end{equation}

Let $\cG$ be a generic connected subset of $\bbZ^d$ and consider FA-1f with a single empty site at the boundary on $\cG$. We call $\cL_{\cG}$  the generator and $\mbox{gap}(\cL_{\cG})$ the corresponding spectral gap defined as in \eqref{defgap} with $\cD_{N}$ substituted by $\cD_{\cG}(f):=-\nu(f,\cL_{\cG}f).$  Then 
$${\mbox{gap}}({\cL}_{\cG})>S_{\rho} \, .$$
\end{Proposition}
The result for the East model has been derived in  \cite{AD} and subsequently proven in  \cite{CMRT1} for a larger class of constraints including FA-1f in any dimension with completely empty boundary conditions. Actually, the latter result could be easily derived by the result in \cite{AD} without using the technique of \cite{CMRT1}. Instead, the result for FA-1f with generic boundary condition and on a generic graph has been derived in \cite{CMRT2}.

%\subsection{General bounds via Jensen inequality and using blocked configurations}
%***quyesti li tolgo li metto nelle sezioni corrispondenti**
%
%

\section{Local equilibrium}
\label{zeroblocks}

\subsection{East and FA-1f in $d=1$}
\label{zeroblocks_eastFA-1f}

% Furthermore Corollary \ref{densitycor} establishes that for East and FA-1f analogous results hold also for the mesoscopic labels corresponding to local densoty.

Throughout this section, we consider East and FA-1f models in one dimension with mean density $\rho\in(0,1)$.  
Let us start by defining the coarse grained activity.
Let $K$ be such that $N/K$ is integer and partition $\Lambda_N$ into boxes   $B_{i}$ with $i\in[1,N/K]$ of size $K$, namely $B_i=[(i-1)K+1,iK]$.
% with $i=(i_1,\dots i_d)$ and $i_i\in(0,\dots N/K-1)$ is defined as
%$\Lambda_K+\sum_{j=1}^d\vec e_j i_j$.  Let \begin{equation} 
%\mathcal{H}_i(\eta)=\sum_{i\in B_i}\frac{c_i(\eta)}{K}\end{equation}
 We define the activity in the interior of $B_i$  as \begin{equation}\label{defact}
\mathcal{H}_i(\eta)=\sum_{i\in \widetilde B_i}c_i(\eta),\end{equation}
where $\widetilde B_i\subset B_i$ are the sites such that the corresponding constraints depend
only on the configuration inside $B_i$, namely $\widetilde B_i:=B_i\setminus iK$ for East and $\widetilde B_i:=B_i\setminus \{(i-1)K+1,iK\}$ for FA-1f.  
%Fix $\epsilon>0$, we
%define the {\sl activity labels} associated to a configuration $\eta$ as
%\begin{eqnarray}
%\label{see}
%\label{c}
%u_{K,\gep}^i (\eta) =
%\begin{cases}
%%1, \qquad   & {\rm if} \quad \hat \eta_K (i) = 1 \quad {\rm and} \quad \hat \eta_K (i+1) = 1 \, ,\\
%%-1, \qquad  & {\rm if} \quad  |\hat \eta_K (i) - \gr | \leq \gep \quad {\rm and} \quad |\hat \eta_K (i+1) - \gr | \leq \gep  \, ,\\
%%0,  \qquad  & {\rm otherwise} \, .
%1, \qquad   & {\rm if} \quad \eta_j= 1\,\,\qquad \forall j\in B_i  \, ,\\
%-1, \qquad  & {\rm if} \quad  |\mathcal{H}_{i} (\eta) - \bbA {|\widetilde B_i|}| \leq \gep \, ,\\
%0,  \qquad  & {\rm otherwise} \, .
%\end{cases}
%\label{activitylabels}
%\end{eqnarray}
%In order for the above definition to be well posed we
%restrict  $\epsilon$ to the values s.t. $\bbA{|\widetilde B_i|}-\epsilon>0$ (we will always imply this restriction in the following).
%
 We also define the coarse grained density as
\begin{equation}
{\mathcal{R}}_i(\eta)=\sum_{i\in B_i}\frac{\eta_i}{K}.\end{equation}
Fix $\epsilon>0$, we
define the {\sl activity-density} %and  the {\sl density labels} 
associated to a configuration $\eta$ as
\begin{eqnarray}
\label{activity densitylabels}
 u_{K,\gep}^i (\eta) =
\begin{cases}
1\qquad   & {\rm if} \quad  \eta_j= 1\qquad\forall j\in B_i  \, ,\\
-1 \qquad  & {\rm if} \quad  |\mathcal{H}_{i} (\eta) - \bbA {|\widetilde B_i|} \; | \leq \gep K  {\mbox{ and }} |\mathcal{R}_{i} (\eta) - \rho | \leq \gep\, ,\\
0  \qquad  & {\rm otherwise} \, 
\end{cases}
\end{eqnarray}
where  the mean instantaneous activity $\bbA$ has been defined in formula \eqref{bbA} and 
in order  for the above definition to be well posed we
restrict $\epsilon$ to the values such that $\rho+\epsilon<1$ and $\bbA|\widetilde B_i|-\epsilon K >0$.
In the rest of the paper in any result that uses  these labels we imply that $N$ and $K$ are integers chosen in order that $N/K$ is also integer and that the above restriction on $\epsilon$ is satisfied.

%Note that from the above definitions it follows immediately that $u_{K,\gep}^i=1$ iff $\tilde u_{K,\gep}^i=1$ and
%$u_{K,\gep}^i=0$ implies $\tilde u_{K,\gep}^i=0$. Note also that both $u^i(\eta)$ and $\tilde u^i(\eta)$ depend only on the restriction of the configuration
%$\eta$ on the box $B_i$.

The main result of this section is Lemma \ref{prop: bad blocksnew} which states that the probability of finding under the empirical measure $\pi_T$ a density of boxes with activity-density label equal to zero, namely with activity different from the mean activity and from zero and/or particle density different from $\rho$ and one is suppressed exponentially in $T$ and $N/K$. In other words locally we are equilibrated either in the completely filled state or in the mean state identified by $\nu$. Instead Lemma \ref{lem: variance 0} guarantees that the probability of finding a density of boxes with activity-density label equal to one, namely completely filled boxes,  is suppressed exponentially in $T$ (but not wrt $N/K$).

Before stating and proving these results we give some inequalities which immediately follow from the above definition of the labels and
 which will be used in the subsequent sections.
Recall the definition of $\cH$ \eqref{eq: H}, then
%\eqref{defact} and \eqref{activitylabels} we notice that 
\begin{eqnarray}
\cH(\eta)\geq \sum_{i=1}^{N/K}\cH_i(\eta)& \geq & 
%\sum_{i=1}^{N/K} \mathds{1}_{u_{K,\epsilon} ^i  = -1}\left(\bbA|\widetilde B_i|-\epsilon\right)\nonumber\\ & & \geq 
\left(K\bbA -2\bbA-\epsilon K \right)\sum_{i=1}^{N/K} \mathds{1}_{u_{K,\epsilon} ^i  = -1}\label{usoH}
\end{eqnarray}
\begin{equation}
\cH(\eta)\leq \sum_{i=1}^{N/K}\cH_i(\eta)+2\frac{N}{K}\leq \sum_{i=1}^{N/K} \mathds{1}_{u_{K,\epsilon} ^i  = -1}\left(K\bbA -2\bbA+\epsilon K \right)+K\sum_{i=1}^{N/K} \mathds{1}_{u_{K,\epsilon} ^i  = 0}+2\frac{N}{K}.
\label{usoH2}
\end{equation}

%Recall the definition of $\cH$ \eqref{eq: H}, then
%%\eqref{defact} and \eqref{activitylabels} we notice that 
%\begin{eqnarray}
%\cH(\eta)\geq \sum_{i=1}^{N/K}\cH_i(\eta)& \geq & 
%%\sum_{i=1}^{N/K} mathds{1}_{u_{K,\epsilon} ^i  = -1}\left(\bbA|\widetilde B_i|-\epsilon\right)\nonumber\\ & & \geq 
%\left(\bbA K-\bbA 2-\epsilon\right)\sum_{i=1}^{N/K} \mathds{1}_{u_{K,\epsilon} ^i  = -1}\geq \left(\bbA K-\bbA 2-\epsilon\right)\sum_{i=1}^{N/K} \mathds{1}_{\tilde u_{K,\epsilon} ^i  = -1}
%\label{usoH}
%\end{eqnarray}
%\begin{equation}
%\cH(\eta)\leq \sum_{i=1}^{N/K}\cH_i(\eta)+2\frac{N}{K}\leq \sum_{i=1}^{N/K} \mathds{1}_{u_{K,\epsilon} ^i  = -1}\left(K\bbA -2\bbA+\epsilon\right)+K\sum_{i=1}^{N/K} \mathds{1}_{u_{K,\epsilon} ^i  = 0}+2\frac{N}{K}.
%\label{usoH2}
%\end{equation}
%\begin{equation}
%\cH(\eta)\leq \sum_{i=1}^{N/K}\cH_i(\eta)+2\frac{N}{K}\leq \sum_{i=1}^{N/K}\mathds{1}_{\tilde u_{K,\epsilon} ^i  = -1}\left(K\bbA -2\bbA+\epsilon\right)+K\sum_{i=1}^{N/K} \mathds{1}_{\tilde u_{K,\epsilon} ^i  = 0}+2\frac{N}{K}.
%\label{usoH2bis}
%\end{equation}
%
%
%

Recall the definition of the empirical measure $\pi_T$, see equation \eqref{eq: emprirical measure}, we define the following events:
\begin{definition}
Fix $T>0$, $N,K$ integers, $\delta\in[0,1]$ and $\epsilon>0$. Then for $j\in\{-1,0,1\}$ 
we define the event $\cW_{j,\delta}$  by requiring a density at least $\delta$ of $j$  activity-density labels.
In formulas $\cW_{j,\delta}$ is verified iff
$$\pi_T\left( \sum_{i =1}^{N/K} \, \mathds{1}_{ \big\{ u_{K,\epsilon} ^i  = j \big\} } \right) \geq \delta\frac{N}{K} \, .$$
%and $\widetilde  \cW_{j,a}$ is verified iff
%$$\pi_T\left( \sum_{i =1}^{N/K} \, \mathds{1}_{ \big\{ \widetilde u_{K,\epsilon} ^i  = j \big\} }\geq a\frac{N}{K}\right).$$ Trivially $\cW_{1,a}=\widetilde\cW_{1,a}$.
For any integer $\ell$ we also define  $V_{j,\ell}$ which is verified iff
$$\pi_T\left( \sum_{i =1}^{N/K} \, \mathds{1}_{ \big\{ u_{K,\epsilon} ^i  = j \big\}}\right)\in[\ell,\ell+1).$$

\label{defW}
\end{definition}

We stress that, even if it is not explicited for simplicity of notation, the event $\cW_{j,\delta}$ and $\cV_{j,\ell}$ depend on the choice of $T,N,K,\epsilon$.

\begin{lem}
\label{prop: bad blocksnew}  
%Consider East and FA-1f model in one dimension with $\rho\in[0,1)$.  
There  exists $C(\rho),C'(\rho)>0$ such that for any $\delta, \epsilon>0$, any $\lambda\in\bbR$ and any integer $N,K$ provided  $K>\bar K(\delta,\epsilon)=C'\frac{|\log(\delta)|}{\epsilon ^2}$ it holds 
\begin{eqnarray}
\label{comparinew3} 
\lim_{T \to \infty} \; \frac{1}{T} \log \bra \cW_{0,\delta}\ket_{\lambda}
\leq  - C\delta^2 \frac{N}{K} \exp(\lambda)\, .
\end{eqnarray}
%\begin{eqnarray}
%\label{comparinew2} 
%\lim_{T \to \infty} \; \frac{1}{T} \log 
%\left\bra\pi_T\left( \sum_{i=1}^{N/K-1}   \mathds{1}_{\{u_{K,\epsilon}^i\neq u_{K,\epsilon}^{i+1}\}}\right) \geq \delta\frac{N}{K}\right\ket_{\lambda}
%\leq  - C\delta^2 \frac{N}{K} \, .
%\end{eqnarray}
%\begin{eqnarray}
%\label{comparinew3} 
%\lim_{T \to \infty} \; \frac{1}{T} \log 
%\left\bra \widetilde\cW_{0,\delta}\right\ket_{\lambda}
%\leq  - C\delta^2 \frac{N}{K} \, .
%\end{eqnarray}
%\begin{eqnarray}
%\label{comparinew4} 
%\lim_{T \to \infty} \; \frac{1}{T} \log 
%\left\bra\pi_T\left( \sum_{i=1}^{N/K-1}   \mathds{1}_{\{\widetilde u_{K,\epsilon}^i\neq \widetilde u_{K,\epsilon}^{i+1}\}}\right) \geq \delta\frac{N}{K}\right\ket_{\lambda}
%\leq  - C\delta^2 \frac{N}{K} \, .
%\end{eqnarray}
%
\end{lem}

\begin{lem}
\label{lem: variance 0}
%Consider East and FA-1f model in one dimension at density ${\rho}\in[0,1)$.
For any $\delta>0$,   there is $K(\gd)$ such that for any $K \geq  K(\gd)$, for any $\lambda\in\bbR$, any $N\geq K$ and any $\epsilon>0$ it holds
%and $\ell \geq \gd \frac{N}{K}$, 
\begin{eqnarray}
\label{eq: scaling: 4.7}
\lim_{T\to \infty} \frac{1}{T} \log \bk{ \cW_{1,\delta}}_{\lambda} \leq - \frac{S_\gr\delta }{4} \exp(\lambda)\, ,
\end{eqnarray}
where $S_\gr>0$ has been defined in Proposition \ref{teogap}.
\end{lem}

\begin{remark}
The blocking mechanism described in Section \ref{heuristics} implies
that  
\begin{eqnarray}
\lim_{T\to \infty} \frac{1}{T} \log \bk{ \cW_{1,\delta}}_{\lambda} \geq - (1-\rho)  \exp(\lambda)\, ,
\end{eqnarray}
thus the scaling in \eqref{eq: scaling: 4.7} cannot be improved.
\end{remark}

As a consequence of Lemma \ref{prop: bad blocksnew}, we will see that
\begin{lem}\label{bad}
%Consider East or FA-1f in one dimension at density $\rho\in[0,1)$.
There exists $C(\rho)>0$ s.t. for any $\epsilon,\delta>0$, any $\alpha\in\bbR$ and any $K\geq\bar K$ with $\bar K$ specified in Lemma \ref{prop: bad blocksnew} and any $N\geq K$ it holds
\begin{eqnarray}
\label{eq: bad blocks bound}
\lim_{T \to \infty} \; \frac{1}{T} \log \mu_{\alpha,T}^{N}
\left(\cW_{0,\delta}
\right)  
 \leq  -C \delta^2\frac{N}{K}\exp(\frac{\alpha}{N})+|\alpha|(1+\frac{C}{N})\, .
\end{eqnarray}
\end{lem}

\medskip

We start by proving a preliminary result. For $i\in [1,N/K]$, let $\mathcal{Z}_i$ be the event  that  there exists at least one empty site inside the box $B_i$ and $\cE_i$ 
be the event that is verified iff $ u^i_{K,\epsilon}\in\{1,0\}$ and define the function $V:\Omega_N\to\bbR$ as
\begin{equation}\label{defVeta}
V(\eta):=\sum_{i =1}^{N/K-1} \, \mathds{1}_{\mathcal{E}_i } (\eta)\mathds{1}_{\mathcal{Z}_{i+1}}(\eta)+\mathds{1}_{\mathcal{E}_{N/K} } (\eta) \, .
\end{equation}

%Let us start by proving a result from which
%Proposition \ref{prop: bad blocksnew} will easily follow.

%For $i\in [1,N/K]$ we let $\mathcal{Z}_i$ be the event  that  there exists at least one empty site inside the box $B_i$ and $\cE_i$  be an event which depends only on the restriction of the configuration to $B_i$. If $\cE_i$ are such that $\nu(\cE_i)=\nu(\cE_j)$that is verified
%iff $u^i\neq -1$ and $\tilde\cE_i$ namely
%$mathds{1}_{\mathcal{Z}_i}=1-\prod_{j\in B_i}\eta_j$, 
%$\mathcal{H}_i(\eta)\neq 0$, namely $u^i_{K,\epsilon}\in(-1,0)$. Note that  definition \eqref{defact} guarantees that this event depends only on the restriction of the configuration to $B_i$ (this would not have been true of the sum in the definition if the activity was performed over all the sites of $B_i$).
%For any event
% $\mathcal{E}\subset \Omega_K^d$  
%  define the event $\mathcal{E}_i$ corresponding to the configuration restricted to box $B_i$, namely we let  $\eta\in\mathcal{E}_i$ iff the restriction to $B_{0}$
%of $\theta_{-i}\eta$ belongs to $\mathcal{E}$ where $\theta_j$ is the  operator which translates the configuration of $j$, $\theta_j\eta_x=\eta_{x+j}$.
%
\begin{lem}
\label{prop: bad blocks} 
%Consider East or FA-1f in one dimension at density $\rho\in[0,1)$ and 
Set $m:=\nu(\cE_i)$.
There exists a constant $C=C(\rho)>0$ such that uniformly in  $N$ for any $\lambda\in\bbR$ it holds
\begin{eqnarray}
\label{compari} 
\lim_{T \to \infty} \; \frac{1}{T} \log 
\left \bra\pi_{T}(V) \geq x + \mm \frac{N}{K} \right\ket_{\lambda}
 \leq  - C \frac{K}{N} x^2 \exp(\lambda) \, .
\end{eqnarray}
\end{lem} 

\begin{proof}

% {\sl uso chebichev esponenziale prioma di usare DV per poter ricostruire delle varianze tante su cui ottengo ordine N errori...}
%We make the proof in one dimenision valid for East and for strongly non cooperative. Note that both for East and for FA-1f (the latter being a case of strongly non coop) $Z_i$ is the event that there is a zero. For simplicity of notation we let $i+\vec e_1=i+1$.
 
For any $\gga >0$
\begin{eqnarray}
\label{eq: chebyshev}
\left\bra\pi_T(V) \geq x + \mm \frac{N}{K}\right\ket_{\lambda}
\leq
\exp \left( - T \gga \left( x + \mm \frac{N}{K} \right) \right) \bk{ \exp \left( \gga \int_0^T dt \,  V(\eta(t)) \right) }_{\lambda}  \, .
\end{eqnarray}
%forse DV da' nei due versi fk solo upper bound.
Then using \eqref{dv} we get
\begin{eqnarray}
\label{fey} 
\lim_{T \to \infty} \frac{1}{T}\log \; \left\bra\pi_T(V)  \geq x +\frac{N}{K}\right\ket_{\lambda}
 \leq -\gga \left( x + \mm \frac{N}{K} \right)+ \sup_{f} \left\{ \gga  \nu \big( f V \big) - \exp(\lambda) \cD_N (\sqrt{f}) \right \},\nonumber\\
 \end{eqnarray}
%\begin{eqnarray}
%\label{eq: eigenvalue}
%\gG = \inf_{f} \left\{ \gga  \nu \big( f V \big) - \exp(\lambda) \cD_N (\sqrt{f}) \right \}  \, ,
%\end{eqnarray}
where the supremum is over the $f:\Omega_N\to\bbR$ such that $\nu(f) =1$ and $f \geq 0$.
%$\nu(fV)=\sum_{i=1}^{N/K-1}\nu(f mathds{1}_{\mathcal{C}_i } (\eta)mathds{1}_{\mathcal{Z}_{i+1}}(\eta))$
%and 

Notice that
\begin{equation}
\label{divido}
\nu\big (f \mathds{1}_{\mathcal{E}_i }\mathds{1}_{\mathcal{Z}_{i+1}}\big )=\nu\big (\mathds{1}_{\mathcal{E}_i } \mathds{1}_{\mathcal{Z}_{i+1}}\nu_{B_i}(\sqrt f)^2\big )+
\nu\left (\mathds{1}_{\mathcal{E}_i } \mathds{1}_{\mathcal{Z}_{i+1}}\left[f-\nu_{B_i}(\sqrt f)^2\right]\right ).
\end{equation}
The first term in \eqref{divido} can be bounded from above by
\begin{equation}
\label{impero1}
\nu\big (\mathds{1}_{\mathcal{E}_i }\mathds{1}_{\mathcal{Z}_{i+1}}\nu_{B_i}(\sqrt f)^2\big )\leq 
\nu\big (\mathds{1}_{\mathcal{E}_i } \nu_{B_i}(\sqrt f)^2\big )\leq \nu(f)m=m \, ,
\end{equation}
where we use the fact that $\nu_{B_i}(\sqrt f)^2$ does not depend on the variables inside $B_i$, $\mathcal{E}_i$ does not depend on the variables outside $B_i$ and the fact that 
$\nu(\nu_{B_i}(\sqrt f)^2)\leq\nu( \nu_{B_i}(f))=\nu(f)=1$.

On the other hand for  the second term in \eqref{divido}  we have
\begin{eqnarray}\label{impero2}
\nu\left (\mathds{1}_{\mathcal{E}_i } \mathds{1}_{\mathcal{Z}_{i+1}}\left[f-\nu_{B_i}(\sqrt f)^2\right]\right ) & \leq 
\nu\left (\mathds{1}_{\mathcal{Z}_{i+1}}(\eta)\left[\sqrt f-\nu_{B_i}(\sqrt f)\right]^2\right )^{1/2}\nu\left (\left[\sqrt f+\nu_{B_i}(\sqrt f)\right]^2\right )^{1/2}\nonumber\\
\leq 2~\nu\left (\mathds{1}_{\mathcal{Z}_{i+1}}\mbox{Var}_{B_i}(\sqrt{f})\right )^{1/2} \, ,
\end{eqnarray}
where to obtain the first inequality we upper bound $\mathds{1}_{\mathcal{E}_i }$ by one and  we 
use Cauchy-Schwartz while for the second inequality we use 
the fact that the event $\mathcal {Z}_{i+1}$ depends only on the variables inside $B_{i+1}$, thus it is independent on the variables in the block $B_{i}$.
Then we notice that $\mathds{1}_{\mathcal{Z}_{i+1}}=1$ guarantees the existence of (at least) one zero inside $B_{i+1}$ and we let $\xi$ be the position of the first zero starting from the right border of this box. Thus 
%(recall that box $B_{i+1}$ contains sites $[Ki,K(i+1)-1]$)
\begin{equation}\label{us1}
\nu\left (\mathds{1}_{\mathcal{Z}_{i+1}}\mbox{Var}_{B_i}(\sqrt{f})\right )=\sum_{j=1}^{K}\nu\left ( \mathds{1}_{\xi=j+iK}\mbox{Var}_{B_i}(\sqrt{f})\right ) \, ,
\end{equation}
and by letting $L_{i,j}$ and $R^{i,j}$ be the subset of $B_i\cup B_{i+1}$ to the left (respectively right) of $j+iK$, namely
$L_{i,j}:=B_i\cup [iK,\dots, j+iK-1]$ and $R_{i,j}=B_{i+1}\setminus L_{i,j}$ we have
\begin{eqnarray}
\nu\left ( \mathds{1}_{\xi=j+iK}\mbox{Var}_{B_i}(\sqrt{f})\right )
\leq \sum_{\eta^r} \nu(\eta^r)\mathds{1}_{\xi=j+iK}\sum_{\eta^l}\nu(\eta^l)\mbox{Var}_{B_i}(\sqrt{f})\nonumber\\
\leq \sum_{\eta^r} \nu(\eta^r)\mathds{1}_{\xi=j+iK}\mbox{Var}_{L_{i,j}}(\sqrt{f}) \, ,
\end{eqnarray}
where $\eta^r$ ($\eta^l$) is the configuration restricted to $B^{jr}$ ($B^{jl}$) and we use the product form of $\nu$ and, in the last passage,  the convexity of the variance and the fact that $B_i \subset B^{jl}$.
Then  by using the spectral gap inequality \eqref{spgapineq} for the model on $L_{i,j}$ with a frozen zero at the right boundary  together with the expression of the Dirichlet form \eqref{eq: dirichlet} we get (recall that $S_{\rho}>0$ is the lower bound on the infimum over $N$ of the spectral gap of $\Lambda_N$ at density $\rho$)
\begin{equation}
\label{us}
\mbox{Var}_{L_{i,j}}(\sqrt{f})\leq S_{\rho}^{-1} \sum_{\eta^l}\nu(\eta^l)(\sum_{x\in L_{i,j}} r_x^l(\eta^l) \mbox{Var}_x(\sqrt{f})) \, ,
\end{equation}
 where 
 we denote by $r_x^l$ the constraints for the model on $L_{i,j}$ with empty boundary condition on the right boundary, namely $r_x^l(\eta)=1$ if $x=j+iK-1$ and otherwise
$r_x^l(\eta^l)=1-\eta^l_{x+1}$ if we are considering East or $r_x^l(\eta^l)=1-\eta^l_{x+1}\eta^l_{x-1}$ if we are considering FA-1f.
Then we note that for any $\eta$ such that  $\xi(\eta)=j+iK$ and which equals $\eta^l$ on $B^{jl}$, it holds $c_x^l(\eta^l)=c_x(\eta)$ for any $x\in L_{i,j}$.
Thus
we can insert \eqref{us} into \eqref{us1} and use this observation to get
\begin{eqnarray}\label{impero3}
\nu\left ( \mathds{1}_{\mathcal{Z}_{i+1}} \mbox{Var}_{B_i}(\sqrt{f})\right )
\leq \frac{1}{S_{\rho}}\nu \left( \sum_{x\in B_i\cup B_{i+1}} r_x \mbox{Var}_x(\sqrt{f}) \right) \, .
\end{eqnarray}

Then \eqref{divido}, \eqref{impero1}, \eqref{impero2} and \eqref{impero3} yield
\begin{equation}
\nu\big (f \mathds{1}_{\mathcal{E}_i }\mathds{1}_{\mathcal{Z}_{i+1}}\big )
\leq m+2\sqrt{\frac{\nu(\sum_{x\in B_i\cup B_{i+1}} r_x \mbox{Var}_x(\sqrt{f}))}{S_{\rho}}}\end{equation}
and for $\nu\big (f \mathds{1}_{\mathcal{E}_{N/K}} (\eta)\big )$ the same upper bound can be obtained along the same lines (actually easily because the boundary condition guarantees a zero at the right border of $B_{N/K}$).
Thus
for any function $f$ s.t. $
\nu(f)=1$ and $f>0$ 
it holds
\begin{eqnarray}
& \gamma\nu(fV)-\exp(\lambda)\cD(\sqrt f)\leq \gamma m \frac{N}{K}+ \sum_{i=1}^{N/K}\left[\frac{2\gamma \sqrt 2 }{\sqrt S_{\rho}} \sqrt {\cD_{K,i}(\sqrt f)}-\exp(\lambda)\cD_{K,i}(\sqrt f)\right] \\
& = \frac{N}{K}\left[\gamma m+\frac{ 2 \gamma^2}{ S_{\rho}\exp(\lambda)} -\frac{K}{N}\sum_{i=1}^{N/K}\left(\exp(\lambda/2)\sqrt {\cD_{K,i}(\sqrt f)}-\frac{\sqrt 2\gamma }{\sqrt S_{\rho}\exp(\lambda/2)}\right)^2\right]\leq \frac{N}{K}\left[\gamma m+\frac{ 2\gamma^2 }{ S_{\rho}\exp(\lambda)}\right]\nonumber
\end{eqnarray}
with
$\cD_{K,i}(\sqrt f)$
the contribution to the Dirichlet form coming from the sites in the box $B_i$, namely
\begin{eqnarray*}
\cD_{K,i} \left( \sqrt{f }   \right) = \nu \left( \sum_{j \in B_i}
 c_j  (\eta)  \big( \sqrt{f(\eta^j)} - \sqrt{f (\eta)} \big)^2  \right) \, .
\end{eqnarray*}

Then by using \eqref{fey} and optimizing over $\gamma$ we get
\begin{equation}
\lim_{T\to\infty}\frac{1}{T}\log\left\bra
 \frac{1}{T} \int_0^T dt \,  V(\eta(t))  \geq x + \mm \frac{N}{K}\right\ket_{\lambda}
 \leq - \frac{K}{N}\frac{S_{\rho}\exp(\lambda)}{8}x^2 \, .
 \end{equation}
 This completes the proof.
 \end{proof}

% specificare che non dipende dal sito a parte bordo.
%We define the mean probability of the 0-labels
%\begin{equation}
%\label{eq: m}
%\mm =\nu( \big\{ u_K (\eta(t),i) = 0 \big\} \big)
%\end{equation}
%which is independent of  $i \in \gL_{N,K}$ and which can be made arbitrarily 
%small for any $\gep>0$ by choosing $K$ large enough. 
%

We are now ready to prove the main results of this section.

\begin{proof} [Proof of Lemma \ref{prop: bad blocksnew} ]
We recall the definition \eqref{defVeta} for the function $V$, where $\mathcal{E}_i$ is the event which is verified iff $u^i_{K,\gep} \in\{0,1\}$. Thus from the definition of the activity-density labels it follows immediately that
 the probability of $\mathcal{E}_i$ goes to zero as the size of the box, $K$, goes to infinity and it is bounded from above by $\exp(-K\epsilon^2 C)$.
 Thus provided $K\geq \bar K$ it holds $\nu(\mathds{1}_{\mathcal{E}_i})<\delta/6$.
 Thanks to these facts  we can apply Lemma \ref{prop: bad blocks} with the choice $x=\delta N/(6K)$ to obtain that
 \begin{eqnarray}
\label{compari6} 
\lim_{T \to \infty} \; \frac{1}{T} \log 
\left \bra\pi_{T}\left(V\right) \geq  \frac{\delta N}{3K} \right\ket_{\lambda}
 \leq  - C \frac{N}{K} \delta^2 \exp(\lambda) \, ,
\end{eqnarray}
 % and $\nu(mathds{1}_{\mathcal{D}_i})<\delta/8$;\\
 where $V$ is defined in \eqref{defVeta}.
 We will now prove that the following inequality holds for any $\eta\in\Omega_N$
 \begin{equation}
 \label{soloq}
\sum_{i=1}^{N/K}
\mathds{1}_{u_{K,\epsilon}^i=0}(\eta)\leq V(\eta) \, .
\end{equation}
%\begin{equation}
%\sum_{i=1}^{N/K-1}
%mathds{1}_{v_{K,\epsilon}^i=0}\leq
%\sum_{i=1}^{N/K-2}mathds{1}_{\mathcal{D}_i}mathds{1}_{\mathcal{Z}_{i+1}}+mathds{1}_{\mathcal{D}_{N/K-1}}
%\end{equation}
% \begin{equation}
%  \label{soloq2}
%\frac{1}{3}\sum_{i=1}^{N/K-2}
%\mathds{1}_{u_{K,\epsilon}^i\neq u_{K,\epsilon}^{i+1}}\leq
%\sum_{i=1}^{N/K-2}\mathds{1}_{\mathcal{E}_i}\mathds{1}_{\mathcal{Z}_{i+1}}+\mathds{1}_{\mathcal{E}_{N/K-1}}
%\end{equation}
Then collecting \eqref{compari6} and \eqref{soloq}  the proof of \eqref{comparinew3} is completed.
%Inequality \eqref{comparinew4} follows instead from \eqref{compari6} and \eqref{soloq2}. The result \eqref{comparinew} follows from \eqref{comparinew3} by using the fact that $ u^i_{K,\epsilon}=0$ implies $\tilde u^i_{K,\epsilon}=0$ and the result \eqref{comparinew2} follows from \eqref{comparinew4} and the fact that if $u_{K,\epsilon}^{i}\neq u_{K,\epsilon}^{i+1}$ and both are different from zero then also $\tilde u_{K,\epsilon}^{i}\neq \tilde u_{K,\epsilon}^{i+1}$.
 %\begin{equation}
%\frac{1}{2}\sum_{i=1}^{N/K-2}
%mathds{1}_{v_{K,\epsilon}^i\neq v_{K,\epsilon}^{i+1}}\leq
%\sum_{i=1}^{N/K-2}mathds{1}_{\mathcal{D}_i}mathds{1}_{\mathcal{Z}_{i+1}}+mathds{1}_{\mathcal{D}_{N/K-1}}
%\end{equation}
%\label{sc}
%\end{pro}
%\begin{proof}

\smallskip
We are therefore left with the proof of \eqref{soloq} which immediately follows from the following observation. Let $i<j$ be such that $u^i_{K,\epsilon}=u^j_{K,\epsilon}=0$ and $u^k_{K,\epsilon}\neq 0$ for all $k\in [i+1,j-1]$. Then there exists $k\in [i,j-1]$ such that $\mathds{1}_{\mathcal{E}_{k}}\mathds{1}_{Z_{k+1}}=1$. In order to prove this statement we consider separately  the case (a) $j=i+1$  and (b) $j>i+1$. In case (a) the result holds since $u^{i}_{K,\epsilon}=0$ implies $\mathds{1}_{\mathcal{E}_{i}}=1$ and $u^{i+1}_{K,\epsilon}=0$ implies $\mathds{1}_{\mathcal{Z}_{i+1}}=1$, thus $ \mathds{1}_{\mathcal{E}_{i}}\mathds{1}_{Z_{i+1}}=1$. In case (b)
 we distinguish subcases (b1) $u^{k}_{K,\epsilon}=1$ for all $k\in [i+1,j-1]$ and (b2) there exists at least one site $k\in [i+1,j-1]$ such that $u^{k}_{K,\epsilon}=-1$.
If (b1) holds then $ \mathds{1}_{\mathcal{E}_{j-1}}\mathds{1}_{Z_{j}}=1$
(since $u^{j-1}_{K,\epsilon}=1$ implies  $\mathds{1}_{\mathcal{E}_{j-1}}=1$ and $u^{j}_{K,\epsilon}=0$ implies $\mathds{1}_{Z_{j}}=1$) and the desired result is proven.
In case (b2) if we let $\ell$ be the smallest index in $  [i+1,j-1]$ such that $u^{\ell}_{K,\epsilon}=-1$ then $ \mathds{1}_{\mathcal{E}_{\ell-1}}\mathds{1}_{Z_{\ell}}=1$
(indeed $u^{\ell-1}_{K,\epsilon}\in\{0,1\}$ and therefore $\mathds{1}_{\mathcal{E}_{\ell-1}}=1$  and  $u^{\ell}_{K,\epsilon}=-1$ implies $\mathds{1}_{Z_{\ell}}=1$) and again the desired result is proven.

\end{proof}

\begin{proof} [Proof of Lemma \ref{lem: variance 0}]

By Donsker-Varadhan large deviation principle \eqref{DV2} and the spectral gap inequality \eqref{spgapineq}
 we get
\begin{eqnarray}
\label{princ}
\lim_{T\to \infty} \frac{1}{T} \log \bk{ \cW_{1,\delta} }_{\lambda} \leq -\exp(\lambda)\inf_{f}\left\{\cD_{N}(\sqrt f)\right\}\leq - \exp(\lambda)S_\gr \inf_f\left\{\mbox{Var}(\sqrt f)\right\} \, ,
\end{eqnarray}
where the infimum is over the positive functions $f$ such that 
\begin{equation}
\label{conditionf}
\nu(f)=1, \qquad \quad 
\nu \left( f\sum_{i=1}^{N/K}\mathds{1}_{\{u^i_{K,\epsilon}=1\}} \right)
\geq \delta\frac{N}{K} \, .
\end{equation}
We will now show that under the latter constraint 
\begin{equation}
\label{added}
\mbox{Var}(\sqrt f)  \geq \delta -\rho^K-\rho^{K/2}.
\end{equation}
Then  by choosing $K$ sufficiently large so that $\rho^{K}\leq \delta^2/4$ by collecting \eqref{princ} and  \eqref{added} 
the desired result follows (note that $\rho^K\leq \delta^2/4$ implies $\rho^K\leq \delta/4$ since we only have to deal with the case $\delta \leq 1$). We are therefore left with proving  that under conditions \eqref{conditionf} the inequality \eqref{added} holds.
\smallskip

For each box $B_i$ with $i\in[1,N/K]$ we define the coarse grained variable $\omega_i$ by 
$$\omega_i=\mathds{1}_{\{u_{K,\epsilon}^i=1\}},$$
and let $p:=\nu(u_{K,\epsilon}^i=1)=\rho^{K}$ and $m_p$ be the Bernoulli product measure with density $p$ on $\Omega_{N/K}$.
The marginal of $\nu(\eta) f(\eta)$ on the coarse grained variables is given by $m_p (\go) g(\go)$ with 
\begin{eqnarray}\label{defg}
g(\go): = \frac{1}{m_p(\go)} \; \sum_{\eta \sim \go} f(\eta) \nu(\eta)\, ,
\end{eqnarray}
where the sum is over the $\eta$'s compatible with $\go$. Then
\begin{eqnarray}
\label{eq: average variance}
\sqrt{g(\go)} =\sqrt{\sum_{\eta \sim \go} 
\frac{\nu(\eta)}{m_p(\go)} f(\eta)} \geq \sum_{\eta \sim \go}  \frac{\nu(\eta)}{m_p(\go)}
\sqrt{f(\eta)} \, ,
\end{eqnarray}
where we used the fact that for each fixed $\omega$ it holds $ \sum_{\eta \sim \go} 
\frac{\nu(\eta)}{m_p(\go)}=1$ and the concavity of the square root.
Then from \eqref{eq: average variance} we get
\begin{equation}
\label{eccola}
\mbox{Var}(\sqrt f)=1 - \nu ( \sqrt{f})^2  \geq 1 - m_p ( \sqrt{g})^2.
\end{equation}
Thus in order to prove \eqref{added} it is sufficient to show that  
\begin{equation}
\label{tobe}
m_p(\sqrt g)^2 \leq 1-\delta+p+\sqrt p \, ,
\end{equation}
when $g$ is defined as in \eqref{defg} and  $f$  satisfies conditions \eqref{conditionf} which imply
$$
m_p(g\sum_{i=1}^{N/K}\omega_i)\geq \delta N/K \, .
$$ 
 The latter inequality implies that  there exists 
(at least) a box  $B_j$ with $j\in[1,N/K]$ such that  $\gt_j: = m_p \left(g (\go) \go_j \right) \geq \delta$. 
Let $j$ be the rightmost box which verifies this constraint and rewrite each configuration $\omega$ via the couple $(\omega_j,\sigma)$ with $\sigma=  \{ \go_i \}_{i \not = j}$.  Thus
 \begin{equation}\label{mpg}
 m_p (\sqrt{g}) = p Z_1 + (1-p) Z_2
 \end{equation}
 where
 \begin{eqnarray*}
Z_1: = \sum_\gs  m_p^j(\gs) \sqrt{g (1,\gs)}
\quad {\rm and}  \quad Z_2 :=  \sum_\gs   m_p^j(\sigma) \sqrt{ g (0,\gs)} \, .
\end{eqnarray*}
where now $m_p^j$ denotes  the product measure with density $p$ on $\{1,\dots, N/K \}\setminus j$.
By Jensen inequality  
\begin{eqnarray}
\label{z12}Z_1^2  \leq  \sum_\gs  m_p^j(\gs) g(1,\gs) = \frac{\gt_j}{p}, 
\qquad 
Z_2^2  \leq  \sum_\gs   m_p^j(\gs) g (0,\gs) = \frac{1-\gt_j}{1-p} \, .
\end{eqnarray}
Thus \eqref{mpg} yields
\begin{eqnarray*}
m_p (\sqrt{g})^2 &=& p^2 Z_1^2 + (1-p)^2 Z_2^2 +  2 p(1-p) Z_1  Z_2\\
&\leq& p \gt_j + (1-p) (1-\gt_j) + 2 \sqrt{p (1-p)} \; \sqrt{\gt_j (1-\gt_j)} \\
&\leq& 1- \gt_j + p + \sqrt{p}\leq 1- \delta+ p + \sqrt{p}\, ,
\end{eqnarray*}
where we used \eqref{z12} for the first inequality, the fact that $\gt_j \leq 1$ thus $\gt_j (1-\gt_j) \leq 1/4 $ and that $p<1$ for the second inequality and finally the bound $\gt_j\geq \delta$ for the last inequality. Thus \eqref{tobe} is proven and the proof of the Lemma is concluded.
%where now $m_p^j$ denotes  the product measure on $\{1,\dots L \}\setminus j$.

\end{proof}

Before proving the last result of this section, Lemma \ref{bad}, we state separately a result
which will be also useful in other proofs.

\begin{lem}
\label{easy}
Consider East and FA-1f model in $d=1$.
There exists  $C>0$ s.t. for any $N$ if $\alpha<0$
\begin{eqnarray}
\label{b1}
\bk{\exp \left( \frac{\ga}{N} \cA(T) \right)} \geq \exp\left( \ga \bbA T \left(1+\frac{C}{N}\right)\right) \, 
\end{eqnarray}
 if $\alpha>0$
 \begin{eqnarray}
\label{b1}
\bk{\exp \left( \frac{\ga}{N} \cA(T) \right)} \geq \exp\left( \ga \bbA T  \left(1-\frac{C}{N}\right)\right) \, 
\end{eqnarray}
%\begin{eqnarray}
%\label{b2}
%\bk{\exp \left( \frac{\ga}{N^{a}} \cA(T) \right)} \geq \exp ( - |\mathcal{U}|T ) \gr^{N^d} ,
%\end{eqnarray}
%where $\bbA := \nu (c_i(\eta))$ is the mean istantaneous activity 
%on any site $i$ which does not belong to the interior boundary, $i\in\Lambda_N^d\setminus \mathcal{I}$.
\end{lem}
\begin{proof}
By Jensen inequality it holds
\begin{equation}\label{jensen}
\bk{\exp \left( \frac{\ga}{N} \cA(T) \right)} \geq \exp \left( \frac{\ga}{N} \bk{ \cA(T)} \right)\end{equation}
Then 
from the definition of the process and recalling \eqref{bbA} it follows immediately that \begin{equation}
<\cA(t)>=t\sum_{i=1}^N\nu(c_i(\eta))=t (N-2)\bbA+t\nu(c_1(\eta))+t\nu(c_N(\eta))
\end{equation}
where $\bbA=\nu(c_1(\eta))=\rho(1-\rho)^2$ and $\nu(c_N(\eta))=\rho(1-\rho)$ for East; $\bbA=\rho(1-\rho)(1-\rho^2)$, $\nu(c_1(\eta))=\nu(c_N(\eta))=\rho(1-\rho)$ for 
FA-1f with two zeros at the boundary; $\bbA=\rho(1-\rho)(1-\rho^2)$, $\nu(c_1(\eta))=\rho(1-\rho)^2$ and $\nu(c_N(\eta))=\rho(1-\rho)$ for FA-1f with one zero at the (right) boundary.
Thus
\begin{equation}
\label{boundA}
\bbA t N  \left(1-\frac{C}{N}\right)\leq <\cA(t)>\leq \bbA t N   \left(1+\frac{C}{N}\right)\end{equation}
and \eqref{b1} immediately follows from these bounds and inequality \eqref{jensen}. 
%Instead \eqref{b2} follows from the observation that if at time zero the configuration is
%completely filled and the clocks of all
%the sites which have the constraint satisfied in such configuration (i.e. sites belonging to the unconstrained set $\mathcal{U}$)
%do not ring before $T$, then $\mathcal{A}(T)=0$.
\end{proof}

\begin{proof}[Proof of Lemma \ref{bad}]

From \eqref{eq: measure} and Lemma \ref{easy} for any $\ga\leq 0$  

\begin{eqnarray*}
\mu_{\alpha,T}^{N}(\cW_{0,\delta}) \leq   \exp( -\alpha\bbA T(1+\frac{C}{N})) \bk{ \mathds{1}_{\cW_{0,\delta}}  }_{\frac{\ga}{N}} \,
\end{eqnarray*}
if instead $\alpha>0$ by using
the fact that $\cH  (\eta(s))\leq N$  we get
\begin{eqnarray*}
\mu_{\alpha,T}^{N}(\cW_{0,\delta}) \leq   \exp( -\alpha \bbA T+\alpha T) \bk{ \mathds{1}_{\cW_{0,\delta}}  }_{\frac{\ga}{N}} \, .
\end{eqnarray*}
The result immediately follows from the above inequalities and
applying Lemma \ref{prop: bad blocksnew}.
\end{proof}

\subsection{FA-1f in $d \geq 2$}
\label{secbad}

We start by extending to the higher dimensional case the notion of coarse grained activity. We let again $N/K$ be integer and partition $\Lambda_N^d$ into $(N/K)^d$ boxes of linear size $K$. We let $B_i$ be the boxes (with $i\in [1,(N/K)^d]$) numbered in such a way that for any $i\in [1,(N/K)^d-1]$  
%either the box $B_i$ contains at least one site of the uncostrained set (namely $B_i$ has at least one empty boundary site) or 
there is $j\in[1,\dots,d]$ such that $B_{i+1}$ is obtained by shifting $B_i$ of $K\vec e_j$.
Then we define the activity-density labels as in the one dimensional case and the event $\cW_{i,\delta}$ as in Definition \ref{defW} with $N/K$ substituted by $(N/K)^d$. 
The following holds
\begin{lem}\label{ahh}
Consider FA-1f model in dimension $d\geq 2$ with any boundary condition which guarantees ergodicity and $\rho\in(0,1)$.
The results in Lemma \ref{prop: bad blocksnew} hold also for FA-1f in dimension $d$ if we substitute  $N/K$ with $(N/K)^d$.
\end{lem}
\begin{proof}
The proof follows the same lines as for Lemma \ref{prop: bad blocksnew} with some new ingredients that we detail here. Since there is at least one empty boundary condition, there exists at least one box which has an empty site on its boundary. We let $j$ be the index of the smallest such box. Then we  let $V(\eta):=\sum_{i=1}^{j-1}\mathds{1}_{\mathcal{E}_i } (\eta)\mathds{1}_{\mathcal{Z}_{i+1}}(\eta)+\sum_{i=j+1}^{(N/K)^d}\mathds{1}_{\mathcal{E}_{i} } (\eta)\mathds{1}_{\mathcal{Z}_{i-1}}(\eta)+\mathds{1}_{\mathcal{E}_j} $ and notice (along the same lines as for the unidimensional case) that $\sum_{i=1}^{(N/K)^d}\mathds{1}_{u_{K,\epsilon}=0}(\eta)\leq V(\eta)$ for any $\eta$. Thus provided we can establish for $V$ the validity of Lemma \ref{prop: bad blocks} with $N/K$ substituted by $(N/K)^d$ we can conclude along the same lines as for the one dimensional case. The validity of this modified Lemma  \ref{prop: bad blocks} also follows along the same lines as  the one dimensional case with a different point that we detail here.
Recall that in one dimension under the event $Z_{i+1}$ which guarantees that there exists at least one empty site in $B_{i+1}$ we identified the rightmost such site, which we denoted by $\xi$. Then under the event that $\xi=j+iK$ we somehow extended the variance up to $L_{i,j}:=B_i\cup[iK,\dots,j+iK-1]$ and used the positivity of the spectral gap on $L_{i,j}$  (which is a volume of the type $\Lambda_{K+j}$) with fixed empty boundary condition at $j+iK$.
Now instead under the event $Z_{i+1}$ we number the sites of $B_{i+1}$ as $x_1,\dots,x_{K^d}$ in a way that $x_{i+1}$ is nearest neighbour of $x_i$ and then let $\xi$ be the empty site with the biggest label. Then under the event $\xi=x_j$ we let
 $L_{i,j}:=B_i\cup  B_{i+1}\setminus [x_{j},\dots,x_{K^d}]$ and use the positivity of the spectral gap on $L_{i,j}$ with fixed empty boundary condition at $x_j$.
Note that now $L_{i,j}$ is not an hypercube and  the boundary condition is a single empty boundary condition even if we are considering the dynamics on $\Lambda_N$ with, for example, completely empty boundary conditions. Nevertheless we can use the positivity of the spectral gap on a generic connected graph with empty boundary condition (see Proposition \ref{teogap}) to bound again the variance on $L_{i,j}$  with the Dirichlet form and then proceed as in the one dimensional case.
\end{proof}

 \begin{lem}
 \label{nonhopiu}
 Consider FA-1f model in dimension $d\geq 2$ with any boundary condition which guarantees ergodicity and $\rho\in(0,1)$.
 There exists $C(\rho)$ such that 
 \begin{equation}
\lim_{T\to\infty}\frac{1}{T}\log \mu_{\alpha,T}^{N,d}( \cW_{0,\delta})\leq -C\delta^2\left(\frac{N}{K}\right)^d\exp(\alpha/N^{d-c})+|\alpha|N^c(1+CN^{c-d})
  \end{equation}\end{lem}
\begin{proof}

The proof uses Lemma \ref{ahh} and follows exactly the same lines as the proof for Lemma \ref{bad}.

\end{proof}

\begin{lem}
Consider FA-1f model in dimension $d\geq 2$ with  $\rho\in(0,1)$ and boundary condition of dimension $c\in[0,d-1]$.
Then 
\begin{eqnarray}
\lim_{T\to \infty} \frac{1}{T} \log \bk{ \cW_{1,\delta} }_{\lambda} \leq - \frac{N^c S_\gr\delta }{4} \, ,
\end{eqnarray}
with $S_\gr$ defined in Proposition \ref{teogap}.
\end{lem}
\begin{proof}
We detail the proof in the case $d=2$, $c=1$ with the specific choice that all the sites in the boundary set which have first coordinate equal to $N+1$ are empty. The other cases can be proven analogously. 
As in the proof of Lemma \ref{lem: variance 0} we start by applying Donsker-Varadhan large deviation principle \eqref{DV2} which leads to
\begin{eqnarray}
\label{princ2}
\lim_{T\to \infty} \frac{1}{T} \log \bk{ \cW_{1,\delta}}_{\lambda} \leq -\exp(\lambda)\inf_{f\in\cC} \left\{ \cD_{N}^{(2)}(\sqrt f) \right \} \, ,
\end{eqnarray}
where $\cC$ is the set of positive functions $f$ s.t. \begin{equation}
\label{conditionff}
\nu(f)=1, \quad\quad 
\nu \left ( f\sum_{i=1}^{(N/K)^2}1_{u^i_{K,\epsilon}=1} \right)
\geq \delta\left(\frac{N}{K}\right)^2 \, ,
\end{equation}
and we added the index $(2)$ to explicitate the fact that we refer here to the Dirichlet form of the two dimensional model.
By the monotonicity of the rates for any function $f$, one has
$$
\cD_{N}^{(2)}(\sqrt f)=\sum_{i=1}^N\cD_N^{(2,i)}(\sqrt f)\geq \sum_{i=1}^N\mu\left(\cD_N^{(1,i)}(\sqrt f) \right) \geq S_{\rho}\sum_{i=1}^N\mu\left(\mbox{Var}_i(\sqrt f)\right)\, ,
$$
where $\cD_N^{(2,i)}$ is the contribution of the $i-th$ line to the Dirichlet form and $\cD_N^{(1,i)}$ is instead the Dirichlet form of the one dimensional FA-1f model on the $i-th$ line with empty boundary condition on the right border (note that $\cD_N^{(1,i)}(\sqrt f)$ is a function of the configuration on all the sites that do not belong to the $i$-th line) and $\mbox{Var}_i$ denotes the variance w.r.t. the Bernoulli measure restricted to the $i$-th line with the other variables held fixed. The first inequality follows from the fact that for any site  belonging to the $i-th$ line if the constraint is satisfied 
for the one dimensional model so it is for the two dimensional model (but the converse is not true), indeed for any $x$ it holds  
$1-\eta_{x+\vec e_1}\eta_{x-\vec e_1}\leq 1-\eta_{x+\vec e_1}\eta_{x-\vec e_1}\eta_{x+\vec e_2}\eta_{x-\vec e_2}$. The second inequality follows
%, or in formulas the influence classes of the two dimensional model include the influence classes of the one dimensional model. 
by using the spectral gap inequality \eqref{spgapineq} for the one dimensional model.
Then we notice that the condition \eqref{conditionff} implies that
\begin{equation}\label{conditionff2}
\nu \left ( f\sum_{i=1}^{N}\sum_{j=1}^{(N/K)}1_{\tilde u^{i,j}_{K,\epsilon}=1} \right)
\geq \delta N \frac{N}{K} \, ,
\end{equation}
where $\tilde u^{i,j}_{K,\epsilon}$ stands for the activity label on the $j$-th  one-dimensional box (of size $K$) of the line $i$. 
Thus we get
\begin{equation}\label{ssss}
\inf_{f\in\cC}\cD_N^{(2)}(\sqrt f)\geq S_{\rho} \inf_{\delta_1,\dots,\delta_N\atop \sum_{i=1}^N\delta_i\geq \delta N}\sum_{i=1}^N \mu(\inf_{f\in \cC_i(\delta_i)}\mbox{Var}_i(\sqrt{f})) \, ,
\end{equation}
where we denote by $\cC(\delta)$ the set of positive functions which satisfy the conditions $\nu(f)=1$ and
$$
 \nu\left(f\sum_{j=1}^{N/K}1_{\tilde u^{i,j}_{K,\epsilon}=1}\right)\geq\delta\frac{N}{K}$$
Following the lines of Lemma \ref{lem: variance 0}  we  get for any $i\in[1,N]$
$$
\inf_{f\in\cC_i(\delta_i)} \mbox{Var}_i(\sqrt{f}) \geq \delta_i-\rho^K-\rho^{K/2},
$$
and inserting this result in  \eqref{ssss} and \eqref{princ2}
yields the desired result provided $K$ is chosen sufficiently large so that $\rho^K+\rho^{K/2}\leq \delta 3/4$.
%***mettere remark per il caso non ergodico. qui se $\alpha<0$ ho sempre LDF=0. E per $\alpha>0$?
%**unless otherwise stated our theorems
%refer to any choice of the boundary condition which guarantee ergodicity**
\end{proof}

\section{Active regime}

In this section we prove Theorem \ref{teo:phasetrans} (i) which establishes the 
linearity of the moment generating function $\gp$ for FA-1f and East model in the small $\alpha$ regime. Then, we prove Theorem \ref{teo:condmes} (i) and the stronger result Lemma  \ref{usolemma} on the conditional measure. In section \ref{genesmall} we generalize these results to higher dimensions.
The key ingredints which will be used in these (quite technical) proofs are the main results obtained in the previous section, namely Lemma \ref{prop: bad blocksnew}  and Lemma \ref{lem: variance 0}.

\subsection{Linearity of the generating function: proof of Theorem \ref{teo:phasetrans} (i)}
\label{linearity}

In this section we will prove the  following proposition
%We collect here two lower bounds on $<\exp(\alpha\mathcal{A}(t)/N^a) >$ (and the corresponding upper bounds on $\mu_{\alpha,T}^{N,a}$) which will be used in the following sections. Recall the definition of the uncostrained and interior boudnary sets $\mathcal{U}$ and $\mathcal{I}$ given in Section \ref{models}, then
%By using definition \ref{eq: measure}, equation \eqref{cm} and the above Proposition the following corollary follows immediately
%\begin{Corollary}
%Let $\cV$ be any event in the sigma algebra generated by $\{\sigma_s\}_{s\leq t}$, then for any $N$ there exists $c(N)$ s.t.
%\begin{eqnarray}
%\label{eq: representation}
%\mu_{\alpha,T}^{N,a}({\cV}) 
%\leq c  \exp \big(  T \min \{ - \ga \bbA N^{d-a}(1+o(N)),  |\mathcal{U}|\} \,   \big)
%\bk{ mathds{1}_\cV \;  \exp \left( \left(\exp( \frac{\ga}{N^{a}} ) - 1 \right) \int_0^T \, ds \;   \cH  (\eta_s) \right) }_{\frac{\alpha}{N^a}} \, .
%\end{eqnarray}
%\end{Corollary}
%
%\bigskip
%
%
%
%
%%\medskip

\begin{Proposition}
\label{altra}
There exists $C<\infty$ s.t.
for any $\gd>0$ and $\alpha>\alpha_0=-\frac{S_{\rho}}{4\bbA}$
 there is $\bar N(\gd,\alpha)<\infty$ such that for all $N\geq \bar N$ it holds
\begin{eqnarray}
\label{eq: borne inf sup}
 \qquad 
\lim_{T \to \infty} \; \frac{1}{T} \log  \bk{ \exp \left( \frac{\ga}{N}  \cA(T) \right) }
\leq \bbA \ga + C \gd + C\frac{\alpha}{\sqrt N}+C\frac{\alpha^2}{\sqrt N}\, .
\end{eqnarray}
\end{Proposition}

From this proposition, we deduce
\begin{proof}[Proof of Theorem \ref{teo:phasetrans}(i)]
The result follows  from the definition \eqref{rescaled}  using the lower bound of Lemma \ref{easy} and the upper bound of Proposition \ref{altra}.
\end{proof}

We are therefore left with proving Proposition  \ref{altra}.

\begin{proof}[Proof of Proposition \ref{altra}]

We distinguish two cases.
  
 \smallskip

\noindent
{\bf Case $\alpha\leq 0$.} 

Choose $\epsilon= \bbA \frac{\delta}{2}$ and $K$ such that $K\geq \bar K( \frac{\delta}{2}, \bbA \frac{\delta}{2})$ 
with $\bar K$ defined  in Lemma  \ref{prop: bad blocksnew}, then 
%Let us start by noticing that the event $\cW_{1,a}$ can be decomposed as
%$\cW_{1,a}=\cup_{\ell=\lceil a N/K\rceil}^{N/K}\cV_{1,\ell}$ and that
\begin{equation}\label{ineevents}
\mathds{1}_{ \cW_{0,\delta/2}}+\mathds{1}_{ \cW_{-1,1-\delta}}+\sum_{\ell=\lceil \delta N/(2K)\rceil }^{N/K}\mathds{1}_{ \cV_{1,\ell}\cap \cW_{0,\delta/2}^c}\geq 1.\end{equation}
where the sets $\cV$ and $\cW$ were introduced in Definition \ref{defW}.
From \eqref{cm} it follows that
\begin{equation}
\label{brutta}
\bk{\exp\left(\frac{\alpha}{N}\mathcal{A}(T)\right)}\leq F_1+F_{2}+\sum_{\ell=\lceil \delta N/(2K)\rceil }^{N/K} F_{3,\ell}
\end{equation}
where
\begin{equation}
\label{I}
F_1:=\bk{\mathds{1}_{\cW_{0,\delta/2}}  \exp\left(\left(\exp(\frac{\alpha}{N})-1\right) \int_0^T \, ds \;  \cH (\eta(s))\right)}_{\ga/N}\end{equation}
\begin{equation}\label{II}
F_{2}:= \bk{ \mathds{1}_{  \cW_{-1,1-\delta}}\, \exp \left( \left(\exp(\frac{\alpha}{N})-1\right)  \int_0^T \, ds \;  \cH (\eta(s)) \right)     }_{\ga/N}\end{equation}
and
\begin{equation}\label{III}
F_{3,\ell}:=  
\bk{  \mathds{1}_{  \cV_{1,\ell}\cap \cW_{0,\delta/2}^c} \, \exp \left( \left(\exp(\frac{\alpha}{N})-1\right)  \int_0^T \, ds \;  \cH (\eta(s)) \right)    }_{\ga/N}.\end{equation}

By using the fact that  $\alpha\leq 0$,
%%%*{\sl si capisce o metto $<mathds{1}_{\cV_0}>$?}
we get from \eqref{I}, \eqref{II} and \eqref{III} 
\begin{equation}\label{Ibis}
F_1\leq  \bk{\mathds{1}_{\cW_{0,\delta/2}} }_{\ga/N}
\end{equation}
\begin{equation}\label{IIbis}
F_{2}\leq \bk{  \mathds{1}_{  \cW_{-1,1-\delta}} \,  \exp \left( \frac{\alpha}{N}  \int_0^T \, ds \;  \cH (\eta(s)) \right)    }_{\ga/N} 
\exp\left(\frac{T\alpha^2}{N}\right)
\end{equation}
and
\begin{equation}
\label{IIIbis}
F_{3,\ell}\leq 
\bk{  \mathds{1}_{  \cV_{1,\ell}\cap \cW_{0,\delta/2}^c}  \, \exp \left( \frac{\ga}{N}   \int_0^T \, ds \;  \cH (\eta(s)) \right)    }_{\ga/N}\exp\left(\frac{T\alpha^2}{N}\right).
\end{equation}
Recall  Lemma \ref{prop: bad blocksnew} then since we have chosen $K\geq \bar K(\frac{\delta}{2},\bbA \frac{\delta}{2})$ there is $C>0$ 
such that
\begin{eqnarray}
%\lim_{T \to \infty} \; \frac{1}{T} 
%\log(I)\leq 
\label{eq: gaussian bound neglect 2}
\lim_{T \to \infty} \; \frac{1}{T} 
\log(F_1)\leq \lim_{T \to \infty} \; \frac{1}{T} 
\log \bk{\mathds{1}_{\cW_{0,\delta/2}} }_{\ga/N}
\leq  - \delta^2C\frac{N}{K}  \exp(\frac{\alpha}{N})\, .
\end{eqnarray}
Then recalling inequality \eqref{usoH}, Definition  \ref{defW} for the event $\cW_{-1,1-\delta}$ and our choice $\epsilon=\bbA \frac{\delta}{2}$, we get from \eqref{IIbis} 
\begin{equation}\label{colc0}
F_{2}\leq \exp\left(T\alpha\left(\bbA ( 1 - \frac{\gd}{2}) ( 1- \gd) -\frac{2 \bbA}{K}\right) + \frac{T\alpha^2}{N} \right)
\end{equation}
Since  the event
$\cV_{1,\ell}\cap \cW_{0,\delta/2}^c$ is a subset of the event $\cW_{-1,1-\delta '}$ with $\delta '= \frac{\delta}{2}+ (\ell+1) \frac{K}{N}$, by using again \eqref{usoH}
we get from \eqref{IIIbis}
\begin{equation}
\label{IIItris}
F_{3,\ell}\leq 
%\bk{  mathds{1}_{  \cV_{1,\ell}\cap\cW_{0,\delta/2}^c}\,\exp \left( \frac{\ga}{N}   \int_0^T \, ds \;  \cH (\eta_s) \right)    }_{\ga/N^a}\leq 
\exp\left(T\alpha \left( \bbA ( 1 - \frac{\gd}{2}) -\frac{2\bbA}{K} \right) \left(1- \frac{\delta}{2} - (\ell+1) \frac{K}{N} \right) 
 + \frac{T \alpha^2}{N}\right)\bk{ \mathds{1}_{  \cW_{1,\ell \frac{K}{N} }}}_{\ga/N}
\end{equation}
where we used the fact that the event $\cV_{1,\ell}$ is a subset of the event $  \cW_{1,\ell \frac{K}{N} }$.
Then by using Lemma \ref{lem: variance 0}, it holds for any $\ell\geq\lceil \delta N/(2K)\rceil $ that
%\begin{equation}\label{colc}
%\lim_{T\to\infty}\frac{1}{T}\log\left\bra \cV_1^{\ell}\right\ket_{\ga/N^a}\leq-\frac{S_{\rho}\ell K^d}{4N^d}\end{equation}
%Then collecting \eqref{colc1} and \eqref{colc} yields for $a=d$ provided $\alpha>-\frac{S_{\rho}}{\bbA 4}$ 
%or for $a>d$ ad any $\alpha$
\begin{equation}
\label{colc31}
\lim_{T\to\infty}\frac{1}{T}\log \bk{ \mathds{1}_{  \cW_{1,\ell \frac{K}{N} }}}_{\ga/N}
\leq -\frac{S_{\rho}\ell K}{4N}\exp(\frac{\alpha}{N}) \, .
\end{equation}
Thus collecting \eqref{IIItris} and \eqref{colc31} we get
%% il manquait un term  +|\alpha|\bbA(1-\delta/2)\frac{K}{N}
%\begin{equation*}
%\lim_{T\to\infty}\frac{1}{T}\log(F_{3,\ell}) 
%\leq 
%\ga \bbA ( 1 - \frac{\gd}{2})^2 - \left( \ga \bbA ( 1 - \frac{\gd}{2}) + \frac{S_{\rho}}{4} \exp \big(\frac{\alpha}{N} \big)  \right) \frac{\ell K}{N} 
%- 2 \ga \frac{\bbA}{K} + \frac{\ga^2 }{N} \, .
%\end{equation*}
\begin{equation*}
\lim_{T\to\infty}\frac{1}{T}\log(F_{3,\ell}) 
\leq 
\ga \bbA ( 1 - \frac{\gd}{2})^2 - \left( \ga \bbA ( 1 - \frac{\gd}{2}) + \frac{S_{\rho}}{4} \exp \big(\frac{\alpha}{N} \big)  \right) \frac{\ell K}{N} +
\frac{C}{K} + \frac{C K }{N} \, .
\end{equation*}
Since $\alpha>\alpha_0=-\frac{S_{\rho}}{\bbA 4}$ and $ \frac{\ell K}{N}  \geq \frac{\gd}{2}$, one gets (for $\gd$ small enough)
\begin{eqnarray}
\label{colc3}
\lim_{T\to\infty}\frac{1}{T}\log(F_{3,\ell}) & \leq& 
\ga \bbA - \ga \bbA \gd - \left( \ga \bbA  + \frac{S_{\rho}}{4} \exp \big(\frac{\alpha}{N} \big)  \right) \frac{\gd}{2}+ \frac{C}{K} + \frac{C K }{N}\nonumber \\
& \leq & 
\ga \bbA - \ga \bbA \gd  \frac{C}{K} + \frac{C K }{N}\, .
\end{eqnarray}
%with $C>0$ provided $\alpha>\alpha_0=-\frac{S_{\rho}}{\bbA 4}$. 
Thus under this hypothesis by using \eqref{brutta} and
collecting \eqref{eq: gaussian bound neglect 2}, \eqref{colc0} and \eqref{colc3} if we now set 
$N\geq \bar K( \frac{\delta}{2}, \bbA \frac{\delta}{2})^2$ and 
$K=\sqrt N$ we get the desired result.
 
\medskip

\noindent
{\bf Case $\alpha>0$.} 

Recall Definition \ref{defW} and the definition \eqref{I} for $F_1$. Then it holds
\begin{equation}
\label{brutta2}
\bk{\exp\left(\frac{\alpha}{N}\mathcal{A}(T)\right)}\leq F_1+F_{4} \, ,
\end{equation}
where we let
\begin{equation}\label{IV}
F_4:=\bk{\mathds{1}_{\cW_{0,\delta/2}^c}  \exp\left(\left(\exp(\frac{\alpha}{N})-1\right) \int_0^T \, ds \;  \cH (\eta(s))\right)}_{\ga/N} \, .
\end{equation}
For  $N>\alpha/\log 2$, then  $\exp(\alpha/N)-1<\alpha/N+(\alpha/N)^2$.
Thus via definition \eqref{I} since $\cH \leq N$, using  Lemma \ref{prop: bad blocksnew} we get 
\begin{equation}\label{GI}
\lim_{T\to\infty}\frac{1}{T}\log (F_1)\leq -\delta^2 C\frac{N}{K}+\alpha+\frac{\alpha^2}{N} \, .
\end{equation}
Then  using definition \eqref{IV} and recalling inequality \eqref{usoH2} we get
\begin{equation}
\label{GII}
\lim_{T\to\infty}\frac{1}{T}\log (F_{4})\leq \alpha\bbA+\alpha\delta+\frac{3 \alpha}{K}+\frac{\alpha^2}{N} \, .
\end{equation}
%\begin{equation}\label{GII}
%\lim_{T\to\infty}\frac{1}{T}\log (F_{2,1})\leq \alpha\delta+\frac{2\alpha}{K}+\frac{\alpha^2}{N}
%\end{equation}
%and
%\begin{equation}\label{GIII}
%\lim_{T\to\infty}\frac{1}{T}\log(F_{3,\ell,1})\leq  \alpha \bbA +\alpha\delta+ \frac{\alpha^2}{N}
%\end{equation}
Thus by using \eqref{brutta2} and
collecting \eqref{GI} and \eqref{GII}  if we now set $N\geq\max( \bar K(\frac{\delta}{2}, \bbA \frac{\delta}{2})^2, \frac{\alpha}{\log 2})$ and 
$K=\sqrt N$ we get the desired result.

%\begin{equation}\label{GIII}
%\lim_{T\to\infty}\frac{1}{T}\log(G_3)\leq  \alpha \bbA +\alpha \frac{\delta}{2}+\frac{3\alpha}{K}+\frac{2\alpha^2}{N} 
%\end{equation}
%
%and from \eqref{GIII} we get
%%by using the fact that under that event $\cV_1$
%%$$\cH(\eta_s)\leq (\bbA+\epsilon)(1-\frac{\delta}{2})N^d+\frac{N^d}{K^d}+\frac{\delta}{2} N^d$$
%%and under the event $\cW_{-1}$
%$$\cH(\eta_s)\leq (\bbA+\epsilon)N^d+\frac{N^d}{K^d}+\delta N^d$$ we get
%$$\lim_{T\to\infty}\log(<\exp(\frac{\alpha}{N^a}\mathcal{A}(T))>)\leq  \alpha \bbA N^{d-a}+C\delta$$
%provided $\alpha\leq C/\delta$ and the desired result is proven with the choice $\alpha_0=C/\delta$.
\end{proof}

%\end{document}

\subsection{Conditional measure: proof of Theorem \ref{teo:condmes} (i) and a stronger result}

\label{condmes1}

The main result of this section is the following Lemma \ref{usolemma} which, as we will explain, is a stronger version of Theorem \ref{teo:condmes}(i).

%{\sl forse per il corollario sulla densita serve una condizione piu' mild da cui $\alpha_0$ indip da delta??}

%{\sl stessa alpha0?}
%{\bf pour bien faire les choses par tout il faudrait mettre $\cW^{K,\epsilon}_{i,a}$...ici c'est necessaire}
\begin{lem}
\label{usolemma}
Let $K_N$ and $\delta_N$ be two sequences such that $$\lim_{N\to\infty}\delta_N=\lim_{N\to\infty}1/K_N=\lim_{N\to\infty}K_N/(N\delta_N^2)=\frac{\log(\delta_N)}{K_N\delta_N^2}=0.$$
Note that there are subsequences which verify the above conditions, e.g. the choice $K_N=\sqrt N$ and $\delta_N=N^{-1/8}$.
For each $N$ let the activity-density labels be defined with $K=K_N$ and $\epsilon_N=\bbA \frac{\delta_N}{2}$.
Then there exists $\alpha_0<0$
such that if $\alpha>\alpha_0$ it holds
\begin{equation}
\lim_{N\to\infty}\lim_{T\to\infty}\mu_{\alpha,T}^N(\cW_{-1,1-\delta_N})=1\, .
\end{equation}
\end{lem}

\begin{proof}
Fix $ N$ sufficiently large in order that it holds $K_N\geq\bar K(\frac{\delta_N}{2}, \bbA \frac{\delta_N}{2})$
with $\bar K$ defined in Lemma \ref{prop: bad blocksnew} (this is possible thanks to the hypothesis $\lim_{N\to\infty}\frac{K_N\delta_N^2}{\log(\delta_N)}=\infty$). We distinguish two cases.
\smallskip

\noindent
{\bf Case $\ga <0$.}

Since $\mu_{\alpha,T}^N$ is a measure from inequality \eqref{ineevents} it holds
\begin{equation}
\label{impo}
1-\mu_{\alpha,T}^N({ \cW_{0,\delta_N/2}})-\sum_{\ell=\lceil \frac{\delta_N}{2} \frac{N}{K} \rceil }^{N/K}\mu_{\alpha,T}^N(\cV_{1,\ell}\cap \cW_{0,\delta_N/2}^c)\,\leq \,\mu_{\alpha,T}^N(\cW_{-1,1-\delta_N})\, \leq \,1
\end{equation}
Recall  equation \eqref{eq: measure} which defines the conditional measure $\mu_{\alpha,T}^N$. 
If we apply Proposition \ref{easy} to bound the denominator and \eqref{cm} to rewrite the numerator we get
\begin{equation}\label{eq1}
\mu_{\alpha,T}^N(\cW_{0,\delta_N/2})\leq 
\exp\left(-\alpha\bbA T (1+\frac{C}{N})\right)F_1 \, ,
\end{equation}
\begin{equation}\label{eq2}
\mu_{\alpha,T}^N( \cV_{1,\ell}\cap  \cW_{0,\delta_N/2}^c)\leq 
\exp\left(-\alpha\bbA T (1+\frac{C}{N})\right) F_{3,\ell} \, ,
\end{equation}
where $F_1$ and $F_{3,\ell}$ are the functions that have been defined respectively in \eqref{I} and \eqref{III} and 
here we set $\delta=\delta_N$.
Then we use \eqref{eq: gaussian bound neglect 2} and the assumption $\lim_{N\to\infty} \frac{\delta_N^2 N}{K_N} =\infty$ to conclude that
\begin{equation}
\label{ea}
\lim_{N\to\infty}\lim_{T\to\infty}\mu_{\alpha,T}^N(\cW_{0,\delta_N/2})=0 \, .
\end{equation}
Then by using \eqref{colc3}, we get for $\frac{\ell K}{N} \geq \frac{\gd}{2}$
\begin{equation*}
\lim_{T\to\infty} \; \frac{1}{T} \log \; \mu_{\alpha,T}^N(\cV_{1,\ell}\cap  \cW_{0,\delta_N/2}^c)
\leq  -  3 \ga \bbA \frac{\gd_N}{2} - \frac{S_{\rho}}{4} \frac{\gd_N}{2} + \frac{C'}{N} - 2 \ga  \frac{\bbA}{K_N} \, .
\end{equation*}
%\eqref{colc31} together with \eqref{eq2} and \eqref{eq3} we get
%for any $\delta_N$ 

By construction $\gd_N \gg \frac{1}{K_N} \gg \frac{1}{N}$.
Thus for $\alpha> \ga_0 = - \frac{S_{\rho}}{12 \bbA}$ and $N$ large enough
\begin{equation}
\label{eq4}
\lim_{T\to\infty}\mu_{\alpha,T}^N(\cV_{1,\ell}\cap  \cW_{0,\delta_N/2}^c)=0 \, .
\end{equation}
%for any $\ell$ such that $\ell K/N>\delta_N/2$ provided 
Note that the threshold $\ga_0$ obtained here is not as sharp as in Proposition \ref{altra}. 
The proof is then completed via \eqref{impo}, \eqref{ea} and \eqref{eq4}.

\smallskip

\noindent
{\bf Case $\ga \geq0$.}

Recall Definition \ref{defW}, then it holds
\begin{equation}\label{newineevents}
\mathds{1}_{\cW_{-1,1-\delta_N}}+\mathds{1}_{\cW_{0,\bbA\delta_N/4}}+\mathds{1}_{\cW_{0,\bbA\delta_N/4}^c \cap \cW_{-1,1-\delta_N}^c}\geq 1\, .
\end{equation} 
Since $\mu_{\alpha,T}^N$ is a measure from inequality \eqref{newineevents} it holds
\begin{equation}
\label{impo2}
1-\mu_{\alpha,T}^N({ \cW_{0,\bbA\delta_N/4}})-\mu_{\alpha,T}^N(\cW_{0,\bbA\delta_N/4}^c\cap \cW_{-1,1-\delta_N}^c)
\, \leq \,\mu_{\alpha,T}^N(\cW_{-1,1-\delta_N})\, \leq \,1 \, .
\end{equation}
Recall  equation \eqref{eq: measure} which defines the measure $\mu_{\alpha,T}^N$. 
Applying Proposition \ref{easy} to bound the denominator and \eqref{cm} to rewrite the numerator, there is $C>0$ such that
\begin{equation}\label{neweq1}
\mu_{\alpha,T}^N(\cW_{0,\bbA \delta_N/4})\leq 
\exp\left(-\alpha\bbA T(1-\frac{C}{N})\right) F_1 \, ,
\end{equation}
where $F_1$ was defined in \eqref{I}  and here we set $\delta= \frac{\delta_N \bbA}{4}$. 
Thus \eqref{eq: gaussian bound neglect 2} together with the hypothesis on $\delta_N$ and $K_N$ imply that 
\begin{equation}\label{impo3}
\lim_{N\to\infty}\lim_{T\to\infty}
\mu_{\alpha,T}^N(\cW_{0,\bbA \delta_N/4})=0 \, .
\end{equation}
On the other hand, by using inequality \eqref{usoH2} and again 
Proposition \ref{easy} to bound the denominator and \eqref{cm} to rewrite the numerator of the conditional measure, we get
with $\gep = \bbA \frac{\gd_N}{2}$
\begin{eqnarray*}
&& \lim_{T\to\infty}
\frac{1}{T} \log \left \bra 
1_{\left\{ \cW_{0,\bbA \frac{\delta_N}{4}}^c \cap \cW_{-1,1-\delta_N}^c \right\} } \exp \big( \frac{\ga}{N} \cA(t) \big)
\right 
\ket \\
&& \qquad 
\leq \ga \gd_N \frac{\bbA}{4} + \ga (1- \gd_N)  \bbA (1 + \frac{\gd_N}{2}) + \frac{2\ga}{K_N}  
\leq \ga \bbA \left( 1- \frac{\gd_N}{4} - \frac{\gd_N^2}{2} \right)  +\frac{2 \ga}{K_N}  \, .
\end{eqnarray*}
Thus
\begin{equation}
\label{impo4}
\lim_{N\to\infty}\lim_{T\to\infty}
\mu_{\alpha,T}^N(\cW_{0,\bbA\delta_N/4}^c \cap \cW_{-1,1-\delta_N}^c)=0.
\end{equation}
%thanks to the assumption $\lim_{N\to\infty}(\delta_N K_N)=\infty$.
The result is then proved thanks to \eqref{impo2}, \eqref{impo3} and \eqref{impo4}.
%uisng \eqref{cm} we get inequality 
\end{proof}

\begin{proof}[Proof of Theorem \eqref{teo:condmes}(i)]

It is immediate to verify that the event $ \cW_{-1,1-\delta}$ implies that $\pi_T(|\sum_{i=1}^N\eta_i-N\rho|)\leq (\epsilon +\delta)N$.
Therefore \eqref{eq: densite haute tris} is proven by using Lemma \ref{usolemma} and taking $\gamma_N=\frac{3}{4} \bbA \delta_N$. A similar argument leads to result \eqref{eq: densite haute tqua}.

\end{proof}

\subsection{FA-1f in dimension $d>1$: proof of Theorem \ref{linearityd>1}}
\label{genesmall}

The proof follows by using the results of Section \ref{secbad} along exactly the same lines as the proof of Theorem \ref{teo:phasetrans}(i)
and \ref{teo:condmes}(i).

\section{Inactive regime}

In this section we prove Theorem \ref{teo:phasetrans} (ii). We analyze in detail the case of FA-1f with two empty boundaries in Section \ref{genesmall1} and then we sketch
how the proof is extended to FA-1f with one empty boundary and East model in Section \ref{largealpha2}. Our proof will provide a variational characterization of   the constant $\Sigma$ which appears in the theorem and which, as explained in section \ref{heuristics}, plays the role of an interface energy. We underline that this variational problem, and therefore the value of $\Sigma$,  depends not only on the choice of the constraints (the interface energy for East and FA-1f at the same density are different) but also on the choice of the boundary conditions (the interface energy for FA-1f with one and two empty boundaries are also different, see Remark \ref{remdiff}). Finally, we prove Theorem \ref{poor} for FA-1f in higher dimensions.
%***dire qui o dopo che e' il doppiio***

\subsection{FA-1f with two empty boundaries: proof of  Theorem \ref{teo:phasetrans} (ii)}
\label{genesmall1}

%In this Section we focus on FA-1f with two empty boundaries.
%{\sl capire se la probva si applica pari pari per East e FA-1f con 1 solo bordo..o altrtimento dare hint di come si fa...}
%{\sl scrivere bene le cose in modo da poterle poi usare nella conditional measure...}
%\subsection{Proof of 
%Theorem \ref{teo:phasetrans} (ii) in the case of FA-1f model with two empty boundaries.
\label{largealpha}
%{\sl Ridefinire tutto come eventi non mi piacciono gli insiemi di misuere

%Faccio la prova per FA-1f e remark idem per east}

%We start by proving that the definition of $\Sigma$ is well posed, namely the limit $\lim_{L\to\infty}\lim_{L'\to\infty}\Sigma_{L,L'}$ exists.

%{\sl definire $\Sigma$ prima del teorema questo mettere prima: dipende se e' uguale per tutti i modelli oppure no}

Fix integer $L,L'>0$ and let $\mathcal{C}_{L,L'}$ be
the set of probability densities on $\Omega_{L+L'+2}$ such that $\eta_{L+1}=\eta_{ L+2}=1$ with probability one, namely
\begin{eqnarray}
\label{eq: CK}
\cC_{L,L'} = \Big \{ f : \qquad \nu \big(f \big) = 1, \quad \nu \big(f\eta_{L+1} \eta_{L+2}\big) = 1 \Big\} \, .
\end{eqnarray}
Let $\gS_{L,L'} = \inf_{f \in \mathcal{C}_{L,L'}}  \cD_{L+L'+2} \big( \sqrt{f} \big)$ with $\cD$ the Dirichlet form of FA-1f model with two empty boundaries and
 define the interface energy $\gS$ as
\begin{equation}
\label{defS1}
\gS:=\lim_{L\to\infty}\lim_{L'\to\infty}\gS_{L,L'}
\end{equation}
The definition is well posed thanks to the following Lemma \ref{monotonicity}.
%\end{proof}
\begin{lem}
\label{monotonicity}
$\gS_{L,L'}$ is non-increasing in $L$ and in $L'$. For any $L,L'$ it holds 
$\gS_{L,L'}\geq 0$. Therefore the limit $\lim_{L\to\infty}\lim_{L'\to\infty}\gS_{L,L'}$ exists.
\end{lem}

\begin{proof}
The positivity of $\gS_{L,L'}$ immediately follows from its definition. Let $f(\eta_1,\dots ,\eta_{L+L'+2})$ be the function s.t. $\Sigma_{L,L'}=\cD_{L+L'+2}(\sqrt f)$. Set $g(\eta_1,\dots,\eta_{L+L'+3}):=f(\eta_2,\dots,\eta_{L+L'+3}).$
Then $g\in\cC_{L+1,L'}$ and
\begin{eqnarray}\label{asfor}
\Sigma_{L+1,L'}\leq \cD_{L+L'+3}(\sqrt g)=\nonumber\\
\sum_{i=3}^{L+L'+3}\sum_{\omega\in\Omega_{L+L'+3}}
\nu(\omega) c_i(\omega) \left(\sqrt {g(\omega^i)}-\sqrt {g(\omega)}\right)^2+\nonumber\\
\sum_{\omega\in\Omega_{L+L'+3}} \nu(\omega)(1-\omega_{1}\omega_{3})(\rho(1-\omega_2)+(1-\rho)\omega_2)\left(\sqrt {g(\omega^2)}-\sqrt {g(\omega)}\right)^2\nonumber\\
\leq\sum_{i=2}^{L+L'+2}\sum_{\omega\in\Omega_{L+L'+2}} \nu(\omega) c_i(\omega)
\left(\sqrt {f(\omega^i)}-\sqrt {f(\omega)}\right)^2+\nonumber\\
\sum_{\omega\in\Omega_{L+L'+2}}\nu(\omega)(\rho(1-\omega_1)+(1-\rho)\omega_1)\left(\sqrt {f(\omega^1)}-\sqrt {f(\omega)}\right)^2=\nonumber\\
\cD_{L+L'+2}(\sqrt f)=\Sigma_{L,L'} \, .
\end{eqnarray}
Note that we have used the fact that the occupation variable at position $1$ is unconstrained. 
Analogously if we set $h(\eta_1,\dots,\eta_{L+L'+3}):=f(\eta_1,\dots,\eta_{L+L'+2}).$ 
Then $h\in\cC_{L,L'+1}$ and $ \cD_{L+L'+3}(\sqrt h)\leq \cD_{L+L'+2}(\sqrt f)$.
Thus $\Sigma_{L,L'+1}\leq\Sigma_{L,L'}$ follows.

\end{proof}

We split the proof of Theorem \ref{teo:phasetrans} (ii) into upper and lower bounds which are stated in
 the two following lemmas 
\begin{lem}Let $\Sigma$ be defined as in \eqref{defS1}. Then for FA-1f model with two empty boundaries and any $\alpha<0$  it holds
\label{lowereasy}\begin{eqnarray}
\label{eq: free energy lower} 
\lim_{N \to \infty} \lim_{T \to \infty}  \frac{1}{T} \log \bk{ \exp \left( \frac{\ga}{N} \cA(T) \right) }
\geq - \gS  \, .
\end{eqnarray}

\end{lem}

\begin{lem}\label{upperdifficult}
Let $\Sigma$ be defined as in \eqref{defS1}. Then for FA-1f model with two empty boundaries 
and any $\alpha<-\frac{\Sigma+8\sqrt {\rho(1-\rho)}}{\bbA}$ it holds
\begin{eqnarray}
\label{eq: free energy upper} 
\lim_{N \to \infty} \lim_{T \to \infty}  \frac{1}{T} \log \bk{ \exp \left( \frac{\ga}{N} \cA(T) \right) }
\leq - \gS.
\end{eqnarray}

\end{lem}

Then 
\begin{proof} [Theorem \ref{teo:phasetrans} (ii) for FA-1f model with two empty boundaries]
The result follows immediately from  Lemma \ref{lowereasy} and Lemma \ref{upperdifficult}.
\end{proof}

We are now left with the proof of the two above Lemmas.
\begin{proof}[Proof of Lemma \ref{lowereasy}]

%\bf Lower Bound.}

We fix $K$ and take $N \geq K$. Let $\cO$ be the event which is verified iff $\pi_T(\eta_{K+1},\dots\eta_{N-K})=1$.
%Let $\cO$ the event such that at any time $t \in [0, T]$, the configuration satisfies $\eta_i(t) = 1$ for $i \leq N-K$. There is no constraint on the remaining $K$ sites. 
Then from \eqref{cm} we get
%\eqref{eq: Radon Nykodim}, one has  
\begin{eqnarray}\label{num}
\bk{ \exp\left( \frac{\ga}{N}  \cA(T)\right) }
\geq \bk{ \exp\left( \frac{\ga}{N}  \cA(T) \right) \mathds{1}_{\cO}}\geq
\bk{  \exp \left( \frac{\ga}{N}   \int_0^T \, ds \;  \cH (\eta(s)) \right)  \, \mathds{1}_{\cO} }_{\ga/N} \exp\left(\frac{T\alpha^2}{N}\right)\, .
\end{eqnarray}
By  using the fact that  on the event $\cO$, it holds $\int_0^T \eta_{i-1}(s)\eta_{i+1}(s) \, ds =T$ for any $i\in [K+2,N-K-1]$ we get
% $\int_0^T\cH(\eta(s))ds\leq T(2K+2)$, indeed
\begin{eqnarray*}
\label{whynot1}
& \int_0^T ds \cH(\eta(s)) \leq  2(K+1)T+\sum_{i=K+2}^{N-K-1}\int_0^T ds~~ c_i(\eta(s))=\nonumber  \\
& 2(K+1)T+\sum_{i=K+2}^{N-K-1}\int_0^T ds (1-\eta_{i-1}(s)\eta_{i+1}(s))=2(K+1)T
\end{eqnarray*}
which together with \eqref{num} yields
\begin{eqnarray}\label{whynot0}
\bk{ \exp( \frac{\ga}{N}  \cA(T) ) }
\geq \exp\left( \frac{\ga}{N} 2(K+1) T +\frac{\ga^2}{N} T\right) \bk{  \mathds{1}_\cO  }_{\ga/N} \, .
\end{eqnarray}
From the Donsker-Varadhan large deviation principle \eqref{DV2} it holds

\begin{eqnarray}\label{whynot}
\lim_{T \to \infty}  \frac{1}{T} \log \bk{ \mathds{1}_\cO }_{\ga/N} =-\exp( \frac{\ga}{N})\inf_{f:\nu(f)=1,\nu(f\eta_{K+1}\dots\eta_{N-K})=1}\cD_N(\sqrt f)\geq -\Sigma_{K,K} \, ,
\end{eqnarray}
where in order to obtain the last inequality we proceed as follows. Denote by $g$ the function that belongs to $\cC_{K,K}$ and s.t. $\Sigma_{K,K}=\cD_{K+K+2}(\sqrt g)$.
Set 
$$
h(\eta_1,\dots,\eta_N):= \frac{1}{\rho^{N-2K-2}} 
g(\eta_1,\dots\eta_{K},\eta_{K+1},\eta_{N-K},\eta_{N-K+1},\dots\eta_{N})\prod_{j=K+2}^{N-K-1}\eta_{j} \, .
$$ 
Then it can be verified that
$\nu(h)=1$, $\nu(h \eta_{K+1}\dots\eta_{N-K})=1$ and $\cD_N(\sqrt h) =  \cD_{K+K+2}(\sqrt g)$.
Thus \eqref{whynot} follows. From \eqref{whynot0} and \eqref{whynot}
we therefore obtain
\begin{eqnarray}
\label{eq: free energy lower int} 
\lim_{T \to \infty}  \frac{1}{T} \log \bk{ \exp \left( \frac{\ga}{N} \cA(T) \right) }
\geq - \gS_{K,K}+   \frac{\ga}{N} 2(K+1)+\frac{\ga^2}{N} \, .
\end{eqnarray}

The result follows by  taking $N$ to infinity and then $K$ to infinity.
\end{proof}

\bigskip

Before proving Lemma \ref{upperdifficult}, we state and prove an auxiliary result.
Fix integers $L,L'>0$ and $u\in (0,1)$ and consider $\mathcal{C}_{L,L'}^u$
the set of probability densities on $\{0,1\}^{L+L'+4}$ such that $\eta_{L+1}=\eta_{ L+2}=\eta_{L+3}=\eta_{L+4}=1$ with probability at least $1-u$, namely
\begin{eqnarray}
\label{eq: CK}
\cC_{L,L'}^u = \Big \{ f : \qquad \nu \big(f(\eta) \big) = 1, \quad \nu \big(f(\eta) \eta_{L+1} \eta_{L+2}\eta_{L+3}\eta_{L+4}\big) \geq 1-u \Big\} \, .
\end{eqnarray}
We will now prove  that provided $u$ is sufficiently small the interface energy is well approximated by the infimum of 
$\cD_{L+L'+2}(\sqrt f)$ restricted to $\cC_{L,L'}^u$. 
More precisely
\begin{lem}
\label{prop: interface energy approx}
For any $u>0$
\begin{eqnarray}
\label{eq: SKu}
\inf_{f \in \cC^u_{L,L'}}  \cD_{L+L'+2} \big( \sqrt{f} \big)  \geq \gS_{L,L'} - \left( 8\sqrt{ \gr(1-\gr)}  + \gS_{L,L'} \right) u \, .
\end{eqnarray}
\end{lem}

\begin{proof}
%We decompose $\eta$ as $ \eta= (\eta_{L+1},\eta_{L+2}, \go)$ where $\go = \{ \eta_1, \dots, \eta_L ,\eta_{L+3},\dots \eta_{L+L'+2}\}$.
Fix $f \in \cC^u_{L+1,L'+1}$. 
%By using the fact that $\nu$ is product and recalling that here $c_{L}(\eta_{L+1},\eta_{},\omega)=1-\omega_L\eta_{L+2}
The Dirichlet form \eqref{eq: dirichlet} can be bounded from below by
\begin{eqnarray}
\label{eq: lower Dirichlet}
\cD_{L+L'+4} (\sqrt{f}) 
&\geq&  \sum_{i=1}^{L+1} 
\sum_{\eta:\eta_{L+2}=\eta_{L+3}=1}\nu(\eta) c_i  (\eta)  \big( \sqrt{f(\eta^i)} - \sqrt{f (\eta)} \big)^2 +\nonumber \\
& & \sum_{i=L+4}^{L+L'+4} 
\sum_{\eta:\eta_{L+2}=\eta_{L+3}=1} \nu(\eta)c_i  (\eta)  \big( \sqrt{f(\eta^i)} - \sqrt{f (\eta)} \big)^2 +\nonumber \\
& &
\sum_{\eta:\eta_{L+3}=1} \nu(\eta)c_{L+2}  (\eta)  \big( \sqrt{f(\eta^{L+2})} - \sqrt{f (\eta)} \big)^2+\nonumber\\
& & \sum_{\eta:\eta_{L+2}=1} \nu(\eta)c_{L+3}  (\eta)  \big( \sqrt{f(\eta^{L+3})} - \sqrt{f (\eta)} \big)^2 \, .
%\sum_{i \in(1,L+L'+2)\setminus(L+1,L+2)}
%\sum_\go \nu(1,1,\go)  c_i  (1, 1,\go)  \big( \sqrt{f(1, 1,\go^i)} - \sqrt{f (1,1,\go)} \big)^2 \nonumber \\
%&& \qquad 
%+ \sum_\go \nu(1,1,\go)  \left( c_{L+1} (1,1,\go)  \big( \sqrt{f(1,1, \go)} - \sqrt{f (0,1,\go)} \big)^2 \right) \nonumber\\
%&&\sum_\go \nu(1,1,\go)  \left( c_{L+2} (1,1,\go)  \big( \sqrt{f(1,1, \go)} - \sqrt{f (1,0,\go)} \big)^2 \right) \, . \nonumber
\end{eqnarray}
We define a new probability density $g$ on $\{0,1\}^{L+L'+4}$ 
\begin{eqnarray*}
g(\eta) = \frac{f(\eta)\eta_{L+2}\eta_{L+3}}{c(u)} \, ,
%\frac{1}{c(u)} f(1,1,\go), \qquad g(0,1,\go) =g(1,0,\go)=g(0,0,\go)= 0
 %\, ,
\end{eqnarray*}
with $c(u): = \nu \big( f(\eta) \eta_{L+2}\eta_{L+3} \big)$. Note that $\nu(g\eta_{L+2}\eta_{L+3})=\nu(g)=1$, thus $g$ belongs to $\cC_{L+1,L'+1}.$ Furthermore, 
since $f$ belongs to $\cC^u_{L+1,L'+1}$ one has $c(u) \geq 1- u$. 
Note that the Dirichlet form of $g$ satisfies
\begin{eqnarray}
c(u)\cD_{L+L'+4}(\sqrt g)
&=&
\sum_{i=1}^{L+1} 
\sum_{\eta:\eta_{L+2}=\eta_{L+3}=1}\nu(\eta) c_i  (\eta)  \big( \sqrt{f(\eta^i)} - \sqrt{f (\eta)} \big)^2 +\nonumber \\
& & \sum_{i=L+4}^{L+L'+4} 
\sum_{\eta:\eta_{L+2}=\eta_{L+3}=1} \nu(\eta)c_i  (\eta)  \big( \sqrt{f(\eta^i)} - \sqrt{f (\eta)} \big)^2 +\nonumber \\
& & + 2(1-\rho)\sum_{\eta:\eta_{L+2}=\eta_{L+3}=1}\nu(\eta) (1-\eta_{L+1})f(\eta)\nonumber \\ 
& & + 2(1-\rho)\sum_{\eta:\eta_{L+2}=\eta_{L+3}=1}\nu(\eta) (1-\eta_{L+4})f(\eta) \, .
%\frac{1}{c(u)}\sum_{i \in(1,L+L'+2)\setminus(L+1,L+2)} \sum_\go \nu(1,1,\go)  c_i  (1, 1,\go)  \big( \sqrt{f(1, 1,\go^i)} - \sqrt{f (1,1,\go)} \big)^2 +\nonumber\\
%\frac{2}{c(u)}\rho^2(1-\rho)\sum_{\omega}\nu(\omega)(1-\omega_{L})f(1,1,\omega)+\frac{2}{c(u)}\rho^2(1-\rho)\sum_{\omega}\nu(\omega)(1-\omega_{L+1})f(1,1,\omega)
\label{eq: lower Dirichletg}
\end{eqnarray}
Decompose $\eta$ as $\eta=(\omega_l,\eta_{L+2},\eta_{L+3},\omega_r)$ where 
$\omega_l=\eta_1,\dots\eta_{L+1}$ and $\omega_r=\omega_{L+4},\dots\omega_{L+L'+4}$, then 
\begin{eqnarray}\label{chec}
& &\sum_{\eta:\eta_{L+3}=1} \nu(\eta)c_{L+2}  (\eta)  \big( \sqrt{f(\eta^{L+2})} - \sqrt{f (\eta)} \big)^2 \geq  2(1-\rho)\sum_{\eta:\eta_{L+2}=\eta_{L+3}=1} \nu(\eta) (1-\eta_{L+1})f(\eta) \nonumber\\
& & \qquad - 4\rho(1-\rho)\sum_{\omega_l,\omega_r}\nu(\omega_l)\nu(\omega_r)\rho(1-\eta_{L+1})\sqrt{f(\omega_l,1,1,\omega_r)}\sqrt{f(\omega_l,0,1,\omega_r)} \, .
\end{eqnarray}
Then by using Cauchy-Schwartz
\begin{eqnarray}\label{chc2}
\sum_{\omega_l,\omega_r}\nu(\omega_l)\nu(\omega_r)\rho(1-\eta_{L+1})\sqrt{f(\omega_l,1,1,\omega_r)}\sqrt{f(\omega_l,0,1,\omega_r)}
\leq\nonumber\\
\frac{1}{\sqrt{\rho(1-\rho)}} \sqrt {\nu\left(f(1-\eta_{L+1})\right)}\sqrt{\nu(f(1-\eta_{L+1}))}\leq \frac{u}{\sqrt{\rho(1-\rho)}}\nonumber\\
\end{eqnarray}
where to obtain the  last inequality we used the fact that $f$ belongs to $\cC^u_{L,L'}$.
Thus
\begin{eqnarray}
\label{chec2}
& &\sum_{\eta:\eta_{L+3}=1} \nu(\eta)c_{L+2}  (\eta)  \big( \sqrt{f(\eta^{L+2})} - \sqrt{f (\eta)} \big)^2 \\
&& \qquad \qquad \geq 2(1-\rho)\sum_{\eta:\eta_{L+2}=\eta_{L+3}=1}\nu(\eta) (1-\eta_{L+1})f(\eta)-4\sqrt{\rho(1-\rho)}u \, .
\nonumber
\end{eqnarray}
Similarly it can be verified that
\begin{eqnarray}\label{chech3}
& &\sum_{\eta:\eta_{L+2}=1} \nu(\eta)c_{L+3}  (\eta)  \big( \sqrt{f(\eta^{L+3})} - \sqrt{f (\eta)} \big)^2\\
&& \qquad \qquad 
\geq 
2(1-\rho)\sum_{\eta:\eta_{L+2}=\eta_{L+3}=1}\nu(\eta) (1-\eta_{L+4})f(\eta)-4\sqrt{\rho(1-\rho)}u \, .
\nonumber
\end{eqnarray}
Thus by using \eqref{eq: lower Dirichlet}, \eqref{eq: lower Dirichletg}, \eqref{chec2} and \eqref{chech3}
we get
\begin{eqnarray*}
\cD_{L+L'+4} (\sqrt{f}) \geq 
c(u) \cD_{L+L'+4} (\sqrt{g}) - 8\sqrt{\rho(1-\rho)}u\geq \Sigma_{L+1,L'+1}- (8\sqrt{\rho(1-\rho)}+\Sigma_{L+1,L'+1})u \, ,
\end{eqnarray*}
where for the last inequality we used that, as noted above, $g$  belongs to $\cC_{L+1,L'+1}$ and $c(u)\geq (1-u)$.
\end{proof}

\begin{proof}[Proof of Lemma \ref{upperdifficult}] 
Recall Definition \ref{defW} and set $\epsilon=\bbA \delta$ and choose $K\geq \bar K(\delta, \bbA \delta)$
with $\bar K$ defined in Lemma \ref{prop: bad blocksnew} and let $a_N=N^{-1/16}$. Then the following inequality holds
\begin{equation}\label{ineevents3}
\mathds{1}_{ \cW_{0,\gd a_N}}+\mathds{1}_{ \cW_{1,1-\delta}}+\sum_{\ell=\lceil \delta(1-a_N)N/K\rceil }^{N/K}\mathds{1}_{ \cV_{-1,\ell}\cap \cW_{0,\delta a_N}^c}\geq 1
\end{equation}
which implies
\begin{equation}
\label{brutta3}
\bk{\exp\left(\frac{\alpha}{N}\mathcal{A}(T)\right)}\leq G_1+G_{2}+\sum_{\ell=\lceil \delta (1-a_N) N/K\rceil }^{N/K} G_{3,\ell}
\end{equation}
where
\begin{equation}
\label{GIbis}
G_1:=\bk{\mathds{1}_{\cW_{0,\delta a_N }}  \exp\left(\left(\exp(\frac{\alpha}{N})-1\right) \int_0^T \, ds \;  \cH (\eta(s))\right)}_{\ga/N}\end{equation}
\begin{equation}
\label{GIIbis}
G_{2}:= \bk{ \mathds{1}_{ \cW_{1,1-\delta}}\, \exp \left( \left(\exp(\frac{\alpha}{N})-1\right)  \int_0^T \, ds \;  \cH (\eta(s)) \right)     }_{\ga/N}\end{equation}
\begin{equation}
\label{GIII}
G_{3,\ell}:=  
\bk{  \mathds{1}_{  \cV_{-1,\ell}\cap\cW_{0,\delta a_N }^c} \, \exp \left( \left(\exp(\frac{\alpha}{N})-1\right)  \int_0^T \, ds \;  \cH (\eta(s)) \right)    }_{\ga/N} \, .
\end{equation}
As in \eqref{eq: gaussian bound neglect 2},  $G_1$ is bounded by 
\begin{eqnarray}
%\lim_{T \to \infty} \; \frac{1}{T} 
%\log(I)\leq 
\label{miserve}
\lim_{T \to \infty} \; \frac{1}{T} 
\log(G_1)\leq - a_N^2\delta^2C\frac{N}{K}  \, .
\label{Itris}
\end{eqnarray}
On the other hand since $\alpha<0$, we have
%\begin{eqnarray}
%\label{IIquo}
%G_{2}\leq \bk{\mathds{1}_{\cW_{1,1-\delta}}}_{\alpha/N}\exp(\frac{T\alpha^2}{N}) \, .
%\end{eqnarray}
\begin{eqnarray}
\label{IIquo}
G_{2}\leq \bk{\mathds{1}_{\cW_{1,1-\delta}}}_{\alpha/N} .
\end{eqnarray}
We notice that the event $\cW_{1,1-\delta}$ implies that there exists at least one box $i\in[1,N/K]$ such that 
$\pi_T (\mathds{1}_{\{u^i_{K,\epsilon}=1\}})\geq 1-\delta$.
Thus in this box, there are at least 4 consecutive sites occupied with high probability 
%$\pi_T(\eta_{(i-1)K+\lceil K/2\rceil }\eta_{(i-1_{K+\lceil K/2\rceil+1}}\eta_{(i-1_{K+\lceil K/2\rceil+2}}\eta_{(i-1_{K+\lceil K/2\rceil+3}})\geq 1-\delta$
\begin{equation}
\label{eq: occupied sites}
\pi_T(\eta_{(i-1)K+\lceil K/2\rceil +1} \; \eta_{(i-1)K+\lceil K/2\rceil+2} \;  \eta_{(i-1)K+\lceil K/2\rceil+3} \;  \eta_{(i-1)K+\lceil K/2\rceil+4})
\geq 1-\delta.
\end{equation}
Let $\cR_i$ be the  event that is verified if \eqref{eq: occupied sites} holds. We get
\begin{eqnarray}
\label{IIquo2} \bk{\mathds{1}_{\cW_{1,1-\delta}}}_{\alpha/N}\leq \sum_{i=1}^{N/K} \bk{\mathds{1}_{\cR_{i}}}_{\alpha/N} \, .
\end{eqnarray}
Donsker-Varadhan large deviation principle \eqref{DV2} implies
\begin{eqnarray*}
%\label{IIquo3} 
\lim_{T\to\infty}\frac{1}{T}\log  \bk{\mathds{1}_{\cR_{i}}}_{\alpha/N}
= -\exp(\alpha/N)\inf_{f\in\cC_{L,L'}^{\delta}} \left\{ \cD_{L+L'+2}(\sqrt f) \right\} \, ,
\end{eqnarray*}
where we have set $L=(i-1)K+\lceil K/2\rceil $ and $L'=N-2-(i-1)K-\lceil K/2\rceil $. 
By using Lemma \ref{prop: interface energy approx}, noticing that $L,L'\geq \lceil K/2\rceil-2$ and recalling the monotonicity property stated by Lemma \ref{monotonicity}, we obtain
\begin{eqnarray}
\label{IIquo3} 
\lim_{T\to\infty}\frac{1}{T}\log  \bk{\mathds{1}_{\cR_{i}}}_{\alpha/N}
&\leq& -\Sigma_{L,L'}+(8\sqrt{\rho(1-\rho)}+\Sigma_{L,L'})\delta-\frac{\alpha}{N}\gS_{L,L'} \nonumber \\
&\leq& -\Sigma+(8\sqrt{\rho(1-\rho)}+\Sigma_{\frac{K}{2},\frac{K}{2}})\delta+\frac{|\alpha|}{N}\gS_{1,1} \, .
\end{eqnarray}
Thus
\begin{eqnarray}
\label{IIquo5}
%\lim_{T\to\infty}\frac{1}{T}\log G_{2}\leq -\Sigma+C\delta+\frac{\max(\alpha^2,|\alpha|\Sigma_{1,1})}{N} \, ,
\lim_{T\to\infty}\frac{1}{T}\log G_{2}\leq -\Sigma+C\delta+\frac{|\alpha|\Sigma_{1,1}}{N} \, ,
\end{eqnarray}
where $C=8\sqrt{\rho(1-\rho)}+\Sigma_{\frac{K}{2},\frac{K}{2}}$.

\medskip

By using definition \eqref{GIII} and inequality \eqref{usoH}, we get
\begin{eqnarray}
\label{IIquo3}
G_{3,\ell}\leq \bk{\mathds{1}_{\cV_{-1,\ell}\cap\cW_{0,\delta a_N}^c}}_{\alpha/N}
\exp \left(\frac{T\alpha^2}{N}\right)\exp\left(\frac{\alpha}{N}T\ell (K\bbA-2\bbA- K\epsilon)\right) \, .
\end{eqnarray}
Note that $\cV_{-1,\ell}\cap\cW_{0,\delta a_N}^c$ is a subset of the event $\cW_{1,1-(\delta a_N+(\ell+1)\frac{K}{N})}$. 
Thus, we can use an estimate similar to \eqref{IIquo5} in order to bound from above 
$\bk{\mathds{1}_{\cW_{1,1-(\delta a_N+(\ell+1)\frac{K}{N})}}}_{\alpha/N}$
\begin{equation}
\lim_{T\to\infty}\frac{1}{T}\log \bk{\mathds{1}_{\cV_{-1,\ell}\cap\cW_{0,\delta a_N}^c}}_{\alpha/N}
\leq 
-\Sigma+C \left( \delta a_N+(\ell+1)\frac{K}{N} \right) 
%+\frac{\alpha}{N}K \ell ( \bbA(1 - \gd) -2\frac{\bbA}{K})
+\frac{C''}{N} \, ,
\label{tis}
\end{equation}
where $C'' >0$ is a constant.
Combining \eqref{IIquo3} and \eqref{tis}, we obtain 
\begin{equation}
\lim_{T\to\infty}\frac{1}{T}\log G_{3,\ell} 
\leq 
-\Sigma+C \left( \delta a_N+(\ell+1)\frac{K}{N} \right) 
+ \alpha \frac{K}{N} \ell  \left( \bbA(1 - \gd) -2\frac{\bbA}{K} \right)
+\frac{C''}{N} \, , 
\label{tis++}
\end{equation}
where we used that $\epsilon = \bbA \gd$.
Recall $\ell\geq \lceil \delta (1-a_N) \frac{N}{K} \rceil $ and $C=8\sqrt{\rho(1-\rho)}+\Sigma_{\frac{K}{2},\frac{K}{2}}$.
Thus for $\ga < -\frac{8\sqrt{\rho(1-\rho)}+\Sigma}{\bbA}$, we can choose $\gd$ small enough and $K$ large enough such that 
\begin{equation}
\lim_{T\to\infty}\frac{1}{T}\log G_{3,\ell}\leq -\Sigma+C' \frac{K}{N} \, .
\label{tistis}
\end{equation}
Sending $N\to \infty$
and then $K\to\infty$, we get the desired result by collecting \eqref{brutta3}, \eqref{Itris}, \eqref{IIquo5} and \eqref{tistis}.

%\vskip2cm

%As for any $\ell\geq \lceil \delta (1-a_N) \frac{N}{K} \rceil $ it holds $\delta a_N\leq \frac{a_N K \ell}{(1-a_N)N}$, we obtain
%\begin{equation}
%\lim_{T\to\infty}\frac{1}{T}\log \bk{\mathds{1}_{\cV_{-1,\ell}\cap\cW_{0,\delta a_N}^c}}_{\alpha/N}\leq -\Sigma+C\frac{K}{N}(\frac{\ell}{1-a_N}+1)+\frac{C}{N} \, ,
%\label{tis}
%\end{equation}
%with $C=8\sqrt{\rho(1-\rho)}+\Sigma_{L,L'}$ with $L=(i-1)K+\lceil K/2\rceil $ and $L'=N-2-(i-1)K+\lceil K/2\rceil $. Thus \eqref{IIquo3} and \eqref{tis}
%imply that there exists $C'>0$ such that
%\begin{equation}
%\lim_{T\to\infty}\frac{1}{T}\log G_{3,\ell}\leq -\Sigma+C' \frac{K}{N}\label{tistis}
%\end{equation}
%provided
%$\alpha<\alpha_N$ with
%\begin{equation}\label{alphaN}\alpha_N:=-\frac{8\sqrt{\rho(1-\rho)}+\Sigma_{L,L'}}{\bbA(1-a_N)}.\end{equation}
%Note that $\lim_{K\to\infty}\lim_{N\to\infty}\alpha_N=-\frac{8\sqrt{\rho(1-\rho)}+\Sigma}{\bbA}$.
%Therefore if we send $N\to \infty$
%and then $K\to\infty$ and we collect \eqref{brutta3}, \eqref{Itris}, \eqref{IIquo5} and \eqref{tistis} we get the desired result.

\end{proof}

\subsection{East and FA-1f with one empty boundaries: proof of  Theorem \ref{teo:phasetrans} (ii)}
\label{largealpha2}
Here we explain how to extend the results of the previous section to the case of East and FA-1f model with one empty boundary, thus completing the proof of Theorem \ref{teo:phasetrans} (ii). We start by giving the definition of the interface energy $\Sigma$ for these models.
Fix integer $L>0$ and consider $\mathcal{C}_{L}$
the set of probability densities on $\Omega_L$ such that $\eta_{1}=1$ with probability one, namely
\begin{eqnarray}
\label{eq: CK2}
\cC_{L} = \Big \{ f ; \qquad \nu \big(f \big) = 1, \quad \nu \big(f\eta_{1}) = 1 \Big\} \, .
\end{eqnarray}
Let $\gS_{L} = \inf_{f \in \mathcal{C}_{L}}  \cD_{L} \big( \sqrt{f} \big)$, then 
we define the interface energy $\gS$ as
\begin{equation}
\label{defS2}
\gS:=\lim_{L\to\infty}\gS_{L}
\end{equation}
The definition is well posed thanks to the following Lemma \ref{monotonicity}.
%\end{proof}
\begin{lem}
\label{monotonicity2}
Let $\cD_L$ be either the Dirichlet form of FA-1f with one empty boundary or the Dirichlet form of East. Then 
$\gS_{L}$ is non-increasing in $L$ and it holds 
$\gS_{L}\geq 0$. Therefore the limit $\lim_{L\to\infty}\gS_{L}$ exists.
\end{lem}
\begin{proof} Let $f(\eta_1,\dots,\eta_L)$ be the function in $\cC_L$ s.t. $\cD(\sqrt f)=\Sigma_L$. Then set  $g(\eta_1,\dots,\eta_{L+1}):=f(\eta_1,\dots,\eta_L)$. Then $g\in\cC_{L+1}$ and, as for inequality \eqref{asfor}, one can verify that $\cD_{L+1}(\sqrt g)
\leq \cD_L(\sqrt f)$ which implies $\Sigma_{L+1}\leq\Sigma_L$ (since $\Sigma_{L+1}\leq \cD_{L+1}(\sqrt g)$ and $\Sigma_L=\cD_L(\sqrt f)$).
\end{proof}

We will now state a result which allows to approximate the interface energy in the spirit of Lemma \ref{prop: interface energy approx}. Let for any integer $L>0$ and $u\in(0,1)$
$$\cC_L^{u}:=\Big \{ f ; \qquad \nu \big(f(\eta) \big) = 1, \quad \nu \big(f(\eta) \eta_{1} \eta_{2}) \geq 1-u \Big\} $$
then 
\begin{lem}\label{energy approxbis}
For any $u>0$
\begin{eqnarray}
\label{eq: SKubis}
\inf_{f \in \cC^u_{L}}  \cD_{L} \big( \sqrt{f} \big)  \geq \gS_{L} - \left( 4\sqrt{ \gr(1-\gr)}  + \gS_{L} \right) u \, .
\end{eqnarray}
\end{lem}
\begin{proof}
The proof is analogous to the proof of Lemma \ref{prop: interface energy approx}, therefore we sketch only the main points.
Let $f\in\cC_L^u$ and set $g(\eta)=\eta_1 f(\eta)/c(u)$ with $c(u)=\nu(f\eta_1)\geq 1-u$. Then $g$ belongs to $\cC_L$ and it holds
\begin{equation}
\cD_L(\sqrt f)\geq c(u)\cD_L(\sqrt g)-4\sqrt {\rho(1-\rho)}u\geq (1-u)\Sigma_L-4\sqrt {\rho(1-\rho)}u
\end{equation}
where the first inequality is obtained following the same lines as in Lemma \ref{prop: interface energy approx}. \end{proof}
By using the above definitions and results we are now ready to prove Theorem \ref{teo:phasetrans} (ii).
%\begin{lem}Let $\Sigma$ be defined as in \eqref{defS2}. Then for FA-1f with one empty boundary and for East the same inequalities as in Lemma \eqref{lowereasy} and Lemma \ref{upperdifficult} hold.\label{chilosa}
%\end{lem}
%Then 
\begin{proof} [Theorem \ref{teo:phasetrans} (ii) for East and for FA-1f model with one empty boundary]
It is enough to prove that if  $\Sigma$ is defined as in \eqref{defS2} then for FA-1f with one empty boundary and for East the same inequalities as in Lemma \ref{lowereasy} and Lemma \ref{upperdifficult} with $\alpha_0=-\frac{4\sqrt{\rho(1-\rho)}+\Sigma}{\bbA}$.

\smallskip
 In order to prove the lower bound one introduces the event $\cO$ which is verified iff $\pi_T(\eta_1\dots\eta_{N-K})=1$.
Along the same lines used to obtain  \eqref{eq: free energy lower int} on can verify that
\begin{equation*}
\lim_{N\to\infty}\lim_{T\to\infty}\frac{1}{T}\log\bk{\exp\left(\frac{\alpha}{N}\cA(T)\right)}\geq-\Sigma_{K}+\frac{\alpha}{N}(K+1) \, .
\end{equation*}
The lower bound follows again by taking $N$ to infinity and then $K$ to infinity. 

\smallskip
In order to prove the upper bound we  follow exactly the same lines as in Lemma \ref{upperdifficult}, the only difference being that in inequality \eqref{IIquo2} now the event $\cR_i$ is substituted by an event $\widetilde\cR_i$ which verified iff
$$\pi_T(\eta_{(i-1)K+1}\eta_{(i-1)K+2})\geq 1-\delta.$$
Then Donsker-Varhadan principle \eqref{DV2} yields
\begin{equation}
\label{IIquo3new} 
 \lim_{T\to\infty}  \frac{1}{T}\log  \bk{\mathds{1}_{\widetilde\cR_{i}}}_{\alpha/N} 
=-\exp(\alpha/N)\inf_{f:\nu(f)=1,
\nu(f(\eta_{(i-1)K+1}\eta_{(i-1)K+2})\geq 1-\delta}\cD_{N}(\sqrt f)
\end{equation}
Given
 $f$  on $\Omega_N$ we define a new function function $g$ on $\Omega_{N-(i-1)K}$ as follows
 $$g(\eta_1,\dots,\eta_{N-(i-1)K}):=\sum_{\eta_1'\dots\eta'_{(i-1)K}}\prod_{j=1}^{(i-1)K}\rho^{\eta_j'}(1-\rho)^{1-\eta'_j}f(\eta_1'\dots\eta'_{(i-1)K},\eta_1,\dots,\eta_{N-(i-1)K}).$$
 Then one can verify that
 it holds $\cD_N(\sqrt f)\geq \cD_{K}(\sqrt g)$ and
 if $f$ satisfies $\nu(f)=1$ and
$\nu(f\eta_{(i-1)K+1}\eta_{(i-1)K+2})\geq 1-\delta$ then $g$ satisfies $\nu(g)=1$ and $\nu (g \eta_1\eta_2)\geq 1-\delta$, therefore $g$ belongs to $\cC^{\delta}_K$.
Therefore from \eqref{IIquo3new} 
we get
\begin{eqnarray}
&& \lim_{T\to\infty}  \frac{1}{T}\log  \bk{\mathds{1}_{\widetilde\cR_{i}}}_{\alpha/N}\leq -\Sigma_{N-(i-1)K}+(4\sqrt{\rho(1-\rho)}+\Sigma_{N-(i-1)K})\delta+\frac{|\alpha|}{N}\gS_{1}\nonumber\\
&& \leq  -\Sigma+(4\sqrt{\rho(1-\rho)}+\Sigma_{K})\delta+\frac{|\alpha|}{N}\gS_{1}
\end{eqnarray}
%& & \qquad =  -\exp(\alpha/N)\inf_{g \in\cC_{N-(i-1)K}^{\delta}}\cD_{K}(\sqrt g)\leq
%-\Sigma_{K}+(8\sqrt{\rho(1-\rho)}+\Sigma_{\frac{K}{2}})\delta+\frac{|\alpha|}{N}\gS_{1}\nonumber\\ 
%& & \qquad  \leq -\Sigma+(8\sqrt{\rho(1-\rho)}+\Sigma_{\frac{K}{2}})\delta+\frac{|\alpha|}{N}\gS_{1} \, ,
 %and  furthewe have applied Lemma \ref{energy approxbis}.
\end{proof}

\begin{rem} 
\label{remdiff}
Fix $\rho\in(0,1)$ and let $\Sigma_1$ and $\Sigma_2$ be the interface energies defined in formulas \eqref{defS1} and \eqref{defS2}  by using the Dirichlet form of FA-1f with two empty boundaries and one empty boundary respectively. Then it can be easily verified that 
$\Sigma_1=2\Sigma_2$.
\end{rem}
%%%proof let f be the function on which $\Sigma_L$ is attained then we set $g=f(\eta_{L+2},\dots\eta_{2L})f(\eta_{L+1}\dots \eta_1)$ and we obtain
%that $g$ belongs to $\cC_{L+(L-2)+2}$ and $\cD(\sqrt g)=2\cD(\sqrt f)$. Also if $g$ is a function on $C_{L+L'+2}$
% we can verify that the dirichelt form  decouples on two pieces and on each one one cannot do best than the infimum over all functions
% belonging to $C_L$ and $C_L'$....

\subsection{Proof of  Theorem \ref{teo:condmes} (ii) and a stronger result}

\label{condmes2}

We start by establishing a result which is stronger than Theorem \ref{teo:condmes} (ii).
%
%
%\begin{lem}\label{easylarge}
%\end{lem}
\begin{lem} 
\label{usolemma2}
Consider East or FA-1f model in one dimension.
Let $K_N=\sqrt N$ and $\delta_N=N^{-1/8}$ and 
%be two sequences such that $$\lim_{N\to\infty}\delta_N=\lim_{N\to\infty}1/K_N=\lim_{N\to\infty}K_N/(N\delta_N^2)=\lim_{N\to\infty}\frac{\log (\delta_N)}{K_N\delta_N^2}=0.$$
%Note that there are subsequences which verify the above conditions, e.g. the choice $K_N=\sqrt N$ and $\delta_N=N^{-1/8}$.
for each $N$ let the activity-density labels be defined with $K=K_N$ and $\epsilon_N=\bbA \delta_N/2$.
Then  if $\alpha<\alpha_0$ with $\alpha_0$ defined in Lemma \ref{upperdifficult} it holds
\begin {eqnarray}
\label{eq: densite haute bis}
\lim_{N\to\infty}\lim_{T \to \infty}
\mu_{\alpha,T}^{N} ( \cW_{1,1-\delta_N})= 1 \, .
\end{eqnarray}
\end{lem}

\begin{proof}

Fix $ N$ sufficiently large such that $K_N\geq\bar K(\frac{\delta_N}{2}, \bbA \frac{\delta_N}{2})$
with $\bar K$ defined in Lemma \ref{prop: bad blocksnew}. Since $\mu^N_{\alpha,T}$ is a measure from inequality \eqref{ineevents3} it holds
\begin{equation}\label{mone}
1-\mu_{\alpha,T}^N( \cW_{0,\delta_N a_N})- \sum_{\ell=\lceil \delta_N (1-a_N) \frac{N}{K} \rceil }^{N/K} \mu_{\alpha,T}^N( \cV_{-1,\ell}\cap \cW_{0,\delta_N a_N}^c)
\leq \mu_{\alpha,T}^N( \cW_{1,1-\delta_N })\leq 1 \, ,
\end{equation}
where we set $a_N=N^{-1/16}$.
%with $\alpha_N$ defined in \eqref{alphaN} from the definition of $\mu_{\alpha,T}^N$
We get from Lemma \ref{lowereasy}
\begin{equation}
\lim_{N\to\infty}\lim_{T\to\infty}\frac{1}{T}\log\mu_{\alpha,T}^N ( \cW_{0,a_N \delta_N})\leq \Sigma+\lim_{N\to\infty}\lim_{T\to\infty}\frac{1}{T}\log(G_1)\leq \Sigma -C\lim_{N\to\infty}a_N^2\delta_N^2\frac{N}{K_N} \, ,
\label{mone1}
\end{equation}
where $G_1$ has  been defined in \eqref{GIbis} and we used inequality \eqref{miserve}.

Choose $\ga$ such that $\alpha<  - 2 \frac{8\sqrt{\rho(1-\rho)}+\Sigma}{\bbA} -\delta_N$.
For  $G_{3,\ell}$ defined in \eqref{GIII} with $\ell \geq \lceil \delta_N (1-a_N) \frac{N}{K} \rceil$, one has by using  \eqref{tis++} for any $N$ large enough 
\begin{eqnarray}
\lim_{T\to\infty}\frac{1}{T}\log\mu_{\alpha,T}^N( \cV_{-1,\ell}\cap \cW_{0,a_N \delta_N}^c)
\leq
\Sigma + \lim_{T\to\infty}\frac{1}{T}\log (G_{3,\ell})
\leq - \frac{\bbA}{2}  \delta_N  + C'' \frac{K_N}{N}  \, .\nonumber\\
\label{mone2}
\end{eqnarray}
Thanks to  the hypothesis  on $K_N$, $a_N$, $\delta_N$, it holds $\lim_{N\to\infty}\delta_N\frac{N}{K_N}=\infty$ and $\lim_{N\to\infty}a_N^2\delta_N^2\frac{N}{K_N}=\infty$ 
thus by using \eqref{mone}, \eqref{mone1} and \eqref{mone2} the proof is concluded.
\end{proof}

\begin{proof}[Proof of Theorem \ref{teo:condmes} (ii)]
The event $ \cW_{1,1-\delta}$ implies 
$$
\pi_T( \sum_i \eta_i)\geq N(1-\delta), \quad \text{and} \quad  \pi_T( \sum_i c_i(\eta))\leq \delta N .
$$ 
Thus the result follows by taking $\gamma_N=\delta_N$ with $\delta_N$ defined in Lemma \ref{usolemma2}.
\end{proof}

%{\bf $\alpha_0$ et $\alpha_1$ diverge si $\rho\to 1$ or $\rho\to 0$...}

\subsection{FA-1f in $d>1$: proof of Theorem \ref{poor}}
\label{secpoor}

\begin{proof}[Proof of Theorem \ref{poor}]
Recall that we have extended definition \ref{defW} to the higher dimensional case by substituting $N/K$ with $(N/K)^d$.
We start from inequality
\begin{equation}
\label{last}
\mathds{1}_{\cW_{0,\delta/2}}+\mathds{1}_{\cW_{1,1-\delta}}+\mathds{1}_{\cW_{-1,\delta/2}}\geq 1
\end{equation}
which leads to
\begin{equation}
\label{uffff}
1-\mu_{\alpha,T}^{N,d}(\mathds{1}_{\cW_{0,\delta/2}})-\mu_{\alpha,T}^{N,d}(\mathds{1}_{\cW_{-1,\delta/2}})\leq\mu_{\alpha,T}^{N,d} (\mathds{1}_{\cW_{1,1-\delta}})\leq 1.
\end{equation}
Lemma \ref{nonhopiu} guarantees that 
\begin{equation}
\lim_{T\to\infty}\lim_{N\to\infty}\mu_{\alpha,T}^{N,d}(\mathds{1}_{\cW_{0,\delta/2}})=0\label{ufffff}\end{equation}
Then we notice that  if at time zero the configuration is completely filled and the clocks on all the $N^c$ sites which are unconstrained do not ring up to time $T$ then $\cA(T)=0$. Therefore 
\begin{equation}
\left \bra \exp\left(\frac{\alpha}{N^{d-c}} \cA(T)\right) \right\ket
\geq\exp(-N^c T)\rho^{N^d}
\label{uffffff}
\end{equation}
and by inserting this bound in the denominator of the definition \eqref{mesureh} of the measure $\mu_{\alpha,T}^{N,d}$ and using
 inequality \eqref{usoH} (extended to dimension $d$), we get
$$\lim_{T\to\infty}\lim_{N\to\infty}\mu_{\alpha,T}^{N,d}(\mathds{1}_{\cW_{-1,\delta/2}})=0$$
provided $\alpha<-2/(\delta\bbA)$ and the result is proven by inserting \eqref{ufffff} and \eqref{uffffff} into \eqref{uffff}.
\end{proof}

\section{Large deviations for a reduced activity}
\label{new}

As we already recalled in the introduction, in absence of constraints (i.e. for the model defined by \eqref{gene} and \eqref{ci} with $r_i(\eta)\equiv 1$), the probability of observing a large deviation from the mean value scales as 
$$
\lim_{N\to\infty}\lim_{t\to\infty} \; \frac{1}{Nt} \log \left \bra  \frac{\cA(t)}{Nt} \simeq a \right \ket  =-f(a) \, ,
$$ 
with $0< f(a) <\infty$ for $a\neq 2\rho(1-\rho)$. 
%This can be immediately verified by noticing that in this case $\gp^{(N)}(\lambda)=N\gp^{(1)}(\lambda)$ and $\gp^{(1)}(\lambda)$ is strictly convex in $\lambda$.
In this section we will prove Theorems \ref{fluctu1} and \ref{fluctu2} which establish that a different scaling occurs for the large deviations of the activity below the mean value in the presence of constraints.

 \begin{proof}[Proof of Theorem \ref{fluctu1}]
 Let us start by the upper bound. 
 Let $\alpha_0$ be defined as in Theorem \ref{teo:phasetrans}, then
 \begin{equation}
% <\frac{\cA(t)}{Nt}\in [c\mathbb A-\delta, c\mathbb A+\delta]>\leq \exp(-\alpha_0 (c\bbA+\delta)t)<\exp\left(\frac{\alpha_0}{N}\cA(t)\right)>
\left \bra \frac{\cA(t)}{Nt}\sim u \bbA  \right \ket 
\leq 
\exp(-\alpha_0 u \bbA t) \; \left \bra \exp\left(\frac{\alpha_0}{N}\cA(t)\right) \right \ket \, .
 \end{equation}
 Therefore by taking the $\limsup_{N\to\infty}$ on the right and left hand side and using Theorem \ref{teo:phasetrans} (i), we obtain
 and the desired upper bound.
 
For the lower bound, we consider  FA-1f with two empty boundaries (the proof in the other cases is analogous). 
The contribution to  $\cA(t)/(Nt)$ can be decomposed into the contributions coming respectively from the configuration changes during the
time interval $[0,u t]$ and $[u t,t]$. 
%Then for $N$ sufficiently large in order that $1/\sqrt N<\delta/2$ for any $\alpha\in[\alpha_0,0)$ we get
%$$<\frac{\cA(t)}{Nt}\in [c\mathbb A-\delta, c\mathbb A+\delta]>\geq <\exp(\alpha\cA_1/N)\exp(-\alpha\cA_1/N)
%\mathds_{\cA_1\in[Nt(c\bbA-\delta/2),Nt(c\bbA+\delta/2)}\exp(\alpha\frac{\cA_2(t)}{Nt}\in [c\mathbb A-\delta, c\mathbb A+\delta]>
With probability (w.r.t. the mean over the initial distribution $\nu$ and the evolution of the process) which
goes to one as $t$ goes to infinity the first contribution goes to $u\bbA$ and, thanks to the reversibility of $\nu$, the distribution at time $u t$ is still $\nu$.
Then we can impose that during the second time interval $[u t,t]$ the contribution to  $\cA(t)/(Nt)$ is at most $2/\sqrt N$  by requiring that $\eta_{\sqrt N+1	}(s)\dots\eta_{N-\sqrt N}(s)=1$ for any time $s$ in $[u t, t]$. 
Thus 
$$
\liminf_{N\to\infty}\lim_{t\to\infty} \; \frac{1}{t} \log \left \bra \frac{\cA(t)}{Nt}\sim u\bbA \right \ket
\geq 
\liminf_{N\to\infty}\lim_{t\to\infty} \; \frac{1}{t}  \log \left \bra \mathds{1}_{\cB} \right \ket \, ,
$$
% the probability of the event $\cA/(Nt)\sim c\bbA$ is lower bounded by
where $\cB$ is the 
%the probability that starting from $\nu$ in a time interval $[0,(1-c)t]$ the activity is smaller than $2\sqrt{N}t$. In turn the latter event is lower bounded by the probability of the
event which is verified iff starting from $\nu$ it holds $\pi_{(1-u)t}(\eta_{\sqrt N+1} \dots\eta_{N-\sqrt N})=1$.
As we did for the event $\cO$ in \eqref{whynot} we get here
$$
\liminf_{N\to\infty} \lim_{t\to\infty}\frac{1}{(1-u)t}\log  \bra \mathds{1}_{\cB} \ket \geq -\Sigma \, ,
$$
and the proof is completed.
% $$\gp(\alpha)\leq\bar\gp(\alpha)$$ where
% 
% \begin{equation}
%  \bar\gp(\alpha)=
%  \begin{cases}
% \bbA\alpha_0 & \text{ if }\,\alpha<\alpha_0\\
% & \text{ if }\, \alpha\geq\alpha_0
%  \end{cases}
%\end{equation}
% 
%Note that the Legendre transform $\pi(a)$ of $\bar \gp(\alpha)$ is 
%\begin{equation}
% \pi(a)=
%  \begin{cases}
%\alpha_0(a-\bbA) & \text{ if }\, a\in[0,\bbA]\\
% \infty & \text{otherwise}
%  \end{cases}
%\end{equation}
%
% Therefore 
 \end{proof}

 \begin{proof}[Proof of Theorem \ref{fluctu2}]
 The proof of the upper bound follows along the same line as for Theorem \ref{fluctu1}
 starting now from the inequality
 \begin{equation}
% <\frac{\cA(t)}{Nt}\in [c\mathbb A-\delta, c\mathbb A+\delta]>\leq \exp(-\alpha_0 (c\bbA+\delta)t)<\exp\left(\frac{\alpha_0}{N}\cA(t)\right)>
\left \bra \frac{\cA(t)}{N^d t}\sim u \bbA \right \ket
\leq \exp(-\alpha_0 u \, \bbA tN^c)  \left \bra \exp\left(\frac{\alpha_0}{N^{d-c}}\cA(t)\right) \right \ket \, .
 \end{equation}

 The lower bound is derived in the same way by freezing the configuration during the time interval $[ut,t]$. The probability of realizing this event can be bounded from below by the probability that the $N^c$ unconstrained sites (those which are in contact with the empty boundary) remain frozen equal to 1. 
\end{proof}

\subsection*{Acknowledgements}
We wish to thank  L.Bertini, B.Derrida, V.Lecomte, F. van Wijland and L.Zambotti for very useful discussions.
We acknowledge the support of the French Ministry of Education
through the ANR BLAN07-2184264. C.T acknowledges the support of ANR DynHet and of the ERC Advanced Grant  PTRELSS 228032.


\begin{thebibliography}{10}


\bibitem[AD]{AD} D.Aldous, P.Diaconis, 
{\sl The asymmetric one-dimensional constrained Ising  model: rigorous results} J.Stat.Phys {\bf 107}, 945  (2002)

\bibitem[BLT]{BLT} 
T. Bodineau, V. Lecomte, C. Toninelli, work in progress.


\bibitem[CMRT]{CMRT}
N. Cancrini, F. Martinelli, C. Roberto, C. Toninelli,
{\sl Kinetically constrained spin models}, Probability Theory and Related Fields \textbf{140}, 459--504,  (2008)

\bibitem[CMRT1]{CMRT1}
N. Cancrini, F. Martinelli, R. Schonmann, C. Toninelli,
{\sl Facilitated Oriented Spin Models: Some Non Equilibrium Results}, J.Stat.Phys {\bf 138}, 1109-1123  (2010)

\bibitem[CMRT2]{CMRT2}
N.Cancrini, F.Martinelli, C.Roberto, C.Toninelli, 
{\sl Facilitated spin models: recent and new results} in Methods of contemporary mathematical statistical physics, Lecture Notes in Mathematics, p.307-339 R.Kotecky Ed., Springer (2009);

\bibitem[DZ]{DZ} 
A.Dembo, O.Zeitouni, {\sl Large deviations techniques and applications}, series Stochastic modelling and applied probability, vol.38, Springer (1998)

\bibitem[FA1]{FA1}
    G.H. Fredrickson, H.C. Andersen,
     \sl{Kinetic Ising model of the glass transition}, Phys. Rev. Lett. {\bf 53}, 1244--1247,   (1984)

\bibitem[FA2]{FA2}
    G.H. Fredrickson, H.C. Andersen, \sl{Facilitated kinetic Ising models and the glass transition}
    J. Chem. Phys. {\bf 83}, 5822--5831 (1985)

\bibitem[GJLPDW1]{GJLPDW1}
J.P. Garrahan, R.L. Jack, V. Lecomte, E. Pitard, K. van Duijvendijk, F. van Wijland, {\sl First-order dynamical phase transition in models of glasses: an approach based on ensembles of histories },  J. Phys. A {\bf 42} 075007 (2009) 

\bibitem[GJLPDW2]{GJLPDW2}
J.P. Garrahan, R.L. Jack, V. Lecomte, E. Pitard, K. van Duijvendijk, F. van Wijland, {\sl Dynamic first-order transition in kinetically constrained models of glasses}  Phys.Rev.Lett.  {\bf 98}, 195702 (2007) 

\bibitem[GST]{GST} J.P. Garrahan, P. Sollich, C.Toninelli, {\sl Kinetically Constrained Models}, preprint arXiv:1009.6113

\bibitem[JGC]{JGC} R. Jack, J.P.  Garrahan, D. Chandler, 
{\sl Space-time thermodynamics and subsystem observables in kinetically constrained models of glassy materials} J.Chem.Phys. {\bf 125}, 184509, (2006)


\bibitem[JE]{JE}
J. J\"{a}ckle, S. Eisinger,
\textit{A hierarchically constrained kinetic Ising model}, Z. Phys. B: Condens. Matter, \textbf{84}, 115--124,  (1991)

\bibitem[KL]{KL}
C. Kipnis, C. Landim, {\it Scaling limits of interacting particle systems}, Grundlehren der Mathematischen Wissenschaften  {\bf 320} Springer (1999)

\bibitem[MGC]{MGC} 
M.Merolle, J.P.Garrahan, D.Chandler {\sl Space-time thermodynamics of the glass transition}, Proc.Matl.Acad.Sci. USA, {\bf 102}, 10837--10840 (2005)


 



%
%\bibitem[BDGJL]{BDGJL} 
%L. Bertini, A. De Sole, D. Gabrielli, G. Jona Lasinio, C. Landim,
%{\it Stochastic interacting particle systems out of equilibrium}, J. Stat. Mech. P07014. (2007) 

%
%\bibitem[BD]{BD2007}
%T.  Bodineau, B. Derrida,
%{\it   Cumulants and large deviations of the current in non-equilibrium steady states},
%C.R. Physique    {\bf 8}  540-555  (2007)

%
%\bibitem[D]{D} B. Derrida, {\it  Non-equilibrium steady states: fluctuations and large deviations of the density and of the current},  J. Stat. Mech. P07023 (2007)

\bibitem[RS]{RS} F. Ritort, P. Sollich, {\sl Glassy dynamics of kinetically constraint models}, Advances in Physics, {\bf 52},  219-342, (2003) 


\bibitem[SS]{SS}
R. Schonmann, S. Shlosman, 
{\em  Complete analyticity for $2$D Ising completed}, Comm. Math. Phys. {\bf 170}, no. 2, 453--482  (1995)
\end{thebibliography}
\end{document}